\documentclass[conference]{IEEEtran}
\usepackage{epsfig,endnotes}
\usepackage{multirow}
\usepackage{graphicx}
\usepackage{subfig}
\usepackage{grffile}
\usepackage{thmtools, thm-restate}

\newcounter{subcopyrightbox@save}
\usepackage{caption}
\usepackage{color, url}
\usepackage{xspace} 
\usepackage{mathrsfs}
\usepackage{amssymb}
\usepackage{amsmath}
\usepackage{epstopdf}
\usepackage{balance}
\usepackage{bm}

\usepackage{textalpha} % <--- Greek letters in text

\usepackage{algorithm}
\usepackage{algorithmic}

\usepackage{lipsum}
\usepackage{xcolor}
\newcommand\hl[1]{%
  \bgroup
  \hskip0pt\color{red!80!black}%
  #1%
  \egroup
}

\usepackage{fancyhdr}

\fancypagestyle{plain}{
  \fancyhf{}
  \fancyhead[C]{Conference on \LaTeX}     %% C or L or R.
}
\usepackage{eso-pic}

\newcommand{\argmin}{\operatornamewithlimits{argmin}}

\newtheorem{definition}{Definition}
\newtheorem{theorem}{Theorem}

\newcommand{\myparatight}[1]{\smallskip\noindent{\bf {#1}:}~}
\newenvironment{packeditemize}{\begin{list}{$\bullet$}{\setlength{\itemsep}{0.2pt}\addtolength{\labelwidth}{10pt}\setlength{\leftmargin}{\labelwidth}\setlength{\listparindent}{\parindent}\setlength{\parsep}{1pt}\setlength{\topsep}{0pt}}}{\end{list}}
\newcommand{\alan}[1]{{{#1}}}

\graphicspath{ {../fig/} }

\begin{document}
\AddToShipoutPictureBG*{%
  \AtPageUpperLeft{%
    \setlength\unitlength{1in}%
    \hspace*{\dimexpr0.5\paperwidth\relax}%%  change \dimexpr0.5\paperwidth\relax appropriately
    \makebox(0,-0.75)[c]{In the 39th IEEE Symposium on Security and Privacy, May 2018, San Francisco, CA}%
}}

%\AddToShipoutPictureBG*{%
%  \AtPageLowerLeft{%
%    \setlength\unitlength{1in}%
%    \hspace*{\dimexpr0.5\paperwidth\relax}%%  change \dimexpr0.5\paperwidth\relax appropriately
%    \makebox(0,0.75)[c]{\Large Notice}%
%}}

% Copyright
%\setcopyright{acmcopyright}
%\setcopyright{acmlicensed}
%\setcopyright{rightsretained}
%\setcopyright{usgov}
%\setcopyright{usgovmixed}
%\setcopyright{cagov}
%\setcopyright{cagovmixed}

%% DOI
%\doi{10.475/123_4}
%
%% ISBN
%\isbn{123-4567-24-567/08/06}
%
%%Conference
%\conferenceinfo{CCS'16}{October 24-28, 2016, Vienna, Austria}
%
%\acmPrice{\$15.00}
%
%
% --- Author Metadata here ---
%\conferenceinfo{WOODSTOCK}{'97 El Paso, Texas USA}
%\CopyrightYear{2007} % Allows default copyright year (20XX) to be over-ridden - IF NEED BE.
%\crdata{0-12345-67-8/90/01}  % Allows default copyright data (0-89791-88-6/97/05) to be over-ridden - IF NEED BE.
% --- End of Author Metadata ---
%\title{\bf \fontsize{15pt}{1em}\selectfont PIANO: Proximity-based Authentication for Internet-of-Things Devices}
\title{Stealing Hyperparameters in Machine Learning}

%\numberofauthors{1} 

% \author{
%     \IEEEauthorblockN{Binghui Wang, Neil Zhenqiang Gong}
%     \IEEEauthorblockA{ECE Department, Iowa State University}
%     \IEEEauthorblockA{\{binghuiw, neilgong\}@iastate.edu}
% }

\author{
    \IEEEauthorblockN{Binghui Wang, Neil Zhenqiang Gong}
    \IEEEauthorblockA{ECE Department, Duke University}
    \IEEEauthorblockA{\{binghui.wang, neil.gong\}@duke.edu}
}
%\IEEEoverridecommandlockouts
%\makeatletter\def\@IEEEpubidpullup{9\baselineskip}\makeatother
%\IEEEpubid{\parbox{\columnwidth}{Permission to freely reproduce all or part
%    of this paper for noncommercial purposes is granted provided that
%    copies bear this notice and the full citation on the first
%    page. Reproduction for commercial purposes is strictly prohibited
%    without the prior written consent of the Internet Society, the
%    first-named author (for reproduction of an entire paper only), and
%    the author's employer if the paper was prepared within the scope
%    of employment.  \\
%    NDSS '17, 21-24 February 2016, San Diego, CA, USA\\
%    Copyright 2016 Internet Society, ISBN 1-891562-41-X\\
%    http://dx.doi.org/10.14722/ndss.2016.23xxx
%}
%\hspace{\columnsep}\makebox[\columnwidth]{}}

\maketitle
\begin{abstract}
%Machine-learning-as-a-service (MLaaS) is an emerging technology to help a user, who has limited computing power and machine learning expertise, to learn a machine learning model through the cloud computing paradigm.   The user is often charged according to the amount of computation the MLaaS platform performs to learn the model. 
Hyperparameters are critical in machine learning, as different hyperparameters often result in models with significantly different performance. 
Hyperparameters may be deemed confidential because of their commercial value and the confidentiality of the proprietary algorithms that the learner uses to learn them. % often learnt through computationally expensive and proprietary cross-validation algorithms. %Due to intellectural 
 In this work, we propose attacks on stealing the hyperparameters that are learnt %and used to train a machine learning model 
by a learner. 
We call our attacks \emph{hyperparameter stealing attacks}. Our attacks are applicable to a variety of popular machine learning algorithms such as ridge regression, logistic regression, support vector machine, and neural network. We evaluate the effectiveness of our attacks both theoretically and empirically. For instance, we evaluate our attacks on Amazon Machine Learning. Our results demonstrate that our attacks can accurately steal hyperparameters. 
%Our evaluation results on several real-world datasets demonstrate that our attacks are accurate and are applicable to a variety of popular machine learning algorithms such as ridge regression, logistic regression, support vector machines, and neural networks. %Our attacks compromise the
%learner's intellectual property and algorithm confidentiality. 
%Moreover, machine-learning-as-a-service (MLaaS) is an emerging technology to help a user, who has limited computing power and machine learning expertise, to learn a machine learning model through the cloud computing paradigm. 
%We demonstrate that a user can use our hyperparameter stealing attacks to learn a machine learning model through MLaaS with less economic costs, while not sacrificing the model's testing performance. 
We also study countermeasures. Our results highlight the need for new defenses against our hyperparameter stealing attacks for certain machine learning algorithms.

\end{abstract}
\section{Introduction}
\label{intro}

%\textbf{ML importance}

%Introduce hyperparameter
Many popular supervised machine learning (ML) algorithms--such as ridge regression (RR)~\cite{hoerl1970ridge} , logistic regression (LR)~\cite{hosmer2013applied}, support vector machine (SVM)~\cite{cortes1995support}, and neural network (NN)~\cite{goodfellow2016deep}--learn the parameters in a model via minimizing an \emph{objective function}, which is often in the form of \emph{loss function + $\lambda\ \times$ regularization term}.  Loss function characterizes how well the model performs over the training dataset, regularization term is used to prevent \emph{overfitting}~\cite{bishop2006pattern}, and $\lambda$ balances between the two. Conventionally, $\lambda$ is called \emph{hyperparameter}. Note that there could be multiple hyperparameters if the ML algorithm adopts multiple regularization terms.  Different ML algorithms use different loss functions and/or regularization terms. %in their objective functions.

Hyperparameters are critical for ML algorithms. For the same training dataset, with different hyperparameters, an ML algorithm might learn models that have significantly different performance on the testing dataset, e.g., see our experimental results about the impact of hyperparameters on different ML classifiers in Figure~\ref{accvshyper} in the Appendix. Moreover, hyperparameters are often learnt through a computationally expensive cross-validation process, which may be implemented by proprietary algorithms that could vary across learners. Therefore, hyperparameters may be deemed confidential. % because of their commercial value and confidentiality of the proprietary cross-validation algorithms. 

\myparatight{Our work} In this work, we formulate the research problem of stealing hyperparameters in machine learning, and we provide the first systematic study on hyperparameter stealing attacks as well as their defenses.    

{\bf Hyperparameter stealing attacks.} We adopt a threat model in which an attacker knows the training dataset, the ML algorithm (characterized by an objective function), and (optionally) the learnt model parameters.  
Our threat model is motivated by the emerging machine-learning-as-a-service (MLaaS) cloud platforms, e.g., Amazon Machine Learning~\cite{amazon} and Microsoft Azure Machine Learning~\cite{ms}, in which the attacker could be a user of an MLaaS platform. When the model parameters are  unknown, the attacker can use model parameter stealing attacks~\cite{tramer2016stealing} to learn them. 
As a first step towards studying the security of hyperparameters, 
\alan{\emph{we focus on hyperparameters that are used to balance between the loss function and the regularization terms in the objective function}. Many popular ML algorithms--such as ridge regression, logistic regression, and SVM (please refer to Table~\ref{sum_ml} for more ML algorithms)--rely on such hyperparameters.  
%ML algorithms whose objective functions are \emph{almost everywhere differentiable}, which include a variety of popular ML algorithms. Moreover, we focus on the hyperparameters that are used to balance between loss function and regularization terms in the objective function. 
It would be an interesting future work to study the security of hyperparameters for other ML algorithms, e.g., the hyperparameter $K$ for $K$NN, as well as network architecture, dropout rate~\cite{dropout}, and mini-batch size for deep neural networks.} However, as we will demonstrate in our experiments, an attacker (e.g., user of an MLaaS platform) can already significantly financially benefit from stealing the hyperparameters in the objective function.

We make a key observation that the model parameters learnt by an ML algorithm are often \emph{minima} of the corresponding objective function. Roughly speaking, a data point is a minimum of an objective function if the objective function has larger values at the nearby data points. This implies that the \emph{gradient} of the objective function at the model parameters is close to $\mathbf{0}$ ({$\mathbf{0}$ is a vector whose entries are all 0).
Our attacks are based on this key observation. First, we propose a general attack framework to steal hyperparameters.  Specifically, in our framework, we compute the gradient of the objective function at the model parameters and set it to $\mathbf{0}$, which gives us a \emph{system of linear equations} about the hyperparameters. This linear system is \emph{overdetermined} since the number of equations (i.e., the number of model parameters) is usually larger than the number of unknown variables (i.e., hyperparameters). Therefore, we leverage the \emph{linear least square} method~\cite{montgomery2015introduction}, a widely used method to derive an approximate solution of an overdetermined system, to estimate the hyperparameters. Second, we demonstrate how we can apply our framework to steal hyperparameters  for a variety of ML algorithms.

{\bf Theoretical and empirical evaluations.} We evaluate our attacks both theoretically and empirically. Theoretically, we show that 1) when the learnt model parameters are an exact minimum of the objective function, our attacks can obtain the exact hyperparameters; and 2)  when the model parameters deviate from their closest minimum of the objective function with a small difference, then our estimation error is a linear function of the difference. Empirically, we evaluate the effectiveness of our attacks using six real-world datasets. Our results demonstrate that our attacks can accurately estimate the hyperparameters on all datasets for various ML algorithms. For instance, for various regression algorithms, the relative estimation errors are less than  $10^{-4}$ on the datasets.

Moreover, \alan{via simulations and evaluations on Amazon Machine Learning}, we show that a user can use our attacks to learn a model via MLaaS with much less economical costs, while not sacrificing the model's testing performance. Specifically, the user samples a small fraction of the training dataset, learns model parameters via MLaaS, steals the hyperparameters using our attacks, and re-learns model parameters using the entire training dataset and the stolen hyperparameters via MLaaS. %We will demon- strate evaluation results on this application scenario in Section 4.6.

{\bf Rounding as a defense.} One natural defense against our attacks is to round model parameters, so attackers obtain obfuscated model parameters. We note that  rounding was proposed to obfuscate confidence scores of model predictions to mitigate model inversion attacks~\cite{fredrikson2015model} and model stealing attacks~\cite{tramer2016stealing}.
We  evaluate the effectiveness of rounding using the six real-world datasets.  

First, our results show that rounding increases the relative estimation errors of our attacks, which is consistent with our theoretical evaluation. However, for some ML algorithms, our attacks are still effective. For instance, for LASSO (a popular regression algorithm)~\cite{tibshirani1996regression}, the relative estimation errors are still less than around $10^{-3}$ even if we round the model parameters to one decimal. Our results highlight the need to develop new countermeasures for hyperparameter stealing attacks. 
Second, since different ML algorithms use different regularization terms, one natural question is which regularization term has better security property. Our results demonstrate that  $L_2$ regularization term can more effectively defend against our attacks than  $L_1$ regularization term using rounding. This implies that an ML algorithm should use $L_2$ regularization in terms of security against hyperparameter stealing attacks.  Third,  we also compare different loss functions in terms of their security property, and we observe that \emph{cross entropy loss} and \emph{square hinge loss} can more effectively defend against our attacks than \emph{regular hinge loss} using rounding. The cross-entropy loss function is adopted by logistic regression~\cite{hosmer2013applied}, while square and regular hinge loss functions are adopted by support vector machine and its variants~\cite{hsu2003practical}. %Fourth, 
 
In summary, our contributions are as follows:
\begin{packeditemize}
\item We provide the first study on hyperparameter stealing attacks to machine learning. We propose a general attack framework to steal the hyperparameters in the objective functions. 
\item We evaluate our attacks both theoretically and empirically. Our empirical evaluations on several real-world datasets demonstrate that our attacks can accurately estimate hyperparameters for various ML algorithms. \alan{We also show the success of our attacks on Amazon Machine Learning}.
\item We evaluate rounding model parameters as a defense against our attacks. Our empirical evaluation results show that our attacks are still effective for certain ML algorithms, highlighting the need for new countermeasures. We also compare different regularization terms and different loss functions in terms of their security  against our attacks. 
\end{packeditemize}

\section{Related Work}
\label{related}

Existing attacks to ML can be roughly classified into four categories: \emph{poisoning attacks}, \emph{evasion attacks},
\emph{model inversion attacks}, and \emph{model extraction attacks}. Poisoning attacks and  evasion attacks are also called \emph{causative attacks} and \emph{exploratory attacks}~\cite{barreno2006can}, respectively. Our hyperparameter stealing attacks are orthogonal to these attacks. 

\myparatight{Poisoning attacks} In poisoning attacks, an attacker aims to pollute the training dataset such that the learner produces a bad classifier, which would mislabel malicious content or activities generated by the attacker at testing time. In particular, the attacker could insert new instances, edit existing instances, or remove existing instances in the training dataset~\cite{Papernot16SoK}. Existing studies have demonstrated poisoning attacks to
 worm signature generators~\cite{newsome2005polygraph,perdisci2006misleading,newsome2006paragraph}, 
 spam filters~\cite{Nelson08poisoningattackSpamfilter,Nelson09poisoningspamfilter},
anomaly detectors~\cite{rubinstein2009antidote,KloftpoisoningAnomalydetection}, 
SVMs~\cite{biggio2012poisoning}, face recognition methods~\cite{biggio2013poisoning},
 %feature selection methods~\cite{xiao2015feature}, 
 as well as recommender systems~\cite{poisoningattackRecSys16,YangRecSys17}.

\myparatight{Evasion attacks} In these attacks~\cite{nelson2010near,huang2011adversarial,biggio2013evasion,laskov2014practical,szegedy2013intriguing, goodfellow2014explaining,featurePositioning,spamEvasion16,Papernot16,XuEvasionMalware,CarliniUsenixSecurity16,sharif2016accessorize,liu2016delving,PracticalBlackBox17}, an attacker aims to inject carefully crafted noise into a testing instance (e.g., an email spam, a social spam, a malware, or a face image) such that the classifier predicts a different label for the instance.  %Most evasion attacks try to minimize the injected noise~\cite{barreno2006can,nelson2010near,huang2011adversarial,biggio2013evasion,laskov2014practical,sharif2016accessorize,szegedy2013intriguing, goodfellow2014explaining,Papernot16,nicolas2016transferability,liu2016delving,XuEvasionMalware}. Moreover, 
The injected noise often preserves the semantics of the original instance (e.g., a malware with injected noise preserves its malicious behavior)~\cite{laskov2014practical,XuEvasionMalware}, is human imperceptible~\cite{szegedy2013intriguing, goodfellow2014explaining,Papernot16,liu2016delving}, or is physically realizable~\cite{CarliniUsenixSecurity16,sharif2016accessorize}.  For instance, Xu et al.~\cite{XuEvasionMalware} proposed a general evasion attack to search for a malware variant that preserves the malicious behavior of the malware but is classified as benign by the classifier (e.g., PDFrate~\cite{SmutzPDFRate} or Hidost~\cite{SrndicHidost}).
Szegedy et al.~\cite{szegedy2013intriguing} observed that deep neural networks would misclassify an image after we inject a small amount of noise that is imperceptible to human. Sharif et al.~\cite{sharif2016accessorize} showed that an attacker can inject human-imperceptible noise to a face image to evade recognition or impersonate another individual, and the noise can be physically realized by the attacker wearing a pair of customized eyeglass frames. %which are designed to incorporate the noise and 3D printed. 
Moreover, evasion attacks can be even black-box, i.e., when the attacker does not know the classification model. This is because an adversarial example optimized for one model is highly likely to be effective for other models, which is known as \emph{transferability}~\cite{szegedy2013intriguing,liu2016delving,PracticalBlackBox17}. 

We note that Papernot et al.~\cite{Papernot16Distillation} proposed a distillation technique to defend against evasion attacks to deep neural networks. However, Carlini and Wagner~\cite{CarliniDistillation06} demonstrated that this distillation technique is not as secure to new evasion attacks as we thought. Cao and Gong~\cite{region} found that adversarial examples, especially those generated by the attacks proposed by Carlini and Wagner, are close to the classification boundary. Based on the observation, they proposed \emph{region-based classification}, which ensembles information in the neighborhoods around a testing example (normal or adversarial) to predict its label. Specifically, for a testing example, they sample some examples around the testing example in the input space and take a majority vote among the labels of the sampled examples as the label of the testing example. Such region-based classification significantly enhances the robustness of deep neural networks against various evasion attacks, without sacrificing classification accuracy on normal examples at all. In particular, an evasion attack needs to add much larger noise in order to construct adversarial examples that successfully evade region-based classifiers.

\myparatight{Model inversion attacks} In these attacks~\cite{fredrikson2014privacy,fredrikson2015model}, an attacker aims to leverage model predictions to compromise user privacy. %infer private information about training dataset. % of instances or entire instances in the training dataset via leveraging the classification model or its predictions.  
%and predicted confidence or labels to infer the training data instances. 
For instance, Fredrikson et al.~\cite{fredrikson2014privacy} demonstrated that model inversion attacks can infer an individual's private genotype information. 
%claim that adversarial access to an ML model is abused to learn sensitive genomic information about individuals. 
Furthermore, via considering confidence scores of model predictions~\cite{fredrikson2015model}, 
model inversion attacks can estimate whether a respondent in a lifestyle survey
admitted to cheating on its significant other  and can recover recognizable images of
people's faces given their name and access to the model. Several studies~\cite{membershipInfer,membershipInferCCS17,membershipLocation} demonstrated even stronger attacks,
 i.e., an attacker can infer whether a particular instance was in the training dataset or not.   

%They also empirically demonstrated that simple mechanisms
%including rounding confidence scores and considering sensitive features while
 %training decision trees can significantly reduce the effectiveness of the
%attacks.

%or predicted labels to recover recognizable images of people'€™s faces. Specifically, in a white-box setting, attackers are allowed to download a description of the model when accessing to a face recognition system. While in a black-box setting, attackers can produce a recognizable image of a person through making prediction queries against a model, given only API access to a face recognition system and the name of the person whose face is recognized by the model.
% but not actually download the model description. In a white-box setting, clients are allowed to download a description of the model.)

\myparatight{Model extraction attacks} These attacks aim to steal parameters of an ML model. Stealing model parameters compromises the intellectual property and algorithm confidentiality of the learner, and also enables an attacker to perform evasion attacks or model inversion attacks subsequently~\cite{tramer2016stealing}. 
%After obtaining these model parameters, an attacker can further perform evasion attacks or model inversion attacks. 
Lowd and Meek~\cite{lowd2005adversarial} presented efficient algorithms to steal model parameters of {linear classifiers} when the attacker can issue membership queries to the model through an API. 
%Vorobeychik and Li~\cite{vorobeychik2014optimal} generalized this work and demonstrated that if a classifier can be efficiently learned in polynomial time, then its model parameters can be efficiently stolen. Both studies
%They assume the model's API only predicts a class label of a queried instance.  
Tram\`{e}r et al.~\cite{tramer2016stealing} demonstrated that model parameters can be more accurately and efficiently extracted when the API also produces confidence scores for the class labels. 

\section{Background and Problem Definition}
\label{back_prob}

\subsection{Key Concepts in Machine Learning}
\label{sec:back}
We introduce some key concepts in machine learning (ML). 
In particular, we discuss supervised learning, which is the focus of this work. 
We will represent vectors and matrices as bold lowercase and uppercase symbols, respectively. 
For instance, $\mathbf{x}$ is a vector while $\mathbf{X}$ is a matrix. $x_i$ denotes the $i$th element of the vector $\mathbf{x}$.
We assume all vectors are column vectors in this paper. $\mathbf{x}^T$ (or $\mathbf{X}^T$) is the transpose of 
$\mathbf{x}$ (or $\mathbf{X}$).

\myparatight{Decision function} Supervised ML aims to learn a decision function $f$, which takes an instance as input and produces label of the instance. The instance is represented by a feature vector; the label can be continuous value (i.e., regression problem) or categorical value (i.e., classification problem). The decision function is characterized by certain parameters, which we call \emph{model parameters}. 
For instance, for a linear regression problem, the decision function is $f(\mathbf{x})$=$\mathbf{w}^T\mathbf{x}$, where $\mathbf{x}$ is the instance and $\mathbf{w}$ is the model parameters. For kernel regression problem, the decision function is $f(\mathbf{x})$=$\mathbf{w}^T\phi(\mathbf{x})$, where $\phi$ is a kernel mapping function that maps an instance to a point in a high-dimensional space.   Kernel methods are often used to make a linear model nonlinear. 

\myparatight{Learning model parameters in a decision function} An ML \emph{algorithm} is a computational procedure to determine the model parameters in a decision function from a given training dataset. Popular ML algorithms include ridge regression~\cite{hoerl1970ridge}, logistic regression~\cite{hosmer2013applied},  SVM~\cite{cortes1995support}, and neural network~\cite{goodfellow2016deep}.

\emph{Training dataset.} Suppose the learner is given $n$ instances $\mathbf{X} =\{ \mathbf{x}_i\}_{i=1}^n \in \mathbb{R}^{m \times n}$. 
For each instance $\mathbf{x}_i$, we have a label $y_i$, where $i=1,2,\cdots,n$. $y_i$ takes continuous value for regression problems and categorical value for classification problems. For convenience, we denote $\mathbf{y} = \{y_i\}_{i=1}^n$. $\mathbf{X}$ and $\mathbf{y}$ form the training dataset.

\emph{Objective function.} Many ML algorithms determine the model parameters via minimizing a certain \emph{objective function} over the training dataset. %The differences between different learning algorithms are that they adopt different objective functions and/or different optimization algorithms to solve the optimization problems. 
An objective function often has the following forms:
\begin{align}
%\label{ml_obj_lin}
%&\text{\bf For linear algorithms:} \nonumber \\
\text{\bf Non-kernel algorithms: } &\mathcal{L}(\mathbf{w}) = \text{L}(\mathbf{X}, \mathbf{y}, \mathbf{w}) +  \lambda \text{R}(\mathbf{w}) \nonumber\\
%& \qquad = \sum_{i=1}^n \text{Loss}(\mathbf{w}; \mathbf{x}_i, y_i) + \sum_{j} h_j \text{Reg}_j(\mathbf{w}),   
%&\text{\bf For kernel algorithms:} \nonumber \\
\text{\bf Kernel algorithms: } &\mathcal{L}(\mathbf{w}) = \text{L}(\phi(\mathbf{X}), \mathbf{y}, \mathbf{w}) +  \lambda \text{R}(\mathbf{w}), \nonumber
\end{align} 
where $\text{L}$ is called \emph{loss function}, $\text{R}$ is called \emph{regularization term}, $\phi$ is a kernel mapping function (i.e., $\phi(\mathbf{X})=\{\phi(\mathbf{x}_i)\}_{i=1}^n$), and $\lambda>0$ is called \emph{hyperparameter}, which is used to balance between the loss function and the regularization term. 
Non-kernel algorithms include linear algorithms and nonlinear neural network algorithms.
In ML theory, the regularization term is used to prevent overfitting. 
Popular regularization terms include $L_1$ regularization (i.e., $\text{R}(\mathbf{w})$=$||\mathbf{w}||_1=\sum_{i} |{w}_i|)$ and $L_2$ regularization (i.e., $\text{R}(\mathbf{w})$=$||\mathbf{w}||_2^2=\sum_i {w}_i^2)$. 
\emph{We note that, some ML algorithms use more than one regularization terms and thus have multiple hyperparameters. Although we focus on ML algorithms with one hyperparameter in the main text of this paper for conciseness, our attacks are applicable to more than one hyperparameter and we show an example in Appendix~\ref{app:enet}.} 
%We note that neural network are a category of nonlinear ML algorithms.
%Recent deep neural network also use \emph{dropout}~\cite{} as a regularization instead of conventional $L_1$ or $L_2$ regularization terms. 
%For such dropout regularization, 
%the dropout probability for each neuron can be viewed as a hyperparameter. Our attacks are currently not applicable to steal such hyperparameters.   
%However, when a neural network adopts the conventional $L_1$ or $L_2$ regularization term, our attacks can steal the hyperparameter, as we will demonstrate in Appendix~\ref{neuralNet}. 

\begin{table}[t]
%\small
\centering
\caption{Loss functions and regularization terms of various ML algorithms we study in this paper.}
\label{sum_ml}
\begin{tabular}{|c|c|c|c|}
\hline
  Category        &       {ML Algorithm} & Loss Function & Regularization \\ \hline
\multirow{4}{*}{Regression} & RR & Least Square & $L_2$ \\ \cline{2-4} 
                  &        LASSO & Least Square & $L_1$ \\ \cline{2-4} 
                  &      ENet  & Least Square & $L_2 + L_1$ \\ \cline{2-4} 
                  & KRR       &  Least Square & $L_2$ \\ \hline
                  %& NN   	&  Least Square & $L_2$ \\ \hline
           
\multirow{4}{*}{Logistic Regression} & L2-LR & Cross Entropy & $L_2$ \\ \cline{2-4} 
                  &    L1-LR & Cross Entropy & $L_1$ \\ \cline{2-4} 
                  
                  & L2-KLR  & Cross Entropy & $L_2$ \\ \cline{2-4} 
                  
                  & L1-BKLR & Cross Entropy & $L_1$ \\ \hline

\multirow{4}{*}{SVM} & SVM-RHL & Regular Hinge Loss & $L_2$ \\ \cline{2-4} 
                  &  SVM-SHL & Square Hinge Loss & $L_2$ \\ \cline{2-4} 
                  & KSVM-RHL & Regular Hinge Loss & $L_2$ \\ \cline{2-4} 
                  & KSVM-SHL & Square Hinge Loss & $L_2$ \\ \hline

\multirow{2}{*}{Neural Network} & Regression & Least Square & $L_2$ \\ \cline{2-4} 
                  &  Classification & Cross Entropy & $L_2$ \\ \hline
                  
\end{tabular} 
\small
%\vspace{-2mm}
\end{table}

An ML algorithm minimizes the above objective function for a given training dataset and a given hyperparameter, to get the model parameters $\mathbf{w}$, i.e., $\mathbf{w}=\argmin \mathcal{L}(\mathbf{w})$. 
%Many ML algorithms--such as ridge regression, logistic regression, support vector machine, and neural network--have objective functions that are \emph{almost everywhere differentiable}.\footnote{Objective functions with $L_1$ regularization are not differentiable at the point $\mathbf{w}=\mathbf{0}$, so we use the term \emph{almost everywhere differentiable}.}  
The learnt model parameters are a \emph{minimum} of the objective function. $\mathbf{w}$ is a minimum if the objective function has larger values at the points near $\mathbf{w}$. %For convex objective functions, the minima is unique, while a non-convex objective function  might multiple minima. For instance, neural network's objective function is non-convex
Different ML algorithms adopt different loss functions and different regularization terms. For instance, ridge regression uses \emph{least-square} loss function and $L_2$ regularization term. Table~\ref{sum_ml} lists the loss function and regularization term used by popular ML algorithms that we consider in this work. %Some optimization algorithms can find the exact model parameters that minimize the objective function, while some optimization algorithms find model parameters that approximately minimize the objective function. 
As we will demonstrate, these different loss functions and regularization terms have different security properties against our hyperparameter stealing attacks. 

For kernel algorithms, the model parameters are in the form $\mathbf{w} = \sum_{i=1}^n \alpha_i \phi(\mathbf{x}_i)$. In other words, the model parameters are a linear combination of  the kernel mapping of the training instances. Equivalently, we can represent model parameters using the parameters $\mathbf{\bm \alpha}=\{\alpha_i\}_{i=1}^n$ for kernel algorithms.

%We note that neural network are a category of nonlinear ML algorithms.
%Modern (deep) neural network use \emph{dropout}~\cite{} as a regularization instead of conventional $L_1$ or $L_2$ regularization terms. 
%For such dropout regularization, 
%the dropout probability for each neuron can be viewed as a hyperparameter. Our attacks are currently not applicable to steal such hyperparameters.   
%However, when a neural network adopts the conventional $L_1$ or $L_2$ regularization term, our attacks can steal the hyperparameter, as we will demonstrate in Appendix~\ref{neuralNet}. 

%In this work, we focus on \emph{convex-inducing} machine learning algorithms, which include a variety of widely used machine learning algorithms such as ridge regression, logistic regression, and support vector machine.  
%Specifically, a machine learning algorithm is convex-inducing if its objective function is convex.  
%When the objective function is convex, there is a single \emph{global minima} $\mathbf{w}$ at which the objective function achieves its minimal value. 

\myparatight{Learning hyperparameters via cross-validation} Hyperparameter is a key parameter in machine learning systems; a good hyperparameter  makes it possible to learn a model that has good generalization performance on testing dataset. In practice, hyperparameters are often determined via cross-validation~\cite{hsu2003practical}. A popular cross-validation method is called \emph{$K$-fold} cross-validation. Specifically, we can divide the training dataset into $K$ folds. Suppose we are given a hyperparameter. For each fold, we learn the model parameters using the remaining $K-1$ folds as a training dataset and tests the model performance on the fold. Then, we average the  performance over the $K$ folds. The hyperparameter is determined in a search process such that the average  performance in the cross-validation is maximized. Learning hyperparameters is much more computationally expensive than learning model parameters with a given hyperparameter because the former involves many trials of learning model parameters. 
Figure~\ref{learningflow} illustrates the process to learn hyperparameters and model parameters.

\myparatight{Testing performance of the decision function}  We often use a testing dataset to measure performance of the learnt model parameters.  
Suppose the testing dataset consists of $\{\mathbf{x}_i^{test}\}_{i=1}^{n^{test}}$, whose labels are $\{{y}_i^{test}\}_{i=1}^{n^{test}}$, respectively. 
For regression, the performance is often measured by \emph{mean square error (MSE)}, which is defined as $\textrm{MSE} = \frac{1}{n^{test}} \sum_{i=1}^{n^{test}} ( {y}_i^{test} - f(\mathbf{x}_i^{test}))^2$. 
%\begin{align}
%\label{mse}
%\small
%\textrm{MSE} = \frac{1}{n^{test}} \sum_{i=1}^{n^{test}} ( {y}_i^{test} - f(\mathbf{x}_i^{test}))^2.  
%\end{align}
 For classification, the performance is often measured by \emph{accuracy (ACC)}, which is defined as $\textrm{ACC} = \frac{1}{n^{test}} \sum_{i=1}^{n^{test}} I({y}_i^{test} = f(\mathbf{x}_i^{test}))$, 
%\begin{align}
%\label{acc}
%\small
%\textrm{ACC} = \frac{1}{n^{test}} \sum_{i=1}^{n^{test}} I({y}_i^{test} = f(\mathbf{x}_i^{test})),  
%\end{align}
where $I$ is 1 if ${y}_i^{test} = f(\mathbf{x}_i^{test})$, otherwise it is 0. 
A smaller {MSE} or a higher {ACC} implies better model parameters.

\begin{figure}[!t]
\center
{\includegraphics[width=0.45\textwidth]{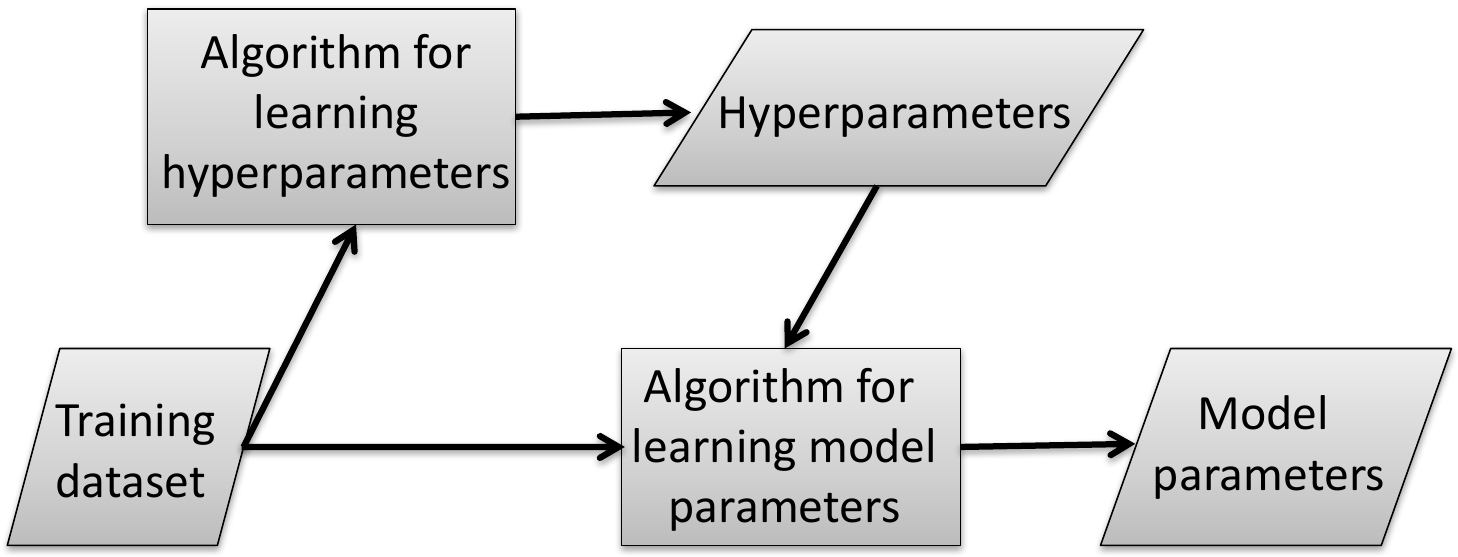}} 
\caption{Key concepts of a machine learning system.}
\label{learningflow}
\vspace{-6mm}
\end{figure}

\subsection{Problem Definition}
\label{prob_def}

\myparatight{Threat model} We assume the attacker knows the training dataset, the ML algorithm, and (optionally) the learnt model parameters. 
Our threat model is motivated by machine-learning-as-a-service (MLaaS)~\cite{BigML,amazon,google,ms}. MLaaS is an emerging technology to aid users, who have limited computing power and machine learning expertise, to learn an ML model over large datasets. Specifically, a user uploads the training dataset to an MLaaS platform and specifies an ML algorithm. The MLaaS platform uses proprietary algorithms to learn the hyperparameters, then learns the model parameters, and finally certain MLaaS platforms (e.g., BigML~\cite{BigML}) allow the user to download the model parameters to use them locally. 
Attackers could be such users. {We stress that when the model parameters are unknown, our attacks are still applicable as we demonstrate in Section~\ref{unknownparameter}. Specifically,
the attacker can first use model parameter stealing attacks~\cite{tramer2016stealing} to learn them and then perform our attacks}. We note that various MLaaS platforms--such as Amazon Machine Learning and Microsoft Azure Machine Learning--make the ML algorithm public. Moreover, for black-box MLaaS platforms such as Amazon Machine Learning and Microsoft Azure Machine Learning, prior model parameter stealing attacks~\cite{tramer2016stealing} are applicable.

%\alan{Our threat model is realistic. Specifically, various MLaaS platforms--such as Amazon Machine Learning and Microsoft Azure Machine Learning--make the ML algorithm public. Moreover, for Amazon Machine Learning and Microsoft Azure Machine Learning, prior model parameter stealing attacks~\cite{tramer2016stealing} are applicable.
%%Specifically, a user  can be an attacker, and the user knows its training dataset, the specified learning algorithm (or the attacker can know the learning algorithm through documentation of the MLaaS~\cite{tramer2016stealing}), and the learnt model parameters. 
%Although different MLaaS platforms adopt different protocols for users to train ML models, we consider MLaaS (e.g., black-box MLaaS platforms such as Amazon Machine Learning and Microsoft Azure Machine Learning) that do not share the learnt hyperparameters with the  user due to intellectual property and algorithm confidentiality. }
%%, when describing our attacks, we assume the model  

We define \emph{hyperparameter stealing attacks} as follows:
\alan{
\begin{definition}[Hyperparameter Stealing Attacks] Suppose an ML algorithm learns model parameters via minimizing an objective function that is in the form of loss function + $\lambda\ \times$ regularization term. Given the ML algorithm, the training dataset, and (optionally) the learnt model parameters, hyperparameter stealing attacks aim to estimate the hyperparameter value in the objective function.
%we are given a training dataset, a learning algorithm, and the model parameters that are learnt using the learning algorithm and the training dataset. Hyperparameter stealing attacks aim to estimate the hyperparameter value in the objective function of the learning algorithm.
\end{definition}
}

%Note that a ML algorithm is characterized by an objective function, which consists of a loss function, a regularization term, and a hyperparameter.

% and the MLaaS charges a user according to the amount of computations performed to learn the model.  

\myparatight{Application scenario} One application of our hyperparameter stealing attacks is that a user can use our attacks to learn a model via MLaaS with much less computations (thus much less economical costs), while not sacrificing the model's testing performance. Specifically, the user can sample a small fraction of the training dataset, learns the model parameters through MLaaS, steals the hyperparameter using our attacks, and re-learns the model parameters via  MLaaS using the entire training dataset and the stolen hyperparameter. We will demonstrate this application scenario in Section~\ref{attackresults} via simulations and Amazon Machine Learning.

\section{Hyperparameter Stealing Attacks}
\label{hyper_attack}

We first introduce our general attack framework. 
Second, we use several regression and classification algorithms as examples to illustrate 
how we can use our framework to steal hyperparameters for specific ML algorithms, 
and we show results of more algorithms in Appendix~\ref{app:other}. 
%Third, we demonstrate the effectiveness of our attacks both theoretically and empirically.  

\subsection{Our Attack Framework}
Our goal is to steal the hyperparameters in an objective function. For an ML algorithm that uses such hyperparameters, 
% Many popular ML algorithms--such as ridge regression, logistic regression, support vector machine, and neural network--have \emph{almost everywhere differentiable} objective functions. 
%Therefore, in this work, we focus on ML algorithms with {almost everywhere differentiable} objective functions.
%For such a ML algorithm, 
 the learnt model parameters are often a minimum of the objective function (see the background knowledge in Section~\ref{sec:back}). 
Therefore, the \emph{gradient} of the objective function at the learnt model parameters should be $\mathbf{0}$, which encodes the relationships between the learnt model parameters and the hyperparameters. We leverage this key observation to steal hyperparameters. 

\myparatight{Non-kernel algorithms} We compute the gradient of the objective function at the model parameters $\mathbf{w}$ and set it to be $\mathbf{0}$. Then, we have
\begin{align}
\label{attack_lin_soln}
\frac{\partial \mathcal{L}(\mathbf{w})}{\partial \mathbf{w}} = \mathbf{b} + \lambda \mathbf{a}= \mathbf{0},
\end{align}
where vectors $\mathbf{b}$ and $\mathbf{a}$ are defined as follows:
\begin{small} 
\begin{align}
% \mathbf{A} = \left[
% \begin{array}{cccc} 
% \frac{\partial \text{R}(\mathbf{w})}{\partial w_1} & 0 & \cdots & 0 \\
% 0 & \frac{\partial \text{R}(\mathbf{w})}{\partial w_1} & \cdots & 0 \\
% \vdots & \vdots \\
% 0 & \cdots & 0 & \frac{\partial \text{R}(\mathbf{w})}{\partial w_m} \\
% \end{array}
% \right], \\
\mathbf{b} = \left[ 
\begin{array}{c}
 \frac{ \partial \text{L}(\mathbf{X}, \mathbf{y}, \mathbf{w})}{\partial {w}_1} \\
\frac{ \partial \text{L}(\mathbf{X}, \mathbf{y}, \mathbf{w})}{\partial {w}_2} \\
\vdots \\
 \frac{ \partial \text{L}(\mathbf{X}, \mathbf{y}, \mathbf{w})}{\partial {w}_m} \\
\end{array}
\right], \, 
\mathbf{a} = \left[
\begin{array}{c}
\frac{\partial \text{R}(\mathbf{w})}{\partial w_1} \\
\frac{\partial \text{R}(\mathbf{w})}{\partial w_2} \\
\vdots \\
\frac{\partial \text{R}(\mathbf{w})}{\partial w_m} \\
\end{array}
\right].
\end{align}
\end{small}

First, Eqn.~\ref{attack_lin_soln} is a \emph{system of linear equations} about the hyperparameter $\lambda$. 
Second, in this system, the number of equations is more than the number of unknown variables (i.e., hyperparameter in our case). Such system is called 
an \emph{overdetermined system} in statistics and mathematics. 
We adopt the \emph{linear least square} method~\cite{montgomery2015introduction}, a popular method to find an approximate solution to an overdetermined system, 
to solve the hyperparameter in  Eqn.~\ref{attack_lin_soln}. More specifically, we estimate the hyperparameter as follows:
\begin{align}
\label{fianl_est}
\hat{\lambda} = -(\mathbf{a}^T \mathbf{a})^{-1}{\mathbf{a}^T \mathbf{b}}.
\end{align}

\myparatight{Kernel algorithms} Recall that, for kernel algorithms, the model parameters are a linear combination of the kernel mapping of the training instances, i.e., $\mathbf{w} = \sum_{i=1}^n \alpha_i \phi(\mathbf{x}_i)$. Therefore, the model parameters can be equivalently represented by the parameters $\mathbf{\bm \alpha}=\{\alpha_i\}_{i=1}^n$. 
We replace the variable $\mathbf{w}$ with $\sum_{i=1}^n \alpha_i \phi(\mathbf{x}_i)$ in the objective function, compute the gradient of the objective function with respect to $\mathbf{\bm \alpha}$, and set the gradient to $\mathbf{0}$. Then, we will obtain an overdetermined system. After solving the system with linear least square method, we again estimate the hyperparameter using Eqn.~\ref{fianl_est} with the vectors $\mathbf{b}$ and $\mathbf{a}$ re-defined as follows:
\begin{small} 
\begin{align}
% \mathbf{A} = \left[
% \begin{array}{cccc} 
% \frac{\partial \text{R}(\mathbf{w})}{\partial w_1} & 0 & \cdots & 0 \\
% 0 & \frac{\partial \text{R}(\mathbf{w})}{\partial w_1} & \cdots & 0 \\
% \vdots & \vdots \\
% 0 & \cdots & 0 & \frac{\partial \text{R}(\mathbf{w})}{\partial w_m} \\
% \end{array}
% \right], \\
\mathbf{b} = \left[ 
\begin{array}{c}
 \frac{ \partial \text{L}(\phi(\mathbf{X}), \mathbf{y}, \mathbf{w})}{\partial {\alpha}_1} \\
 \frac{ \partial \text{L}(\phi(\mathbf{X}), \mathbf{y}, \mathbf{w})}{\partial {\alpha}_2} \\
\vdots \\
 \frac{ \partial \text{L}(\phi(\mathbf{X}), \mathbf{y}, \mathbf{w})}{\partial {\alpha}_n} \\
\end{array}
\right], \, 
\mathbf{a} = \left[
\begin{array}{c}
\frac{\partial \text{R}(\mathbf{w})}{\partial \alpha_1} \\
\frac{\partial \text{R}(\mathbf{w})}{\partial \alpha_2} \\
\vdots \\
\frac{\partial \text{R}(\mathbf{w})}{\partial \alpha_n} \\
\end{array}
\right].
\end{align}
\end{small}

%\vspace{-6mm}

\myparatight{Addressing non-differentiability} Using Eqn.~\ref{fianl_est} still faces two more challenges: 1) the objective function might not be differentiable at certain dimensions of the model parameters $\mathbf{w}$ (or $\mathbf{\bm \alpha}$), and 2) the objective function might not be differentiable at certain training instances for the learnt model parameters. For instance, objective functions with $L_1$ regularization are not differentiable at the dimensions where the learnt model parameters are 0, while the objective functions in SVMs might not be differentiable for certain training instances.  We address the challenges via constructing the vectors $\mathbf{a}$ and $\mathbf{b}$ using the dimensions and training instances at which the objective function is differentiable. Note that using less dimensions of the model parameters is equivalent to using less equations in the overdetermined system shown in Eqn.~\ref{attack_lin_soln}. Once we have at least one dimension of the model parameters and one training instance at which the objective function is differentiable, we can estimate the hyperparameter.

\myparatight{Attack procedure}  We summarize our hyperparameter stealing attacks in the following two steps:
\begin{packeditemize}
\item {\bf Step I.} The attacker computes the vectors $\mathbf{a}$ and $\mathbf{b}$ for a given training dataset, a given ML algorithm, and the learnt model parameters.
\item {\bf Step II.} The attacker estimates the hyperparameter using Eqn.~\ref{fianl_est}.
\end{packeditemize}

\myparatight{More than one hyperparameter}
We note that, for conciseness, we focus on ML algorithms whose objective functions have a single hyperparameter in the main text of this paper. However, our attack framework is applicable and can be easily extended to ML algorithms that use more than one hyperparameter. Specifically, we can still estimate the hyperparameters using Eqn.~\ref{fianl_est} with the vector $\mathbf{a}$ expanded to be a matrix, where each column corresponds to the gradient of a regularization term with respect to the model parameters. We use an example ML algorithm, i.e., Elastic Net~\cite{zou2005regularization}, with two hyperparameters to illustrate our attacks in Appendix~\ref{app:enet}.

Next, we use various popular regression and classification algorithms to illustrate our attacks. 
In particular, we will discuss how we can compute the vectors $\mathbf{a}$ and $\mathbf{b}$. 
We will focus on linear and kernel ML algorithms for simplicity,  
% as modern deep neural network use dropout regularization.
and we will show  results on  neural networks in Appendix~\ref{neuralNet}.  
We note that the ML algorithms we study are widely deployed by MLaaS.  For instance, logistic regression is deployed by Amazon Machine Learning, Microsoft Azure Machine Learning, BigML, etc.; SVM is employed by Microsoft Azure Machine Learning, Google Cloud Platform, and PredictionIO. 

\subsection{Attacks to Regression Algorithms} 
\label{attack_cvx_reg}
%We use some popular regression algorithms to illustrate our attacks. 
%In particular, we will illustrate how we can compute the vectors  $\mathbf{b}$ and $\mathbf{a}$. 
\subsubsection{Linear Regression Algorithms} 
\label{cvx_reg_lin}

We demonstrate our attacks to popular linear regression algorithms including \emph{Ridge Regression (RR)}~\cite{hoerl1970ridge} and \emph{LASSO}~\cite{tibshirani1996regression}.
% and \emph{Elastic Net (ENet)}~\cite{zou2005regularization}. 
Both algorithms use \emph{least square} loss function, and their regularization terms are $L_2$ and $L_1$, respectively. 
Due to the limited space, we show attack details for RR, and the details for LASSO are shown in Appendix~\ref{app:other}.
%For simplicity, we only show the objective function and the stealed hyperparameter(s) of each model. The derivation details are provided in Appendix~\ref{app:cvx_reg}.
The objective function of RR is $\mathcal{L}(\mathbf{w}) = \| \mathbf{y}-\mathbf{X}^T \mathbf{w} \|_2^2 + \lambda \| \mathbf{w} \|_2^2$.
%\begin{align}
%\label{ridge_reg}
%\mathcal{L}(\mathbf{w}) = \| \mathbf{y}-\mathbf{X}^T \mathbf{w} \|_2^2 + \lambda \| \mathbf{w} \|_2^2.
%\end{align}
%RR often uses a SVD algorithm~\cite{weisstein2002singular} to minimize the objective function to get model parameters.  Note that, the model parameters found by RR are usually  an exact minima of the objective function.  
 We compute the gradient of the objective function with respect to $\mathbf{w}$, and we have $\frac{\partial \mathcal{L}(\mathbf{w})}{\partial \mathbf{w}} = -2 \mathbf{X y} + 2 \mathbf{X}\mathbf{X}^T\mathbf{w} + 2 \lambda \mathbf{w}$.
%\begin{align*}
%\frac{\partial \mathcal{L}(\mathbf{w})}{\partial \mathbf{w}} = -2 \mathbf{X y} + 2 \mathbf{X}\mathbf{X}^T\mathbf{w} + 2 \lambda \mathbf{w}. 
%\end{align*}
By setting the gradient to be $\mathbf{0}$, we can estimate $\lambda$ using Eqn.~\ref{fianl_est} with $\mathbf{a} = \mathbf{w}$ and $\mathbf{b} = \mathbf{X} (\mathbf{X}^T\mathbf{w} - \mathbf{y})$.

\subsubsection{Kernel Regression Algorithms}
\label{cvx_reg_nonlin}

We use \emph{kernel ridge regression (KRR)}~\cite{vovk2013kernel} as an example to illustrate our attacks. Similar to linear RR, KRR uses least square loss function and $L_2$ regularization.
After we represent the model parameters $\mathbf{w}$ using $\bm{\alpha}$, the objective function of KRR is $\mathcal{L}(\bm{\alpha})= \| \mathbf{y} - \mathbf{K} \bm{\alpha} \|_2^2 + \lambda \bm{\alpha}^T \mathbf{K} \bm{\alpha}$, 
%% \begin{small}
%\begin{align}
%\label{kernel_ridge_dual}
%\begin{split}
%%\mathcal{L}_{KRR}(\lambda, \bm{\alpha}) 
%\mathcal{L}(\bm{\alpha})
%%& = \| \mathbf{y}-\Phi(\mathbf{X}) \Phi(\mathbf{X})^T \bm{\alpha} \|_2^2 + \lambda \| \Phi(\mathbf{X})^T \bm{\alpha} \|_2^2 \\
%& = \| \mathbf{y} - \mathbf{K} \bm{\alpha} \|_2^2 + \lambda \bm{\alpha}^T \mathbf{K} \bm{\alpha},
%\end{split}
%\end{align}
%%\end{small}
where matrix $\mathbf{K}=\phi(\mathbf{X})^T \phi(\mathbf{X})$, whose $(i,j)$th entry is $ \phi(\mathbf{x}_i)^T \phi(\mathbf{x}_j)$. 
In machine learning, $\mathbf{K}$ is called \emph{gram matrix} and is positive definite. 
%and the model parameters $\bm{\alpha}$ found by SVD algorithm are an exact minima of the objective function~\cite{weisstein2002singular}.
By computing the gradient of the objective function  with respect to $\bm{\alpha}$ and setting it to be \textbf{0}, we have $\mathbf{K}(\lambda \bm{\alpha} + \mathbf{K} \bm{\alpha} - \mathbf{y}) = \mathbf{0}$. 
%%\begin{small}
%\begin{align*}
%%\label{eqn:krr}
%\mathbf{K}(\lambda \bm{\alpha} + \mathbf{K} \bm{\alpha} - \mathbf{y}) = \mathbf{0}.
%\end{align*}
%%\end{small} 
$\mathbf{K}$ is invertible as it is positive definite. Therefore, we multiply both sides of the above equation with $\mathbf{K}^{-1}$. Then, we can estimate $\lambda$ using Eqn.~\ref{fianl_est} with 
$\mathbf{a} = \bm{\alpha}$ and $\mathbf{b} = \mathbf{K} \bm{\alpha} - \mathbf{y}$. 
Our attacks are applicable to any kernel function. In our experiments, we will adopt the widely used Gaussian kernel.

\subsection{Attacks to Classification Algorithms}
\label{attack_cvx_clf}

%We consider attacks on both convex linear and kernel-based nonlinear classification models.
%We consider attacks on both convex linear and kernel classification models.

\subsubsection{Linear Classification Algorithms}
\label{cvx_clf_lin}

We demonstrate our attacks to four popular linear classification algorithms: \emph{support vector machine with regular hinge loss function (SVM-RHL)},
 \emph{support vector machine with squared hinge loss function (SVM-SHL)},  \emph{$L_1$-regularized logistic regression (L1-LR)},
and \emph{$L_2$-regularized logistic regression (L2-LR)}. These four algorithms enable us to compare different regularization terms  and different loss functions. 
%We will describe the details of regular hinge loss and squared hinge loss when we describe our attacks to SVM-RHL and SVM-SHL. 
For simplicity, we show attack details for L1-LR, and defer details for other algorithms to Appendix~\ref{app:other}. L1-LR enables us to illustrate how we address the challenge where the objective function is not differentiable at certain dimensions of the model parameters.  

We focus on binary classification, since multi-class classification is often transformed to multiple binary classification problems via the \emph{one-vs-all} paradigm. However, our attacks are also applicable to multi-class classification algorithms such as \emph{multi-class support vector machine}~\cite{chang2011libsvm} and \emph{multi-class logistic regression}~\cite{chang2011libsvm} that use hyperparameters in their objective functions. For binary classification, each training instance has a label $y_i\in \{1,0\}$.

The objective function of L1-LR is $\mathcal{L}(\mathbf{w}) = \text{L}(\mathbf{X}, \mathbf{y}, \mathbf{w}) + \lambda \| \mathbf{w} \|_1$, 
%\begin{small}
%\begin{align}
%\label{l2lr}
%\mathcal{L}(\mathbf{w}) = \text{L}(\mathbf{X}, \mathbf{y}, \mathbf{w}) + \lambda \| \mathbf{w} \|_1,
%\end{align}
%\end{small}
where $\text{L}(\mathbf{X}, \mathbf{y}, \mathbf{w})$=$-\sum_{i=1}^n  ( y_i \, \log h_\mathbf{w}(\mathbf{x}_i) + (1-y_i)$ $ \log (1- h_\mathbf{w}(\mathbf{x}_i)))$ is called \emph{cross entropy} loss function and $h_\mathbf{w}(\mathbf{x})$ is defined to be $\frac{1}{1+\exp{(- \mathbf{w}^T \mathbf{x})}}$. 
The gradient of the objective function is $\frac{\partial \mathcal{L}(\mathbf{w})}{\partial \mathbf{w}} =\mathbf{X} (\mathbf{h}_\mathbf{w}(\mathbf{X})  - \mathbf{y}) + \lambda \text{sign}(\mathbf{w})$,  
%\begin{small}
%\begin{align*}
%%\label{l2lr_grad}
%\begin{split}
%& \frac{\partial \mathcal{L}(\mathbf{w})}{\partial \mathbf{w}} =
%% - \sum_{i=1}^n \left[ \frac{y_i}{h_\mathbf{w}(\mathbf{x}_i)} \frac{\partial h_\mathbf{w}(\mathbf{x}_i)}{\partial \mathbf{w}} + \frac{-(1-y_i)}{1-h_\mathbf{w}(\mathbf{x}_i)} \frac{\partial h_\mathbf{w}(\mathbf{x}_i)}{\partial \mathbf{w}}\right] + 2 \lambda \mathbf{w} \\
%%& \quad = - \sum_{j=1}^n \left[ \frac{y_j }{h_\mathbf{w}(\mathbf{x}_i)} h_\mathbf{w}(\mathbf{x}_i) (1-h_\mathbf{w}(\mathbf{x}_i)) - \frac{(1-y_j)}{1-h_\mathbf{w}(\mathbf{x}_i)} h_\mathbf{w}(\mathbf{x}_i) (1-h_\mathbf{w}(\mathbf{x}_i)) \right] \frac{\partial (\mathbf{w})^T \mathbf{x}^j}{\partial \mathbf{w}} + 2 \lambda \mathbf{w} \\
%%& \quad = - \sum_{i=1}^n \left(y_i (1-h_\mathbf{w}(\mathbf{x}_i)) - (1-y_i) h_\mathbf{w}(\mathbf{x}_i) \right) \mathbf{x}_i + 2 \lambda \mathbf{w} \\ 
%%& \quad = \sum_{i=1}^n (h_\mathbf{w}(\mathbf{x}_i) - y_i) \mathbf{x}_i + 2 \lambda \mathbf{w} = 
%\mathbf{X} (\mathbf{h}_\mathbf{w}(\mathbf{X})  - \mathbf{y}) + \lambda \text{sign}(\mathbf{w}),
%\end{split}
%\end{align*}
%\end{small}
where  $\mathbf{h}_\mathbf{w}(\mathbf{X}) = \left[ h_\mathbf{w}(\mathbf{x}_1); h_\mathbf{w}(\mathbf{x}_2); \cdots; h_\mathbf{w}(\mathbf{x}_n)) \right]$ and the $i$th entry $\text{sign}(w_i)$ of the vector $\text{sign}(\mathbf{w})$ is defined as follows:
\begin{small}
\begin{align}
\label{vectors}
\text{sign}(w_i)=\frac{\partial |w_i|}{\partial w_i}=
\begin{cases}
-1 & \text{if } w_i < 0 \\
0 & \text{if } \, w_i = 0 \\
1 & \text{if } w_i > 0
\end{cases}
\end{align}
\end{small}
$| w_i |$ is \emph{not differentiable} when $w_i=0$, so
we define the derivative at $w_i=0$ as 0, which means that we do not use the model parameters that are 0 to estimate the hyperparameter.
Via setting the gradient to be \textbf{0}, we can estimate $\lambda$ using Eqn.~\ref{fianl_est} with
$\mathbf{a} = \text{sign}(\mathbf{w})$ and $\mathbf{b} = \mathbf{X} (\mathbf{h}_\mathbf{w}(\mathbf{X})  - \mathbf{y})$.

\subsubsection{Kernel Classification Algorithms}
\label{cvx_clf_nonlin}

We demonstrate our attacks to the kernel version of the above four linear classification algorithms: \emph{kernel support vector machine with regular hinge loss function (KSVM-RHL)}, \emph{kernel support vector machine with squared hinge loss function (KSVM-SHL)}, \emph{$L_1$-regularized kernel LR (L1-KLR)}, and \emph{$L_2$-regularized kernel LR (L2-KLR)}.
We show attack details for KSVM-RHL, and defer details for the other algorithms in Appendix~\ref{app:other}. KSVM-RHL enables us to illustrate how we can address the challenge where the objective function is non-differentiable for certain training instances. Again, we focus on binary classification. 

The objective function of KSVM-RHL  is $\mathcal{L}(\bm{\alpha}) = \sum_{i=1}^n L(\phi(\mathbf{x}_i), y_i,  \bm{\alpha}) + \lambda \bm{\alpha}^T \mathbf{K} \bm{\alpha}$, 
%%\begin{small}
%%\begin{align*}
%%%\label{kernel_bsvc_prime}
%%\mathcal{L}_{KSVM}(\lambda, \mathbf{w^\prime})  = \lambda \sum_{i=1}^n L_{HL}(y_j, \langle \mathbf{w^\prime}, \phi(\mathbf{x}_j) \rangle) + \frac{1}{2} \| \mathbf{w^\prime} \|_2^2. 
%%\end{align*}
%%\end{small}
%%Using $\mathbf{w}^\prime = \Phi(\mathbf{X})^\lambda \bm{\alpha}$, we have its dual form
%\begin{small}
%\begin{align}
%\label{kernel_bsvc_dual}
%\begin{split}
%\mathcal{L}(\bm{\alpha}) = \sum_{i=1}^n L(\phi(\mathbf{x}_i), y_i,  \bm{\alpha}) + \lambda \bm{\alpha}^T \mathbf{K} \bm{\alpha},
%\end{split}
%\end{align}
%\end{small}
where $L(\phi(\mathbf{x}_i), y_i,  \bm{\alpha}) = \max(0, 1 - y_i  \bm{\alpha}^T \mathbf{k}_i )$ is called \emph{regular hinge loss function}. $\mathbf{k}_i$ is the $i$th column of the gram matrix $\mathbf{K}=\phi(\mathbf{X})^T \phi(\mathbf{X})$.
The gradient of the loss function with respect to $\mathbf{\bm{\alpha}}$ is: 
\begin{small}
\begin{align}
%\label{RHL_subg}
\frac{\partial L(\phi(\mathbf{x}_i), y_i, \mathbf{\bm \alpha})}{\partial \mathbf{\bm \alpha}} = 
\begin{cases}
-y_i \mathbf{k}_i & \text{if } y_i  \mathbf{\bm \alpha}^T \mathbf{k}_i  < 1 \\
%[-y_j \mathbf{x}_j, \mathbf{0}] & \text{if } y_j \langle \mathbf{w}, \mathbf{x}_j \rangle = 1 \\
\mathbf{0} & \text{if } y_i  \mathbf{\bm \alpha}^T \mathbf{k}_i   > 1,
\end{cases}
\end{align}
\end{small}
where $L(\phi(\mathbf{x}_i), y_i, \mathbf{\bm \alpha})$ is non-differentiable when $\mathbf{k}_i$ satisfies  $y_i  \mathbf{\bm \alpha}^T \mathbf{k}_i  = 1$.
 We estimate $\lambda$ using $\mathbf{k}_i$ that satisfy $y_i  \mathbf{\bm \alpha}^T \mathbf{k}_i  < 1$.  
Specifically, via setting the gradient of the objective function to be $\mathbf{0}$, we estimate $\lambda$ using Eqn.~\ref{fianl_est} with $\mathbf{a} = 2 \mathbf{K} \bm{\alpha}$ and $\mathbf{b} = \sum_{i=1}^n -y_i \mathbf{k}_i \mathbf{1}_{y_i  \bm{\alpha}^T \mathbf{k}_i < 1}$, where  $\mathbf{1}_{y_i  \bm{\alpha}^T \mathbf{k}_i < 1}$ is an indicator function with value 1 if $ y_i  \bm{\alpha}^T \mathbf{k}_i  < 1 $ and 0 otherwise.

\section{Evaluations}
\subsection{Theoretical Evaluations}
We aim to evaluate the effectiveness of our hyperparameter stealing attacks theoretically. %These theoretic analysis will guide us to design a defense against our hyperparameter stealing attacks. 
In particular, we show that 1) when the learnt model parameters are an exact minimum of the objective function, our attacks can obtain the exact hyperparameter value, and 2) when the model parameters deviate from their closest minimum of the objective function with a small difference, then the estimation error of our attacks is a linear function of the small difference. Specifically, our theoretical analysis can be summarized in the following theorems.

\begin{theorem}
%\begin{theorem}%[$0$-Hyperparameter Learnable]
\label{thm_1}
Suppose an ML algorithm learns model parameters via minimizing an \emph{objective function} which is in the form of \emph{loss function} + $\lambda\  \times$ \emph{regularization term}, $\lambda$ is the true hyperparameter value, and the learnt model parameters $\mathbf{w}$ (or $\bm{\alpha}$ for kernel algorithms) are an exact minimum of the objective function. 
Then, our attacks can obtain the exact true hyperparameter value, i.e., $\hat{\lambda} = {\lambda}$.   
%that the objective function of a ML algorithm is differential and its (locally or globally) optimal parameters and the respective hyperparameter are $\mathbf{w}^{\star}$ (or $\bm{\alpha}^{\star}$) and $\lambda^{\star}$. Suppose further that the ML algorithm is solved via an optimimaztion algorithm which has an \emph{accurate analytic solution}, and the learnt model parameters and hyperparamter are $\tilde{\mathbf{w}}$ (or $\tilde{\bm{\alpha}}$) and $\tilde{\lambda}$. That is, $\tilde{\mathbf{w}} = \mathbf{w}^{\star}$ (or $\tilde{\bm{\alpha}} = \bm{\alpha}^{\star}$). Then, using our attack strategy, the ML algorithm is $0$-hyperparameter learnable, i.e., $\hat{\lambda} = \tilde{\lambda}$. 
\end{theorem}

\begin{proof}
See Appendix~\ref{app:thm_1}.
\end{proof}
%RR and KRR belong to this category. 

% \begin{theorem}%[$0$-Hyperparameter Learnable]
% \label{obs_1}
% Suppose that the objective function of a ML algorithm is differential and its optimal parameters are $\mathbf{w}^{\star}$ (or \bm{\alpha}^{\star}). Assume that the parameters are learnt with an \emph{approximate solution}, denoted as $\hat{\mathbf{w}}$ (or \hat{\bm{\alpha}}), i.e.,  $\hat{\mathbf{w}} = \mathbf{w}^{\star}$ (or $\hat{\bm{\alpha}} = \bm{\alpha}^{\star}$), 
% %$\mathbf{w}_{\text{primal}}$ (or $\bm{\alpha}_{\text{primal}}$) are learnt in the \emph{primal space} with an \emph{accurate analytic solution}, i.e., $\mathbf{w}_{\text{primal}} = \mathbf{w}^{\star}$ (or $\bm{\alpha}_{\text{primal}} = \bm{\alpha}^{\star}$). 
% Then, the ML algorithm is $0$-hyperparameter learnable.
% \end{theorem}

\begin{theorem}%[$0$-Hyperparameter Learnable]
\label{thm_2}
Suppose an ML algorithm learns model parameters via minimizing an \emph{objective function} which is in the form of \emph{loss function} + $\lambda \ \times$ \emph{regularization term}, $\lambda$ is the true hyperparameter value,  the learnt model parameters are $\mathbf{w}$ (or $\bm{\alpha}$ for kernel algorithms), and $\mathbf{w}^{\star}$ (or $\bm{\alpha}^{\star}$) is the minimum of the objective function that is closest to $\mathbf{w}$ (or $\bm{\alpha}$). We denote $\Delta \mathbf{w} =\mathbf{w}-\mathbf{w}^{\star}$ and $\Delta \mathbf{\bm \alpha}=\bm{\alpha} - \bm{\alpha}^{\star}$. Then, when $\Delta \mathbf{w}\rightarrow \mathbf{0}$ or $\Delta \mathbf{\bm \alpha} \rightarrow \mathbf{0}$,  the difference between the estimated hyperparameter $\hat{\lambda}$ and the true hyperparameter can be bounded as follows:
\begin{align}
&\text{\bf Non-kernel algorithms:} \nonumber \\
&\Delta \hat{\lambda} = \hat{\lambda}-{\lambda} = \Delta \mathbf{w}^T \nabla \hat{\lambda}(\mathbf{w}^{\star}) + O(\| \Delta \mathbf{w}\|_2^2) \\
&\text{\bf Kernel algorithms:} \nonumber \\
&\Delta \hat{\lambda} = \hat{\lambda}-{\lambda} = \Delta \mathbf{\bm \alpha}^T \nabla \hat{\lambda}(\mathbf{\bm \alpha}^{\star}) + O(\| \Delta \mathbf{\bm \alpha}\|_2^2),
\end{align}
where  $\nabla \hat{\lambda}(\mathbf{w}^{\star})$ is the gradient of $\hat{\lambda}$ at  $\mathbf{w}^{\star}$ and $\nabla \hat{\lambda}(\mathbf{\bm \alpha}^{\star})$ is the gradient of $\hat{\lambda}$ at $\mathbf{\bm \alpha}^{\star}$. %Note that when $\hat{\lambda}$ is not differentiable at $\mathbf{w}^{\star}$ or $\mathbf{\bm \alpha}^{\star}$ in certain dimensions, we define the corresponding derivatives to be 0.  
%Suppose that the objective function of a ML algorithm is differential and its (locally or globally) optimal parameters are $\mathbf{w}^{\star}$ (or $\bm{\alpha}^{\star}$). Suppose further that the ML algorithm is solved via an optimimaztion algorithm which has an \emph{approximate solution}, and the learnt model parameters and hyperparamter are $\tilde{\mathbf{w}}$ (or $\tilde{\bm{\alpha}}$) and $\tilde{\lambda}$. That is, $\tilde{\mathbf{w}} \neq \mathbf{w}^{\star}$ (or $\tilde{\bm{\alpha}} \neq \bm{\alpha}^{\star}$), and their difference is denoted as $\Delta \mathbf{w} = \tilde{\mathbf{w}} - \mathbf{w}^{\star}$ (or $\Delta \bm{\alpha} = \tilde{\bm{\alpha}} - \bm{\alpha}^{\star}$). Then, using our attack strategy, the ML algorithm is $\Delta \hat{\lambda}$-hyperparameter learnable, where 
%$\Delta \hat{\lambda} = \Delta \mathbf{w}^T \partial_{\mathbf{w}^{\star}} \hat{\lambda}(\mathbf{w}^{\star}) + \frac{1}{2} \Delta \mathbf{w}^T \mathbf{H}(\mathbf{w}^{\star}) \Delta \mathbf{x} + O(\|\Delta \mathbf{w}\|_2^2)$ and $\mathbf{H}(\mathbf{w}^{\star})$ is a \emph{Hessian matrix} defined as $ \big( \mathbf{H}(\mathbf{w}^{\star}) \big)_{i,j} = \frac{\partial^2 \hat{\lambda}(\mathbf{w}^{\star})}{\partial w_i \partial w_j}$. 
\end{theorem}
\begin{proof}
See Appendix~\ref{app:thm_2}.
\end{proof}

\begin{table}[!t]
%\vspace{-4mm}
\centering
\caption{Datasets.}
\label{data_reg}
%\addtolength{\tabcolsep}{-2pt}
\begin{tabular}{|c|c|c|c|}
\hline
 \textbf{Dataset} & \textbf{\#Instances} & \textbf{\#Features} &\textbf{Type}  \\ \hline
 \textbf{Diabetes} &  442 & 10 & \multirow{3}{*}{Regression} \\ \cline{1-3}
\textbf{GeoOrig} & 1059 & 68  & \\ \cline{1-3}
\textbf{UJIIndoor} & 19937 & 529 & \\ \hline \hline
 \textbf{Iris}  &  100  &  4  & \multirow{3}{*}{Classification}  \\ \cline{1-3}
% \textbf{Ionosphere} & 351 & 34  \\ \hline
% \textbf{Sonar} & 208 & 60 &  \textrm{Binary}\\ \hline
% \textbf{Breast Cancer}& 569 & 30 & \textrm{Binary} \\ \hline
% \textbf{LSVT} & 125 & 310 & \textrm{Binary} \\ \hline
 \textbf{Madelon} & 4400 & 500 &  \\ \cline{1-3}
 \textbf{Bank} & 45210 & 16  & \\ \hline
\end{tabular}
%\vspace{-4mm}
\end{table}

%\begin{table}[!tbhp]
%\centering
%\caption{Datasets for classification.}
%\label{data_clf}
%\begin{tabular}{|c|c|c|c|}
%\hline
% \textbf{Dataset} & \textbf{\#Instances} & \textbf{\#Features} & \textbf{Type} \\ \hline
% \textbf{Ionosphere} & 351 & 34 & \textrm{Binary} \\ \hline
%% \textbf{Sonar} & 208 & 60 &  \textrm{Binary}\\ \hline
%% \textbf{Breast Cancer}& 569 & 30 & \textrm{Binary} \\ \hline
%% \textbf{LSVT} & 125 & 310 & \textrm{Binary} \\ \hline
% \textbf{Madelon} & 4400 & 500 & \textrm{Binary} \\ \hline
% \textbf{Bank} & 45210 & 16 & \textrm{Binary} \\ \hline
% \textbf{Iris} & 100 & 4 &  \textrm{Multi-Class} \\ \hline
% \textbf{Digit} & 1797 & 64 & \textrm{Multi-Class} \\ \hline
% \textbf{CovType} & 581012 & 54 & \textrm{Multi-Class} \\ \hline
%\end{tabular}
%\end{table}

%\begin{table}[!tbhp]
%\centering
%\caption{Datasets for classification.}
%\label{data_clf}
%\begin{tabular}{|c|c|c|}
%\hline
% \textbf{Dataset} & \textbf{\#Instances} & \textbf{\#Features}  \\ \hline
% \textbf{Iris}  &  100  &  4 \\ \hline
%% \textbf{Ionosphere} & 351 & 34  \\ \hline
%% \textbf{Sonar} & 208 & 60 &  \textrm{Binary}\\ \hline
%% \textbf{Breast Cancer}& 569 & 30 & \textrm{Binary} \\ \hline
%% \textbf{LSVT} & 125 & 310 & \textrm{Binary} \\ \hline
% \textbf{Madelon} & 4400 & 500  \\ \hline
% \textbf{Bank} & 45210 & 16  \\ \hline
%\end{tabular}
%\end{table}

\subsection{Empirical Evaluations}
\label{attackresults}
% In this section, we aim to 
% \begin{table}[!htp]
% \centering
% \caption{Selected UCI datasets for regression.}
% \label{data_reg}
% \begin{tabular}{|c|c|c|c|}
% \hline
%  \textbf{Dataset} & \textbf{\#Samples} & \textbf{\#Feature} & \textbf{Year} \\ \hline
%  \textbf{Diabetes} &  442 & 10 &  \\ \hline
% % \textbf{Parkinson Speech} & 1040 & 26 & 2014 \\ \hline
% % \textbf{Breast Cancer} &  568 & 32 & 1995 \\ \hline
%  \textbf{GeoOrig} & 1059 & 68 & 2014 \\ \hline
% % \textbf{TomHardWare} &  28197 & 96 &  2013 \\ \hline
% % \textbf{Community \& Crime} & 1993 & 121 & 2009 \\ \hline
%  \textbf{UJIIndoor} & 19937 & 529 & 2014 \\ \hline
% \end{tabular}
% \end{table}

\begin{figure*}[!t]
\vspace{-4mm}
\center
\subfloat[Diabetes]{\includegraphics[width=0.30\textwidth]{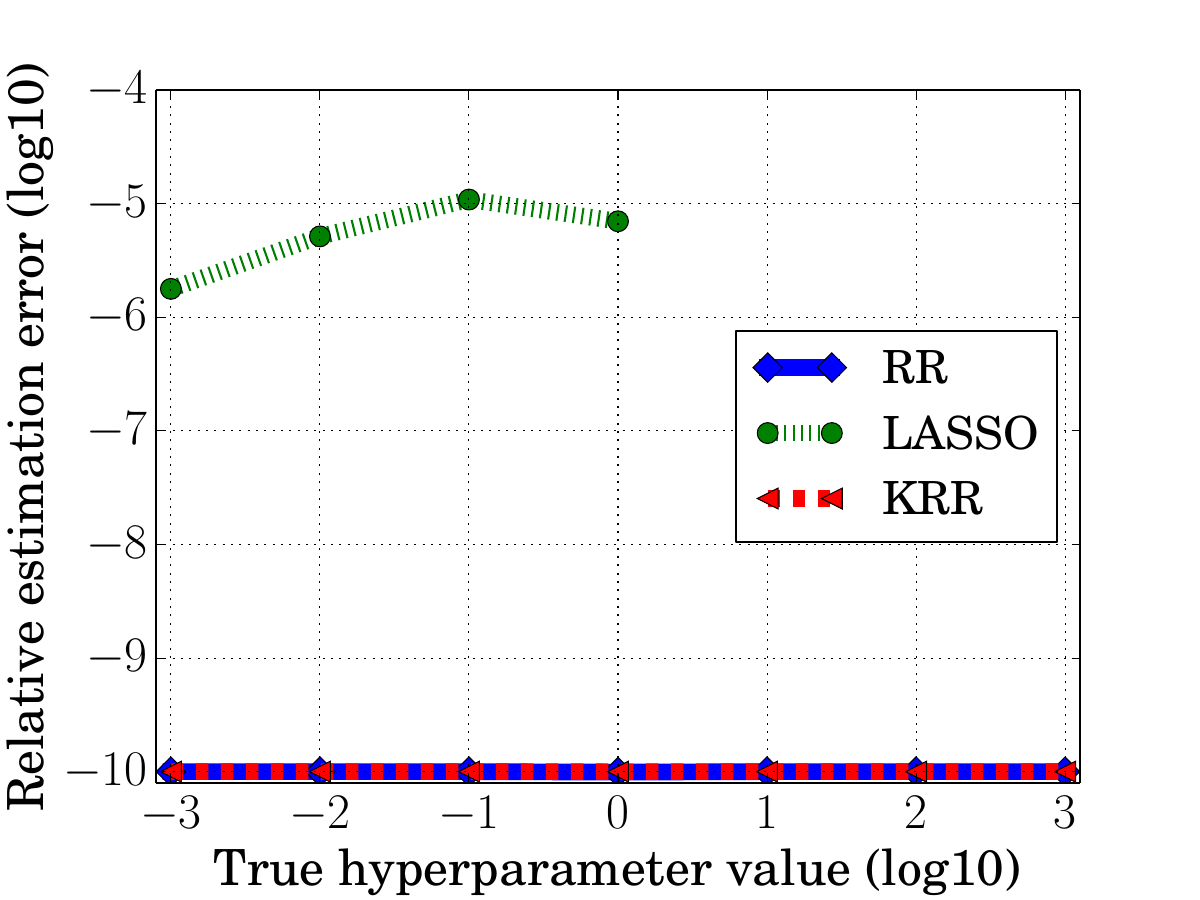}\label{reg_diabetes}} 
\subfloat[GeoOrigin]{\includegraphics[width=0.30\textwidth]{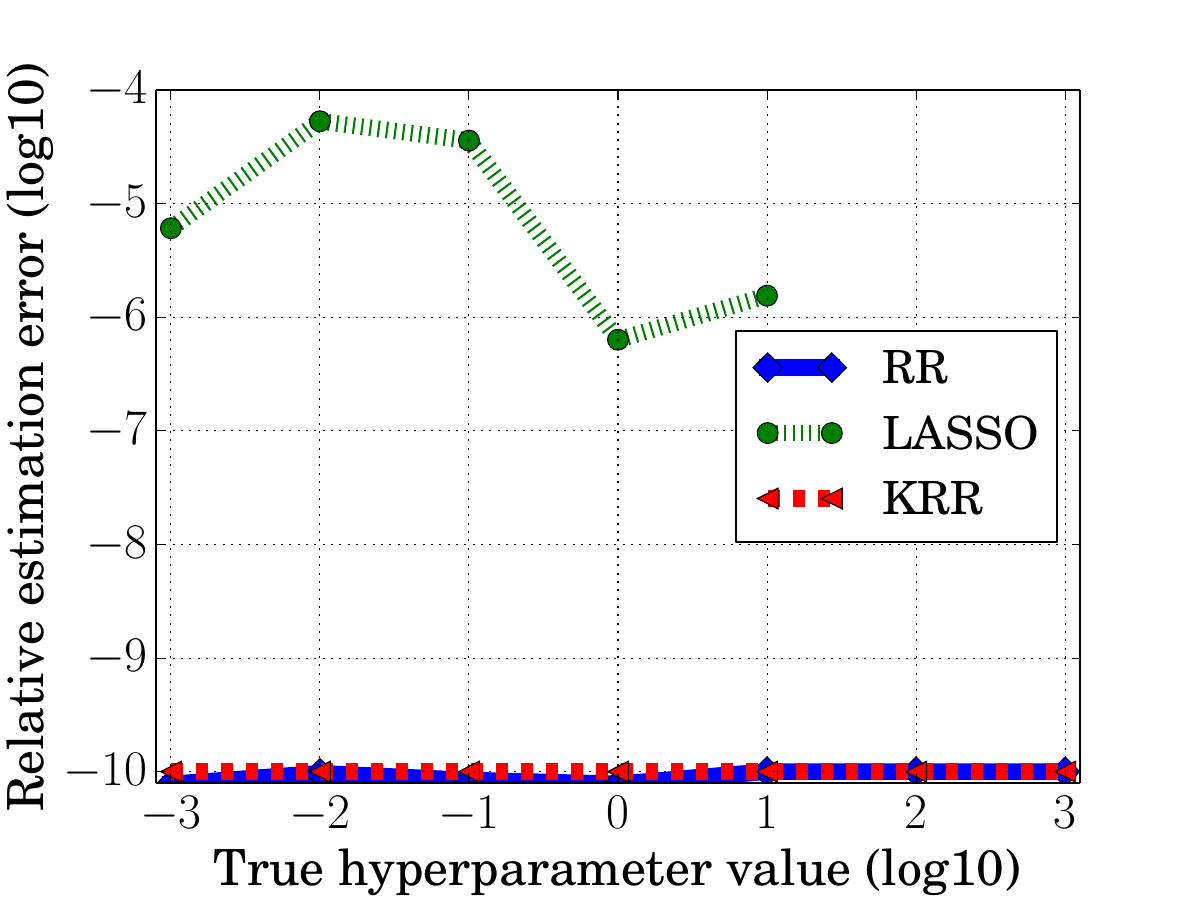}\label{reg_geoorigin}}
\subfloat[UJIIndoor]{\includegraphics[width=0.30\textwidth]{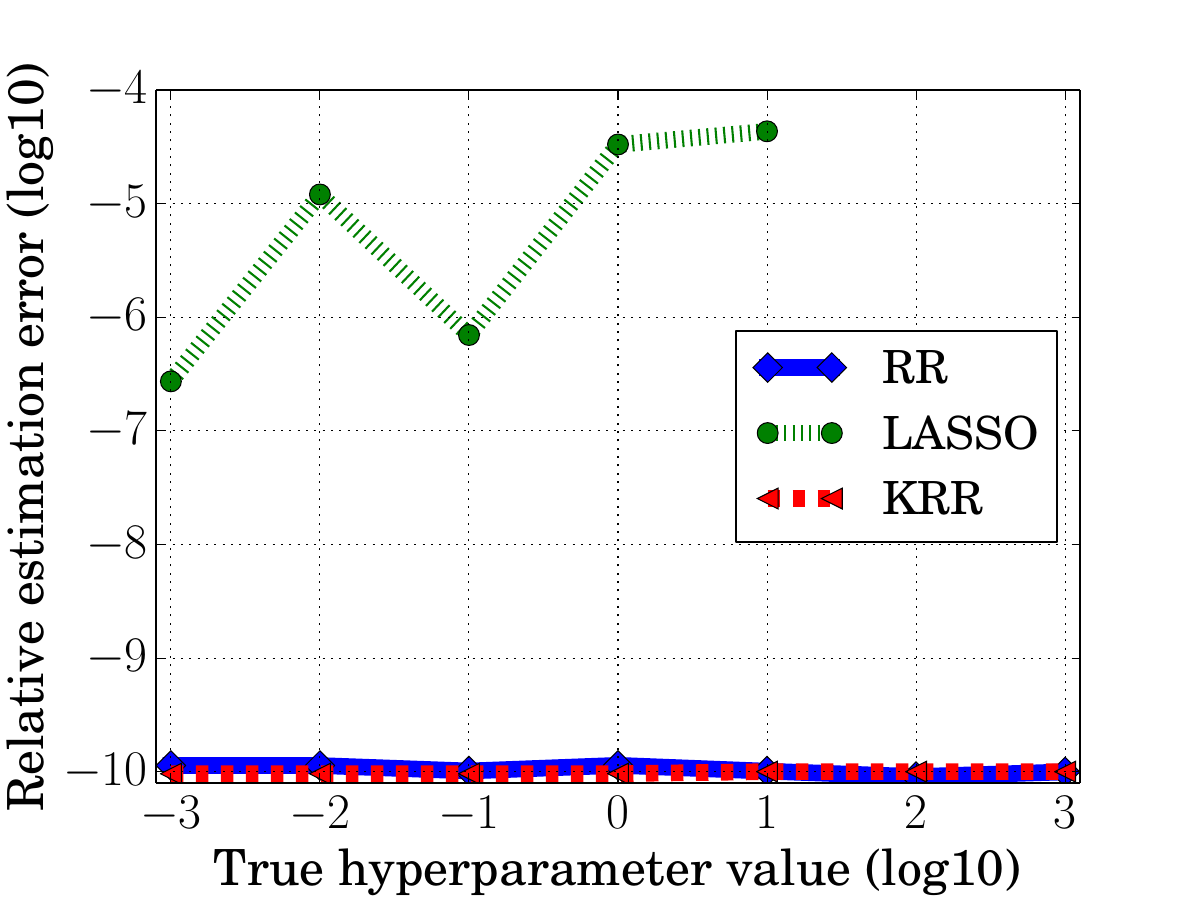}\label{ujiindoor}}
\caption{Effectiveness of our hyperparameter stealing attacks for regression algorithms. %(Change y-axis to be``Relative estimation error (log10)")
}
\label{reg_res_attack}
\vspace{-4mm}
\end{figure*}

\begin{figure*}[t]
\center
\subfloat[Iris]{\includegraphics[width=0.30\textwidth]{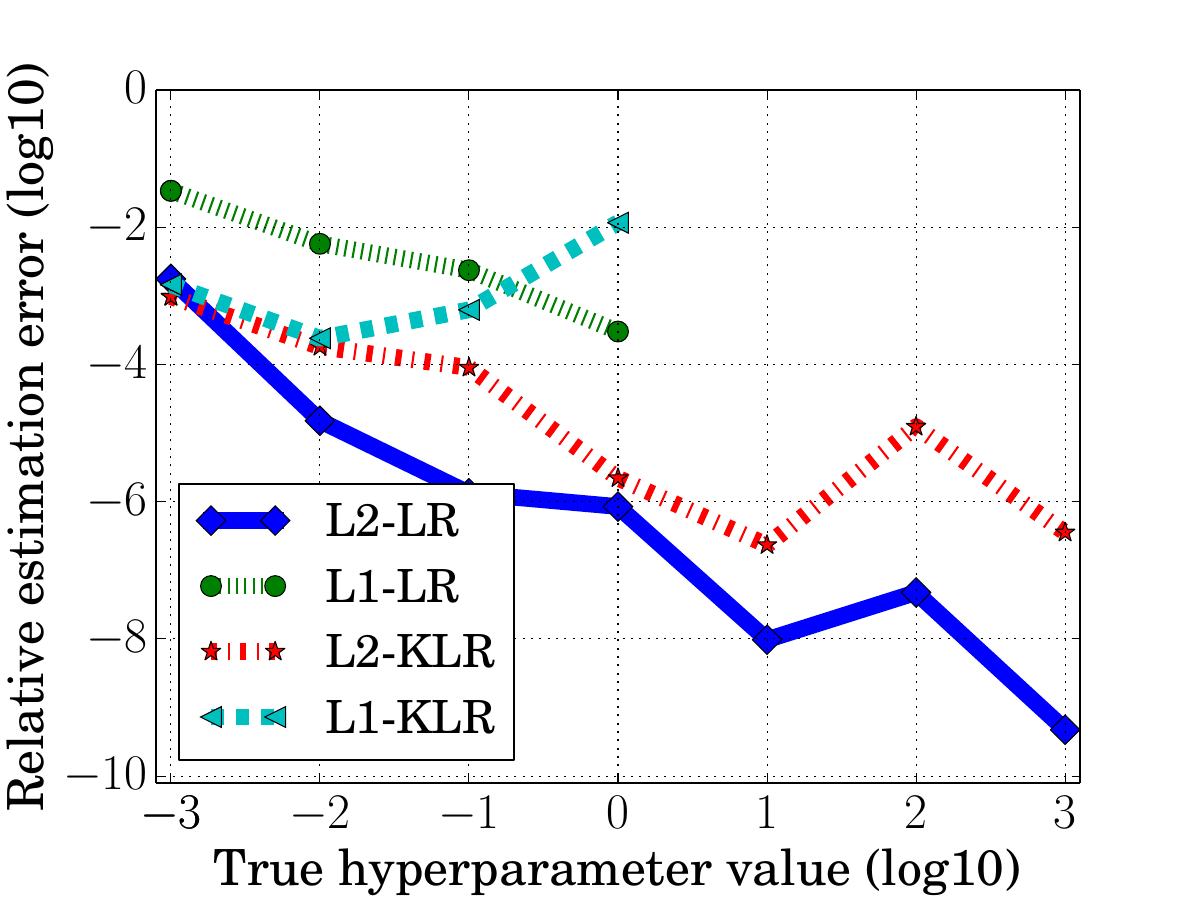}\label{clf_lr_iris}} 
\subfloat[Madelon]{\includegraphics[width=0.30\textwidth]{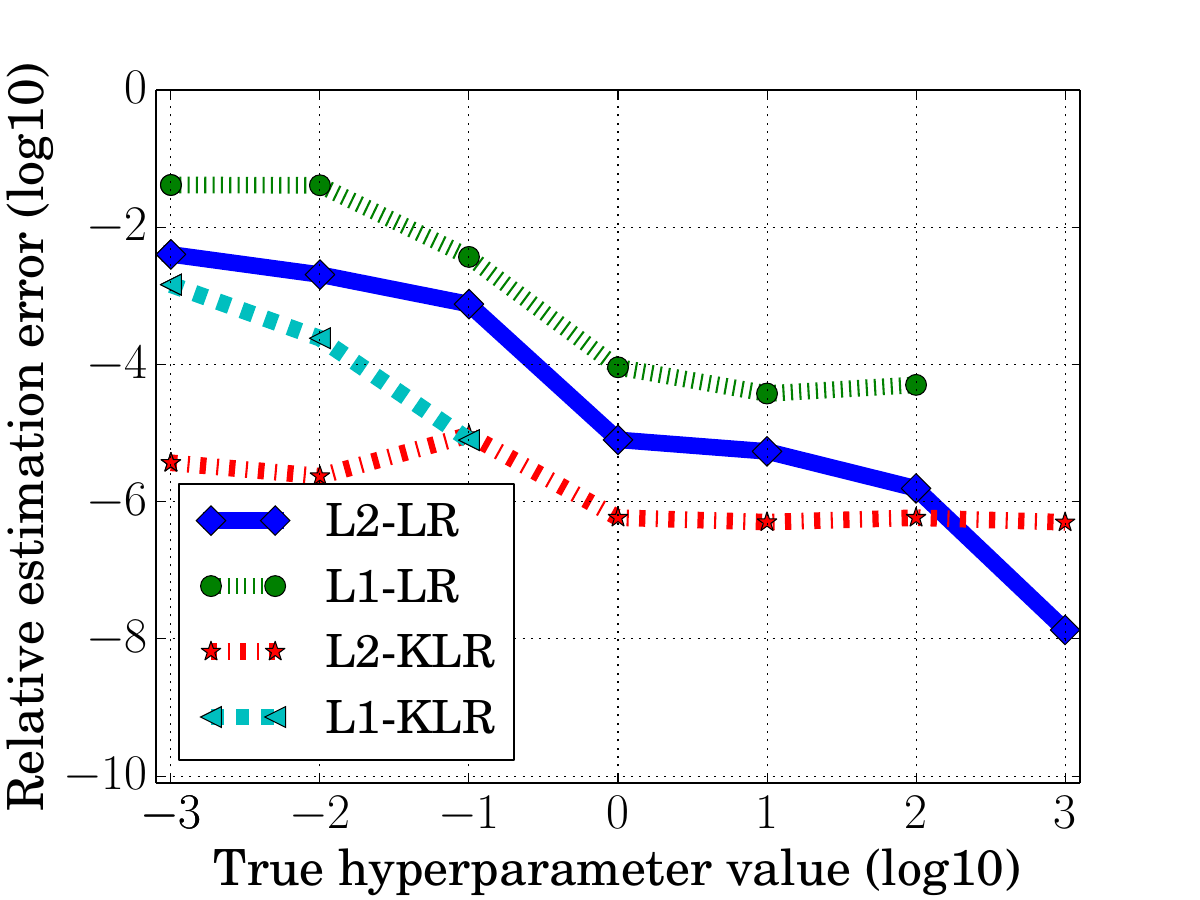}\label{clf_lr_madelon}}
\subfloat[Bank]{\includegraphics[width=0.30\textwidth]{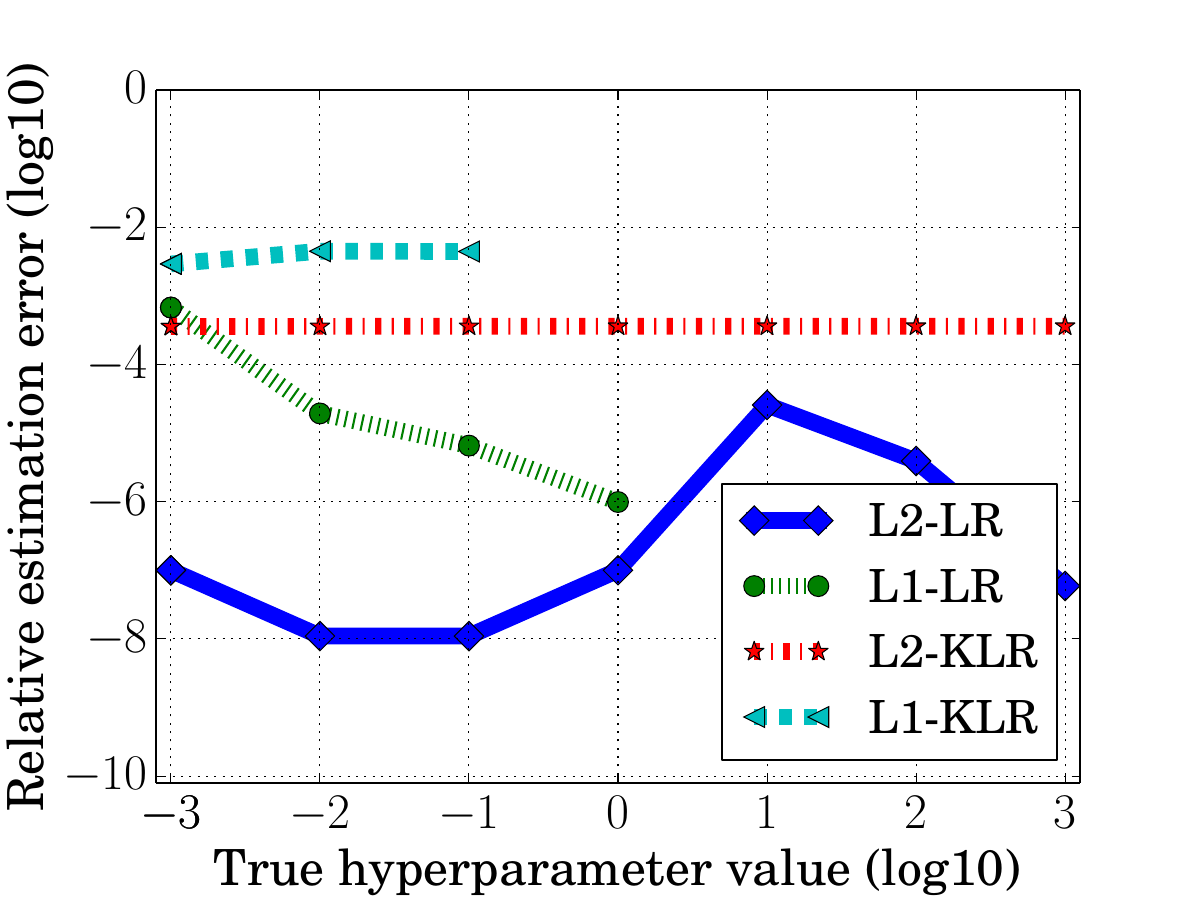}\label{clf_lr_bank}}
%\subfloat[Digits]{\includegraphics[width=0.32\textwidth]{./src_attack/figs/clf/result-clf-mlr-digits-new-final.pdf}\label{clf_mlr_digits}} 
\caption{Effectiveness of our hyperparameter stealing attacks for logistic regression classification algorithms.}
\label{clf_res_attack}
\vspace{-4mm}
\end{figure*}

\begin{figure*}[t]
\center
\subfloat[Iris]{\includegraphics[width=0.30\textwidth]{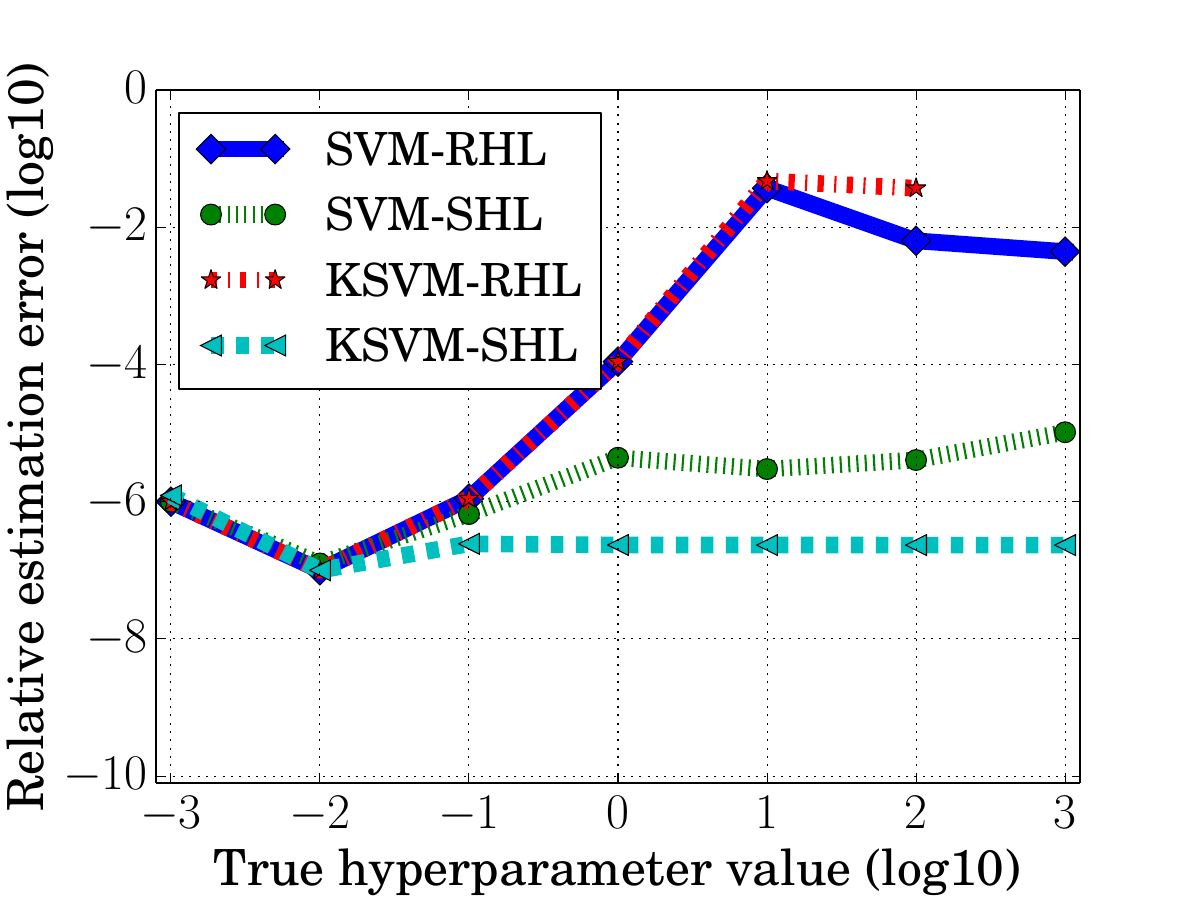}\label{clf_svc_iris}} 
\subfloat[Madelon]{\includegraphics[width=0.30\textwidth]{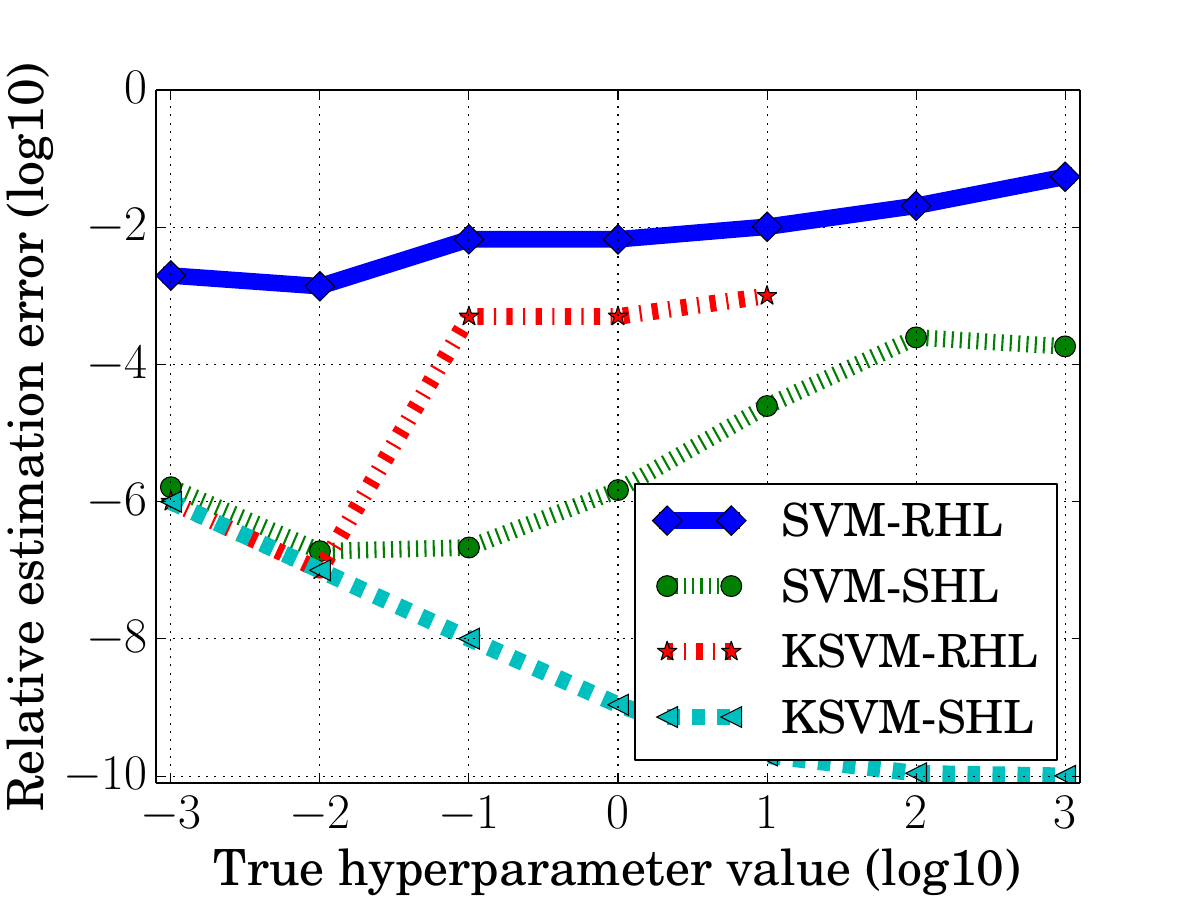}\label{clf_svc_madelon}}
\subfloat[Bank]{\includegraphics[width=0.30\textwidth]{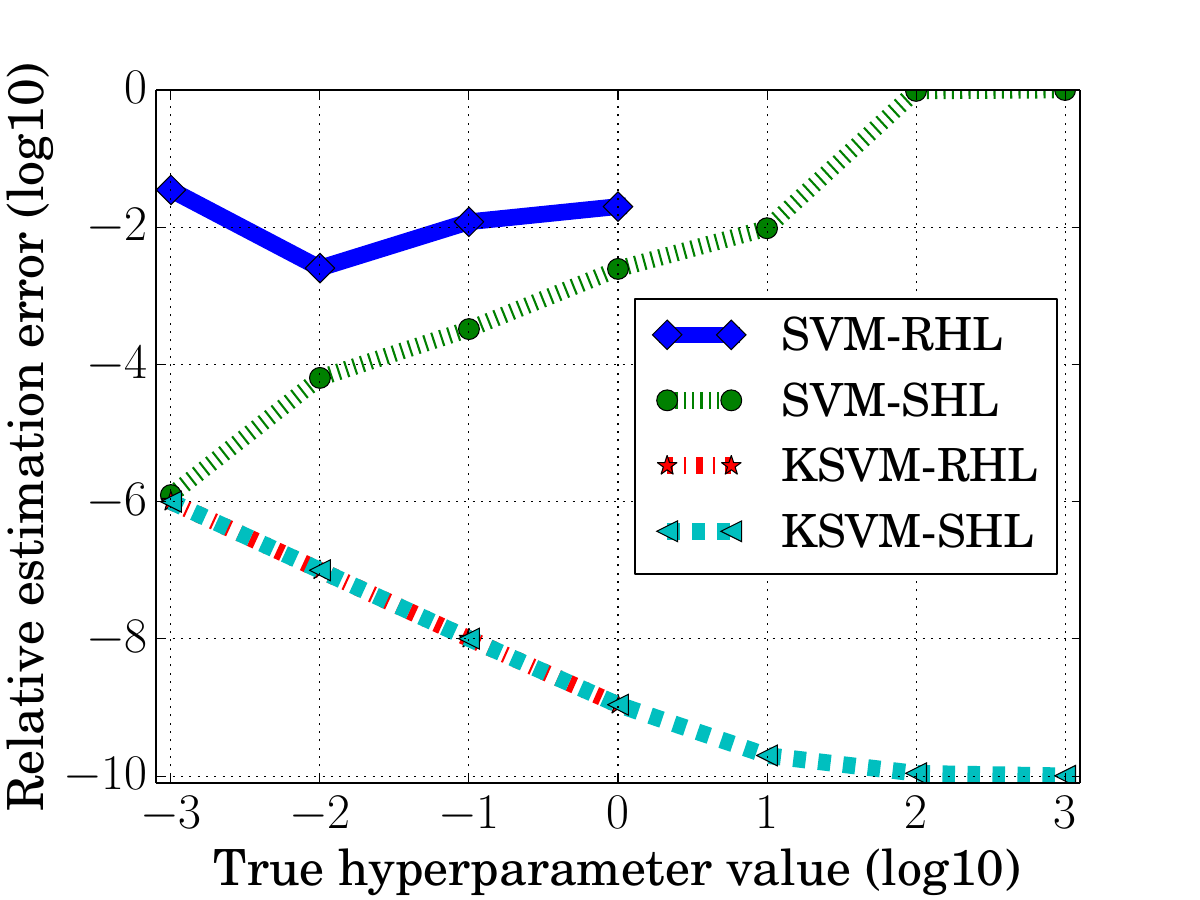}\label{clf_svc_bank}} 
%\subfloat[Digits]{\includegraphics[width=0.32\textwidth]{./src_attack/figs/clf/result-clf-mlr-digits-new-final.pdf}\label{clf_mlr_digits}} 
\caption{Effectiveness of our hyperparameter stealing attacks for SVM classification algorithms.}
\label{clf_res_attack_svm}
\vspace{-4mm}
\end{figure*}

\subsubsection{Experimental Setup}
\label{empiricalAttack}

%\myparatight{Datasets} 
We use several real-world datasets to evaluate the effectiveness of our hyperparameter stealing attacks on the machine learning algorithms we studied.   
We obtained these datasets from the UCI Machine Learning Repository,\footnote{https://archive.ics.uci.edu/ml/datasets.html} and their statistics are summarized in Table~\ref{data_reg}. We note that our datasets have significantly different number of instances and features, which represent different application scenarios. 
We use each dataset as a training dataset.  

% range of \#samples and \#features cover several orders of magnitude. Our purpose is to have a comprehensive understand of each model. Due to space limitation, one can refer to the website for detailed descriptions of the datasets. 

%In this experiment, we evaluate our subgradient-based hyperparameter stealing attack on all aforementioned regression models, i.e., RR, LASSO, ENet, KRR; LR-based classification models, i.e., L2-LR, L1-BLR, L2-BKLR, L1-BKLR, L2-MLR, L2-MKLR; and SVC-based classification models, i.e., BLSVC-RHL, SVM-SHL, BKSVC-RHL, BKSVC-SHL, on several real-world datasets. The selected datasets for regression and classification are from UCI Machine Learning Repository\footnote{https://archive.ics.uci.edu/ml/datasets.html}, and their statistics are summarized in Table~\ref{data_reg} and Table~\ref{data_clf}, respectively. One should note that, we intentionally choose datasets whose range of \#samples and \#features cover several orders of magnitude. Our purpose is to have a comprehensive understand of each model. Due to space limitation, one can refer to the website for detailed descriptions of the datasets. 

\myparatight{Implementation} We use the \emph{scikit-learn} package~\cite{pedregosa2011scikit}, which implements various machine learning algorithms, to learn model parameters.  %This package implements various machine learning algorithims. 
%models, with the optimization algorithms listed in Table~\ref{sum_ml}. 
{All experiments are conducted on a laptop with a  2.7GHz CPU and 8GB memory.} 
We predefine a set of hyperparameters which span over a wide range, i.e., $10^{-3}, 10^{-2}, 10^{-1}, 10^{0}, 10^{1}, 10^{2}, 10^{3}$, in order to evaluate the effectiveness of our attacks for a wide range of hyperparameters. Note that $\lambda>0$, so we do not explore negative values for $\lambda$.  For each hyperparameter and for each learning algorithm, we learn the corresponding model parameters using the scikit-learn package. 
\alan{For kernel algorithms, we use the Gaussian kernel, where the parameter $\sigma$ in the kernel is set to be 10.}
%we can throughly verify the accurateness of our hyperparameter stealing attack. Given the training samples and optimal parameters, we then realize our attack, written in Python, on all models. 
 We implemented our attacks in Python. 

\myparatight{Evaluation metric} We evaluate the effectiveness of our attacks using \emph{relative estimation error}, which is formally defined as follows:
%\emph{attack capability} in terms of $\epsilon$-hyperparameter learnable, with $\epsilon$-error defined as
\begin{align}
%\footnotesize
\label{rer}
\text{\bf Relative estimation error: } \epsilon =  \frac{|\hat{\lambda}-\lambda|}{\lambda},
\end{align}
where $\hat{\lambda}$ and $\lambda$ are the estimated hyperparameter and true hyperparameter, respectively. 
% (it can be $\lambda_j$, $\tau$, or $T$) is the predefined hyperparameter for $j$-th regularization; $\hat{h}_j$ is the average estimated value of $\hat{h}_j^{(i)}$ of $w_i$ (or $\alpha_i$) for $j$-th regularization. 

\alan{
\subsubsection{Experimental Results for Known Model Parameters}

% \begin{figure}
% \center
% \subfloat[Iris]{\includegraphics[width=0.23\textwidth]{./src_attack/figs/clf/result-clf-lr-iris-new-final.pdf}\label{clf_lr_iris}} 
% \subfloat[Ionosphere]{\includegraphics[width=0.23\textwidth]{./src_attack/figs/clf/result-clf-lr-ionosphere-new-final.pdf}\label{clf_lr_ionosphere}}
% %\subfloat[Madelon]{\includegraphics[width=0.32\textwidth]{./src_attack/figs/clf/result-clf-lr-madelon-new-final.pdf}\label{clf_lr_madelon}} 
% \caption{Hyperparameter stealing results of binary LR-based classification models.}
% \label{clf_blr_res_attack}
% \end{figure}

% \begin{figure}
% %\center
% \label{clf_mlr_res_attack}
% \subfloat[Iris]{\includegraphics[width=0.48\textwidth]{./src_attack/figs/clf/result-clf-mlr-iris-new-final.pdf}\label{clf_mlr_iris}} 
% \subfloat[Digits]{\includegraphics[width=0.48\textwidth]{./src_attack/figs/clf/result-clf-mlr-digits-new-final.pdf}\label{clf_mlr_digits}}
% \caption{Hyperparameter stealing results of multi-class LR-based classification models.}
% \end{figure}

% \begin{figure}
% \center
% \label{clf_blsvc_res_attack}
% \subfloat[Iris]{\includegraphics[width=0.23\textwidth]{./src_attack/figs/clf/result-clf-svc-iris-new-final.pdf}\label{clf_svc_iris}} 
% \subfloat[Ionosphere]{\includegraphics[width=0.23\textwidth]{./src_attack/figs/clf/result-clf-svc-ionosphere-new-final.pdf}\label{clf_svc_ionosphere}}
% %\subfloat[Madelon]{\includegraphics[width=0.32\textwidth]{./src_attack/figs/clf/result-clf-svc-madelon-new-final.pdf}\label{clf_svc_madelon}} 
% \caption{Hyperparameter stealing results of binary SVC-based classification models.}
% \end{figure}

We first show results for the scenario where an attacker knows the training dataset, the learning algorithm, and the model parameters.
Figure~\ref{reg_res_attack} shows the relative estimation errors for different regression algorithms on the regression datasets. Figure~\ref{clf_res_attack}  shows the results for 
logistic regression algorithms on the classification datasets. Figure~\ref{clf_res_attack_svm}  shows the results for 
SVM algorithms on the classification datasets. 
Figure~\ref{clf_res_attack_nn} shows the results for three-layer neural networks for regression and classification. In each figure, x-axis represents the true hyperparameter in a particular algorithm, and the y-axis represents the relative estimation error of our attacks at stealing the hyperparameter. For better illustration, we set the relative estimation errors to be $10^{-10}$ when they are smaller than $10^{-10}$.}
Note that learning algorithms with $L_1$ regularization require the hyperparameter to be smaller than a maximum value  $\lambda_{\max}$ in order to learn meaningful model parameters (please refer to Appendix~\ref{app:other} for more details). 
%$\lambda_{\max}$ is determined by the training dataset (please refer to Appendix~\ref{app:other} for more details). 
Therefore, in the figures, the data points are missing for such algorithms when the hyperparameter gets larger than $\lambda_{\max}$, which is different for different training datasets and algorithms. We didn't show results on \emph{kernel LASSO} because it is not widely used. Moreover, we didn't find open-source implementations to learn model parameters in kernel LASSO, and implementing kernel LASSO is out of the scope of this work. However, our attacks are applicable to kernel LASSO. % because its objective function is almost everywhere differentiable. 

\begin{figure}[t]
\center
\vspace{-3mm}
\subfloat[Regression]{\includegraphics[width=0.25\textwidth]{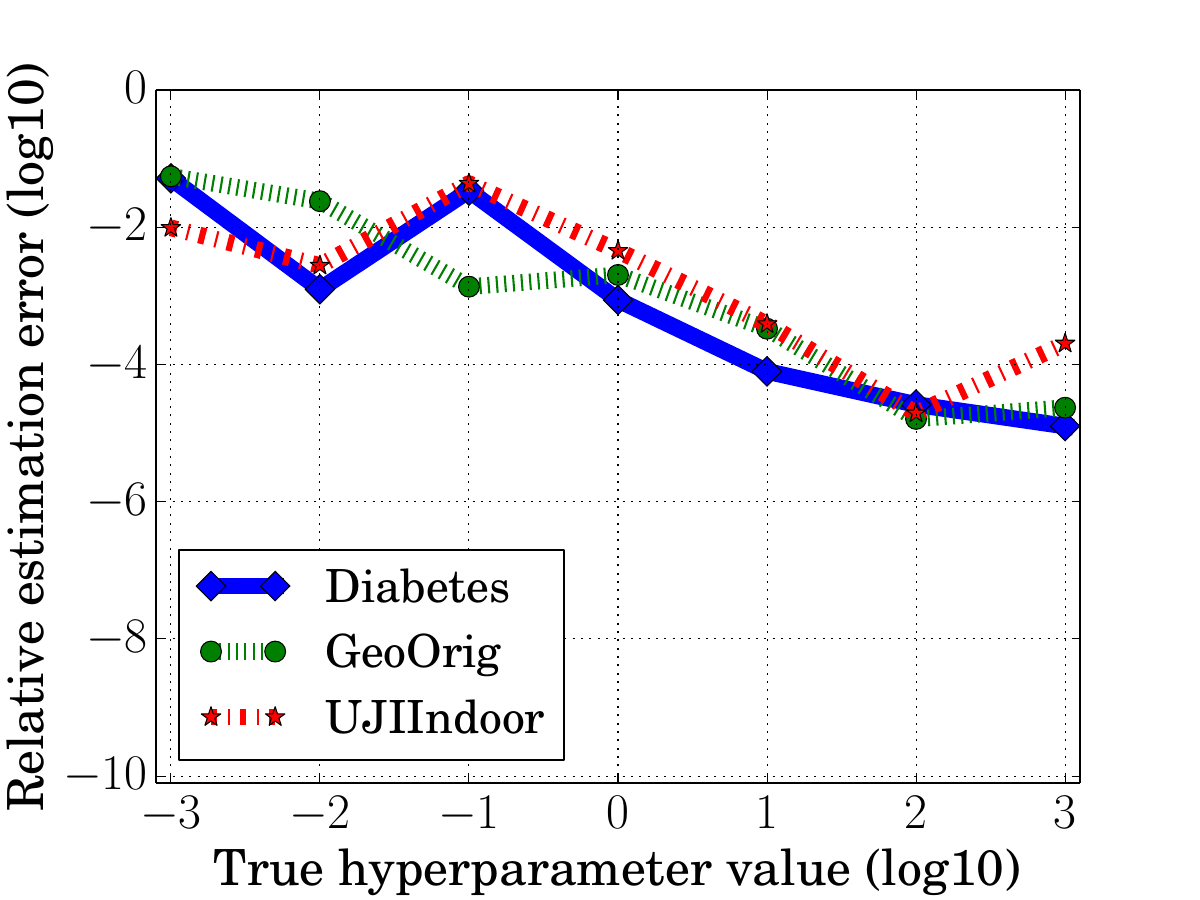}\label{nn_reg_attack}} 
\subfloat[Classification]{\includegraphics[width=0.25\textwidth]{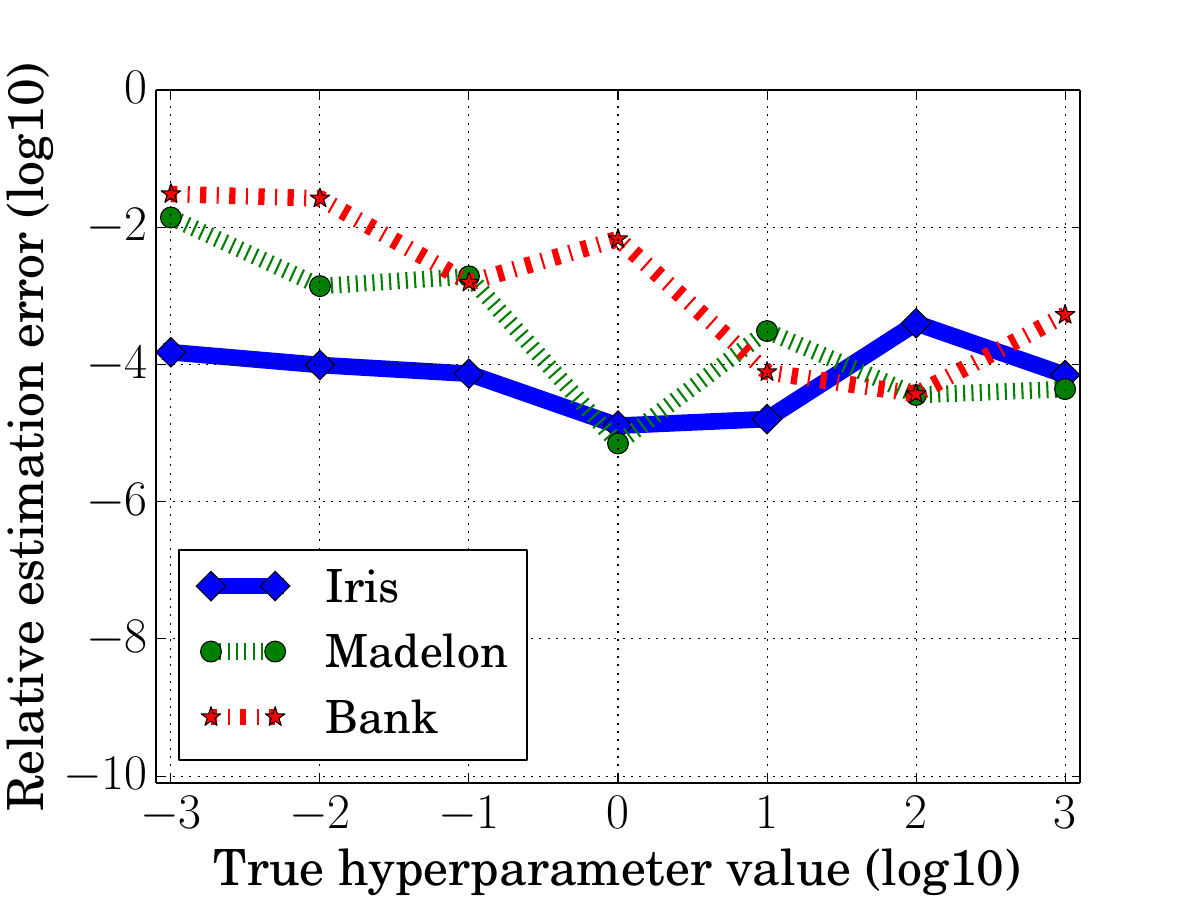}\label{nn_clf_attack}}
\caption{Effectiveness of our hyperparameter stealing attacks for a) a three-layer neural network regression algorithm and b) a three-layer neural network classification algorithm.}
\label{clf_res_attack_nn}
\vspace{-5mm}
\end{figure}

\begin{figure}[!t]
%\vspace{-6mm}
\center
%\subfloat[Diabetes]
{\includegraphics[width=0.28\textwidth]{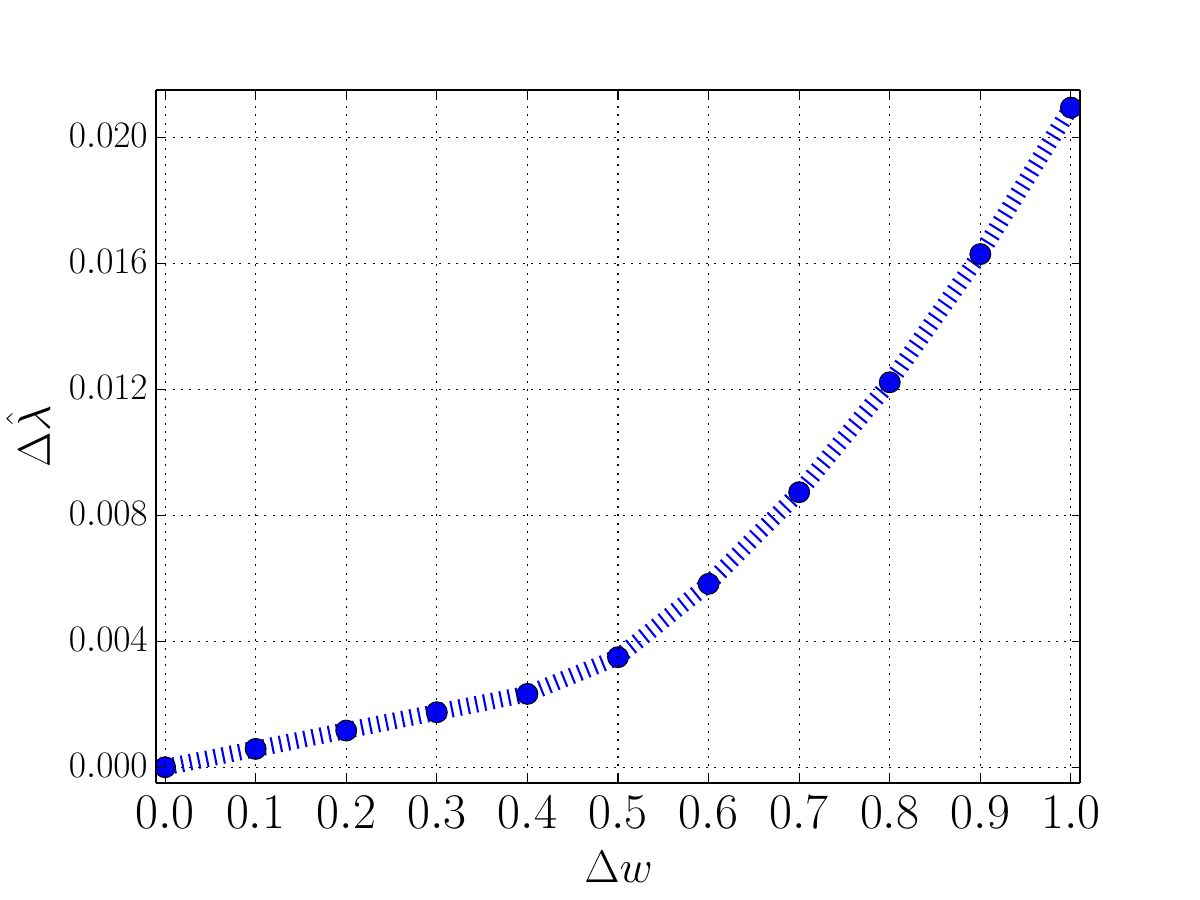}}
%\subfloat[Digits]{\includegraphics[width=0.32\textwidth]{./src_defense/figs/clf/result-clf-mlr-digits-new-new.pdf}\label{clf_mlr_digits_defense}} 
\caption{\alan{Effectiveness of our hyperparameter stealing attacks for RR when the model parameters deviate from the optimal ones.}}
\vspace{-6mm}
\label{rroptimal}
\end{figure}

We have two key observations.
First, our attacks can accurately estimate the hyperparameter for all learning algorithms we studied and for a wide range of hyperparameter values. Second, we observe that our attacks can more accurately estimate the hyperparameter for Ridge Regression (RR) and Kernel Ridge Regression (KRR) than for other learning algorithms. This is because RR and KRR have \emph{analytical solutions} for model parameters, and thus the learnt model parameters are the exact minima of the objective functions. In contrast, other learning algorithms we studied do not have analytical solutions for model parameters, and their learnt model parameters are relatively further away from the corresponding minima of the objective functions.  Therefore, our attacks have larger estimation errors for these learning algorithms. 

\alan{In practice, a learner may use approximate solutions for RR and KRR because computing the exact optimal solutions may be computationally expensive. We evaluate the impact of such approximate solutions on the accuracy of hyperparameter stealing attacks, and compare the results with those predicted by our Theorem~\ref{thm_2}. Specifically, we use the RR algorithm, adopt the Diabetes dataset, and set the true hyperparameter to be 1. We first compute the optimal model parameters for RR.  Then, we modify a model parameter by $\Delta w$  
%randomly sample one model parameter, modify it by $\Delta w$, 
and estimate the hyperparameter by our attack. %We modify the model parameters by the same amount $\Delta w$, because it is hard to visualize the impact of  different $\Delta w$ for different model parameters.  
Figure~\ref{rroptimal} shows the estimation error $\Delta \hat{\lambda}$ as a function of $\Delta w$ (we show the absolute estimation error instead of relative estimation error in order to compare the results with Theorem~\ref{thm_2}). We observe that when $\Delta w$ is very small, the estimation error $\Delta \hat{\lambda}$ is a linear function of  $\Delta w$. As $\Delta w$ becomes larger, $\Delta \hat{\lambda}$ increases quadratically with $\Delta w$. Our observation is consistent with Theorem~\ref{thm_2}, which shows that the estimation error is linear to the difference between the learnt model parameters and the minimum closest to them when the difference is very small. }

\subsubsection{Experimental Results for Unknown Model Parameters} 
\label{unknownparameter}

Our hyperparameter stealing attacks are still applicable when the model parameters are unknown to an attacker, e.g., for black-box MLaaS platforms such as Amazon Machine Learning.  Specifically, 
%Specifically, various MLaaS platforms (e.g., Amazon Machine Learning, Microsoft Azure Machine Learning, BigML, and Google Cloud Platform) return confidence scores for predictions. 
the attacker can first use the \emph{equation-solving-based} model parameter stealing attacks proposed in~\cite{tramer2016stealing} to learn the model parameters and then perform our hyperparameter stealing attacks. 
% for MLaaS platforms that return confidence scores for predictions, e.g., Amazon Machine Learning, Microsoft Azure Machine Learning, BigML, and Google Cloud Platform. Specifically,
%\alan{when MLaaS platforms return confidence values (which are supported by most MLaaS platforms, e.g., Amazon Web Services, Microsoft Azure Machine Learning, BigML, and Google Cloud Platform)}, the attacker can first use \emph{equation-solving-based} model parameter stealing attacks proposed in~\cite{tramer2016stealing} to learn the model parameters and then perform our hyperparameter stealing attacks. 
Our Theorem~\ref{thm_2} bounds the estimation error of hyperparameters with respect to the difference between the stolen model parameters and the closest minimum of the objective function of the ML algorithm.

We also empirically evaluate the effectiveness of our attacks when model parameters are unknown. For instance, Figure~\ref{modelparameterunknown} shows the relative estimation errors of hyperparameters for regression algorithms and classification algorithms, when the model parameters are unknown but stolen by the model parameter stealing attacks~\cite{tramer2016stealing}. For simplicity, we only show results on the Diabetes dataset for regression algorithms and on the Iris dataset for classification algorithms, but results on other datasets are similar. Note that LASSO requires the hyperparameter to be smaller than a certain threshold as we discussed in the above, and thus some data points are missing for LASSO.
We find that we can still accurately steal the hyperparameters. The reason is that the model parameter stealing attacks can accurately steal the model parameters.

\alan{
\subsubsection{Summary}
Via empirical evaluations, we have the following observations. First, our attacks can accurately estimate the hyperparameter for all ML algorithms we studied. Second, our attacks can more accurately estimate the hyperparameter for ML algorithms that have analytical solutions of the model parameters. Third, via combining with model parameter stealing attacks, our attacks can accurately estimate the hyperparameter even if the model parameters are unknown.
}
%\begin{packeditemize}
%\item Our attacks can accurately estimate the hyperparameter for all ML algorithms we studied. 
%\item Our attacks can more accurately estimate the hyperparameter for ML algorithms that have analytical solutions of the model parameters.  
%\item Our attacks can accurately estimate the hyperparameter in the scenario where model parameters are unknown, via combining with model parameter stealing attacks. 
%\end{packeditemize}

%The hyperparameter stealing results of multi-class LR-based classification models on Iris and Digits, are shown in Fig.~\ref{clf_mlr_res_attack}. 

%The hyperparameter stealing results of binary SVC-based classification models on Iris, Ionosphere, and Madelon, are shown in Fig.~\ref{clf_blsvc_res_attack}.

%\subsection{Real-World Deployment}

\begin{figure}[!t]
\vspace{-3mm}
\center
\subfloat[Diabetes]{\includegraphics[width=0.25\textwidth]{./src_attack/figs/reg/result-reg-diabetes-final.pdf}\label{reg_diabetes}}
%\subfloat[Digits]{\includegraphics[width=0.32\textwidth]{./src_defense/figs/clf/result-clf-mlr-digits-new-new.pdf}\label{clf_mlr_digits_defense}} 
\subfloat[Iris]{\includegraphics[width=0.25\textwidth]{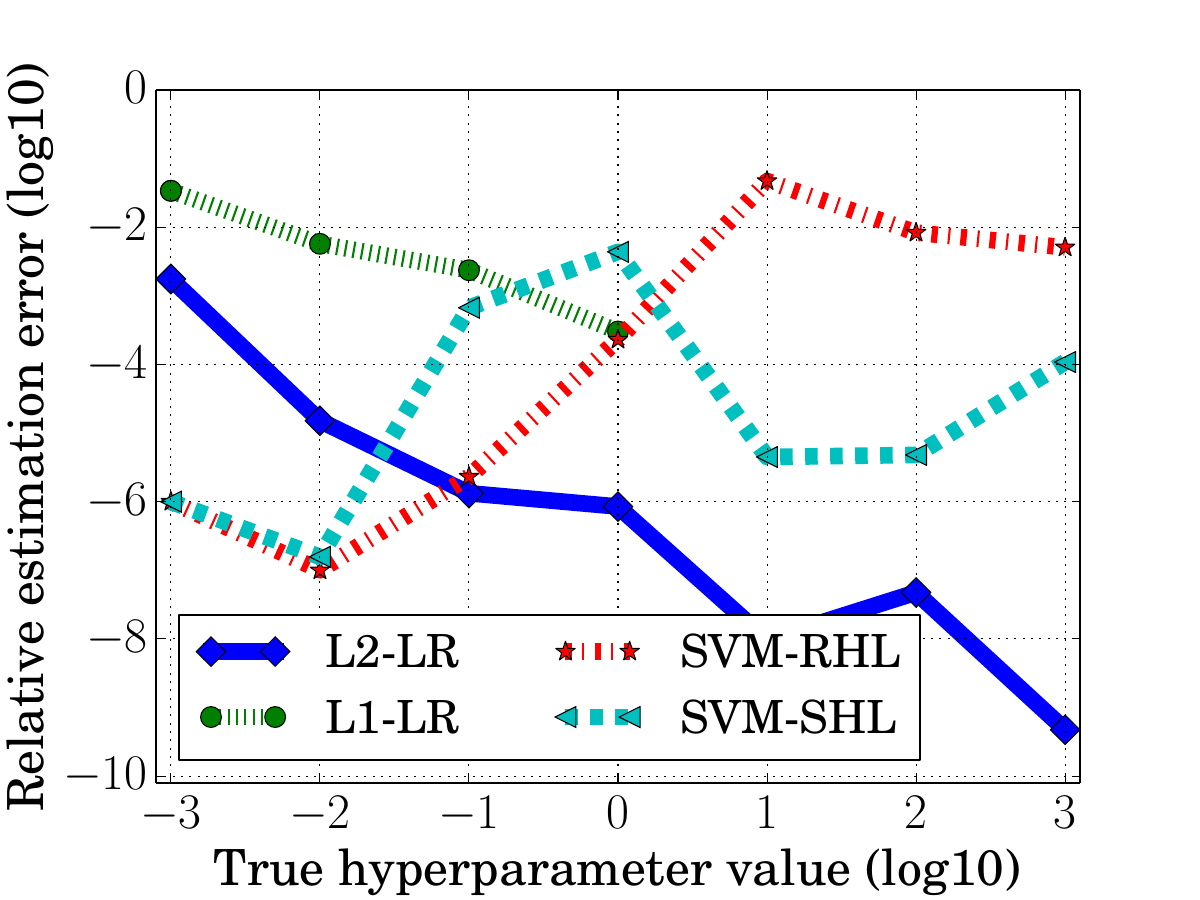}\label{clf_diabetes}}
\caption{Effectiveness of our hyperparameter stealing attacks when model parameters are unknown but stolen by model parameter stealing attacks. (a)  Regression algorithms on Diabetes and (b) Classification algorithms on Iris.}
\vspace{-3mm}
\label{modelparameterunknown}
\end{figure}

\subsection{Implications for MLaaS}
\label{eval_attack}

% \begin{figure}
% \center
% \subfloat[M1 vs. M3]{\includegraphics[width=0.23\textwidth]{./src_time/figs/result-CovType-lr-acc-time-m1-m3.pdf}\label{lr_covtype_m1_m3}} 
% \subfloat[M2 vs. M3]{\includegraphics[width=0.23\textwidth]{./src_time/figs/result-CovType-lr-acc-time-m2-m3.pdf}\label{lr_covtype_m2_m3}}
% \label{lr_covtype}
% \caption{Relative ACC error and speedup of L2-LR on CovType.}
% \end{figure}

% \begin{figure}
% \center
% \subfloat[M1 vs. M3]{\includegraphics[width=0.23\textwidth]{./src_time/figs/result-CovType-svc-acc-time-m1-m3.pdf}\label{svc_covtype_m1_m3}} 
% \subfloat[M2 vs. M3]{\includegraphics[width=0.23\textwidth]{./src_time/figs/result-CovType-svc-acc-time-m2-m3.pdf}\label{svc_covtype_m2_m3}}
% \label{svc_covtype}
% \caption{Relative ACC error and speedup of SVM-SHL on CovType.}
% \end{figure}
We show that a user can use our hyperparamter stealing attacks to learn an accurate model through a machine-learning-as-a-service (MLaaS) platform with much less costs.
% MLaaS~\cite{BigML,tramer2016stealing,amazon,google,ms} is an emerging technology to aid users, who have limited computing power and machine learning expertise, to learn a machine learning model over large datasets. A user gives a training dataset to a MLaaS platform, which learns a model for the user. 
%The user is often charged according to the amount of computation that the MLaaS platform performed to learn the model~\cite{}. 
%The user can use the model to make predictions locally or monetize it through exposing APIs for other users to make queries~\cite{}.  
%We demonstrate that a user can use our hyperparamter stealing attacks to help learn an accurate model through a MLaaS platform with much less costs.
While different MLaaS platforms have different paradigms, we consider an MLaaS platform (e.g., Amazon Machine Learning~\cite{amazon}, Microsoft Azure Machine Learning~\cite{ms}) that charges a user  according to the amount of computation that the MLaaS platform performed to learn the model, and supports two protocols for a user to learn a model. In the first protocol (denoted as \emph{Protocol I}), the user uploads a training dataset to the MLaaS platform and specifies a learning algorithm;  the MLaaS platform learns the hyperparameter using proprietary algorithms and learns the model parameters with the learnt hyperparameter; 
\alan{and then (optionally) the model parameters are sent back to the user. When the model parameters are not sent back to the user, the MLaaS is called black-box.} 
 The MLaaS platform (e.g., a black-box platform) does not share the learnt hyperparameter value with the user considering intellectual property and algorithm confidentiality. 
 
In Protocol I, learning the hyperparameter is often the most time-consuming and costly part, because it more or less involves cross-validation. In practice, some users might already have appropriate hyperparameters through domain knowledge. Therefore, the MLaaS platform provides a second protocol (denoted as \emph{Protocol II}), in which the user uploads a training dataset to the MLaaS platform, defines a hyperparameter value, and specifies a learning algorithm, and then the MLaaS platform produces the model parameters for the given hyperparameter. \emph{Protocol II} helps users learn models with less economical costs when they already have good hyperparameters. {We note that Amazon Machine Learning and Microsoft Azure Machine Learning support the two protocols.}

\subsubsection{Learning an Accurate Model with Less Costs}
We demonstrate that a user can use our hyperparameter stealing attacks to learn a model through MLaaS with much less economical costs without sacrificing model performance. In particular, we assume the user does not have a good hyperparameter yet. We compare the following three methods to learn a model through MLaaS. By default, we assume the MLaaS shares the model parameters with the user. If not, the user can use model parameter stealing attacks~\cite{tramer2016stealing} to steal them.  

\myparatight{Method 1 (M1)} The user leverages \emph{Protocol I} supported by the MLaaS platform to learn the model. Specifically, the user uploads the training dataset to the MLaaS platform and specifies a learning algorithm. The MLaaS platform learns the hyperparameter and then learns the model parameters using the learnt hyperparameter. The user then downloads the model parameters. 
%uploads the entire training set to the cloud; The cloud-based ML services train a ML model on the entire training set via 5-fold CV; The user achieves the learnt ML model and measures the performance on the testing set.  

\myparatight{Method 2 (M2)} In order to save economical costs, the user samples $p\%$ of the training  dataset uniformly at random and then uses \emph{Protocol I} to learn a model over the sampled subset of the training dataset. We expect that this method is less computationally expensive than M1, but it may sacrifice performance of the learnt model.  

% to the cloud; The cloud-based ML services train a ML model on the $p\%$ training set via 5-fold CV; The user achieves the learnt ML model and measures the performance on the testing set.

\myparatight{Method 3 (M3)} In this method, the user uses our hyperparameter stealing attacks. Specifically, 
the user first samples $q\%$ of the training dataset uniformly at random. Second, the user learns a model over the sampled training dataset through the MLaaS via \emph{Protocol I}. We note that, for big data, even a very small fraction (e.g., 1\%) of the training dataset could be too large for the user to process locally, so we consider the user uses the MLaaS.  Third, the user estimates the hyperparamter learnt by the MLaaS using our hyperparameter stealing attacks. Fourth, the user re-learns a model over the entire training dataset through the MLaaS via \emph{Protocol II}.  We call this strategy \emph{``Train-Steal-Retrain"}. %We note that the user does not need to know how the MLaaS implemented the machine learning algorithm.

\subsubsection{Comparing the Three Methods Empirically}
We first show simulation results of the three methods. For these simulation results, we assume model parameters are known to the user. 
In the next subsection, we compare the three methods on Amazon Machine Learning, a real-world MLaaS platform. 

% We note that there are no essential differences between performing experiments on a real MLaaS and simulations once the MLaaS satisfies our threat model.
%We didn't perform experiments on a real MLaaS due to budget constraints.
%Due to our limited budget, we didn't perform experiments on a real MLaaS (e.g., Google Cloud Platform). 
%Instead, we simulate the three methods on a laptop. 
%\alan{Actually, there are no essential differences between performing experiments on a real MLaaS and simulated MLaaS once the MLaaS satisfies our threat model.}

\myparatight{Setup} For each dataset in Table~\ref{data_reg}, we randomly split it into two halves,  which are used as the training dataset and the testing dataset, respectively. We consider the MLaaS learns the hyperparameter through 5-fold cross-validation on the training dataset. We measure the performance of the learnt model through \emph{mean square error (MSE)} (for regression models) or \emph{accuracy (ACC)} (for classification models). MSE and ACC are formally defined in Section~\ref{related}. 
Specifically, we use M1 as a baseline; then we measure the \emph{relative MSE (or ACC) error} of M2 and M3 over M1. For example, the relative MSE error of M3 is defined as $\frac{|\text{MSE}_{M3}-\text{MSE}_{M1}|}{\text{MSE}_{M1}}$. Moreover, we also measure the \emph{speedup} of M2 and M3 over M1 with respect to the overall amount of computation required to learn the model. Note that we also include the computation required to steal the hyperparameter for M3.

\myparatight{M3 vs. M1} Figure~\ref{ridge_UJIndoorLoc_m123} compares M3 with M1 with respect to model performance (measured by relative performance of M3 over M1) and speedup as we sample a larger fraction of training dataset (i.e., $q$ gets larger), where the regression algorithm is RR and the classification algorithm is SVM-SHL. Other learning algorithms and datasets have similar results, so we omit them for conciseness.  

We observe that M3 can learn a model that is as accurate as the model learnt by M1, while saving a significant amount of computation. Specifically, for RR on the dataset UJIndoorLoc, when we sample 3\% of training dataset, M3 learns a model that has almost 0 relative MSE error over M1, but M3 is around 8 times faster than M1. This means that the user can learn an accurate model using M3 with much less economic costs, when the MLaaS platform charges the user according to the amount of computation.   For the SVM-SHL algorithm on the Bank dataset, M3 can learn a model that has almost 0 relative ACC error over M1 and is around 15 times faster than M1, when we sample 1\% of training dataset. The reason why M3 and M1 can learn models with similar performances is that learning the hyperparameter using a subset of the training dataset changes it slightly and the learning algorithms are relatively robust to  small variations of the hyperparameter.

Moreover, we observe that the speedup of M3 over M1 is more significant when the training dataset becomes larger. 
%via comparing the speedup results across different datasets in Table~\ref{data_reg} and Table~\ref{data_clf}, 
%For instance, 
%Figure~\ref{speedup_gaussian} shows the speedup of M3 over M1 on training datasets with different sizes, which are synthesized via a Gaussian distribution.
 {Figure~\ref{speedup_gaussian} shows the speedup of M3 over M1 
 on binary-class training datasets with different sizes, where each class is synthesized via a Gaussian distribution with 10 dimensions. Entries of the mean vectors of the two Gaussian distributions are all 1's and all -1's, respectively. Entries of the covariance matrix of the two Gaussian distributions are generated from the standard Gaussian distribution.  
}
% since we aim to explore different dataset sizes. 
We select the parameter $q\%$ in M3 such that the relative ACC error is smaller than 0.1\%, i.e., 
M3 learns a model as accurately as M1.  The speedup of M3 over M1 is more significant as the training dataset gets larger. 
This is because the process of learning the hyperparameter has a computational complexity that is higher than linear. M1 learns the hyperparameter over the entire training dataset, while M3 learns it on a sampled subset of training dataset. As a result, the speedup is more significant for larger training datasets. 
This implies that a user can benefit more by using M3 when the user has a larger training dataset, which is often the case in the era of big data.

\begin{figure}[!t]
\vspace{-3mm}
\center
\subfloat[]{\includegraphics[width=0.25\textwidth]{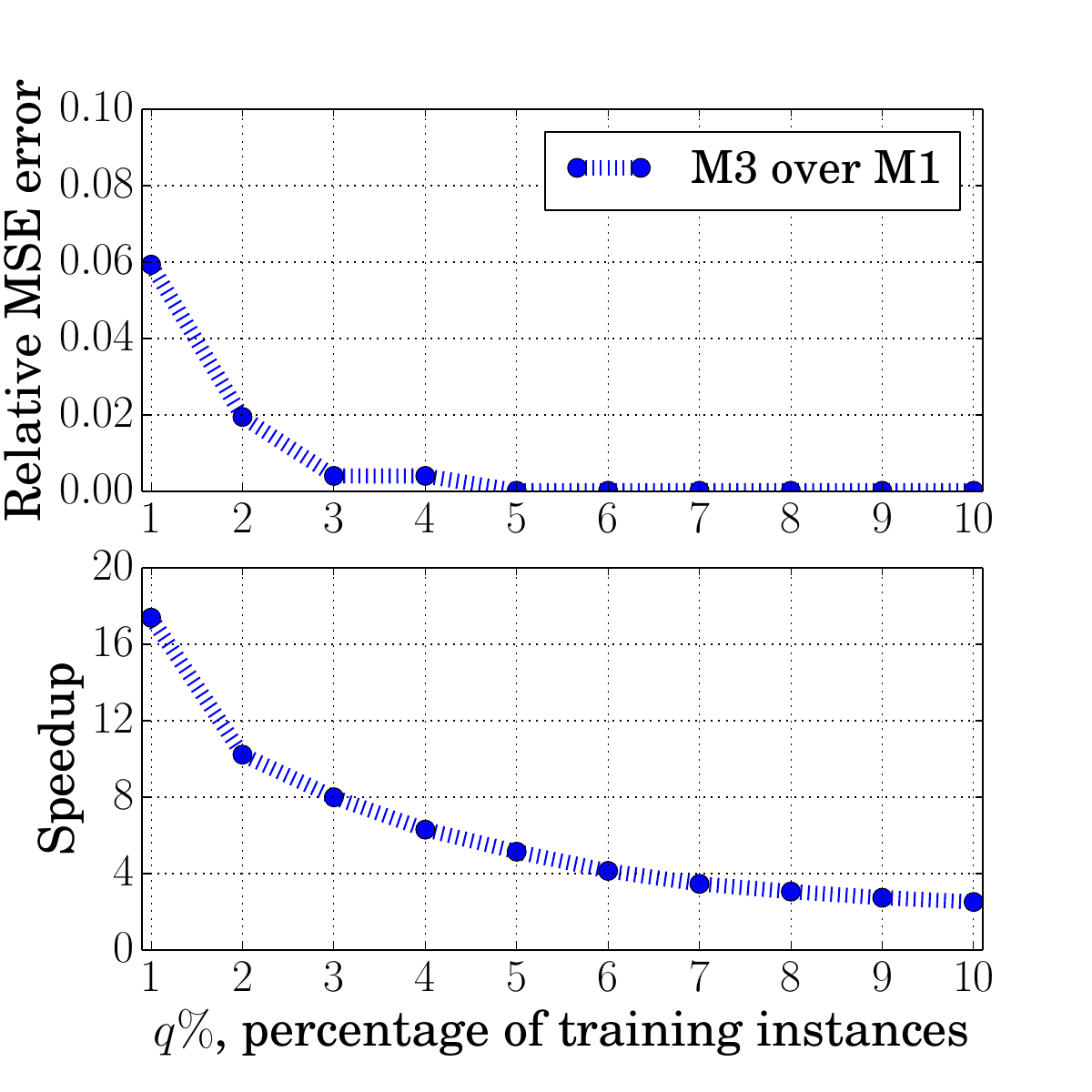}\label{ridge_UJIndoorLoc_m1_m3}} 
\subfloat[]{\includegraphics[width=0.25\textwidth]{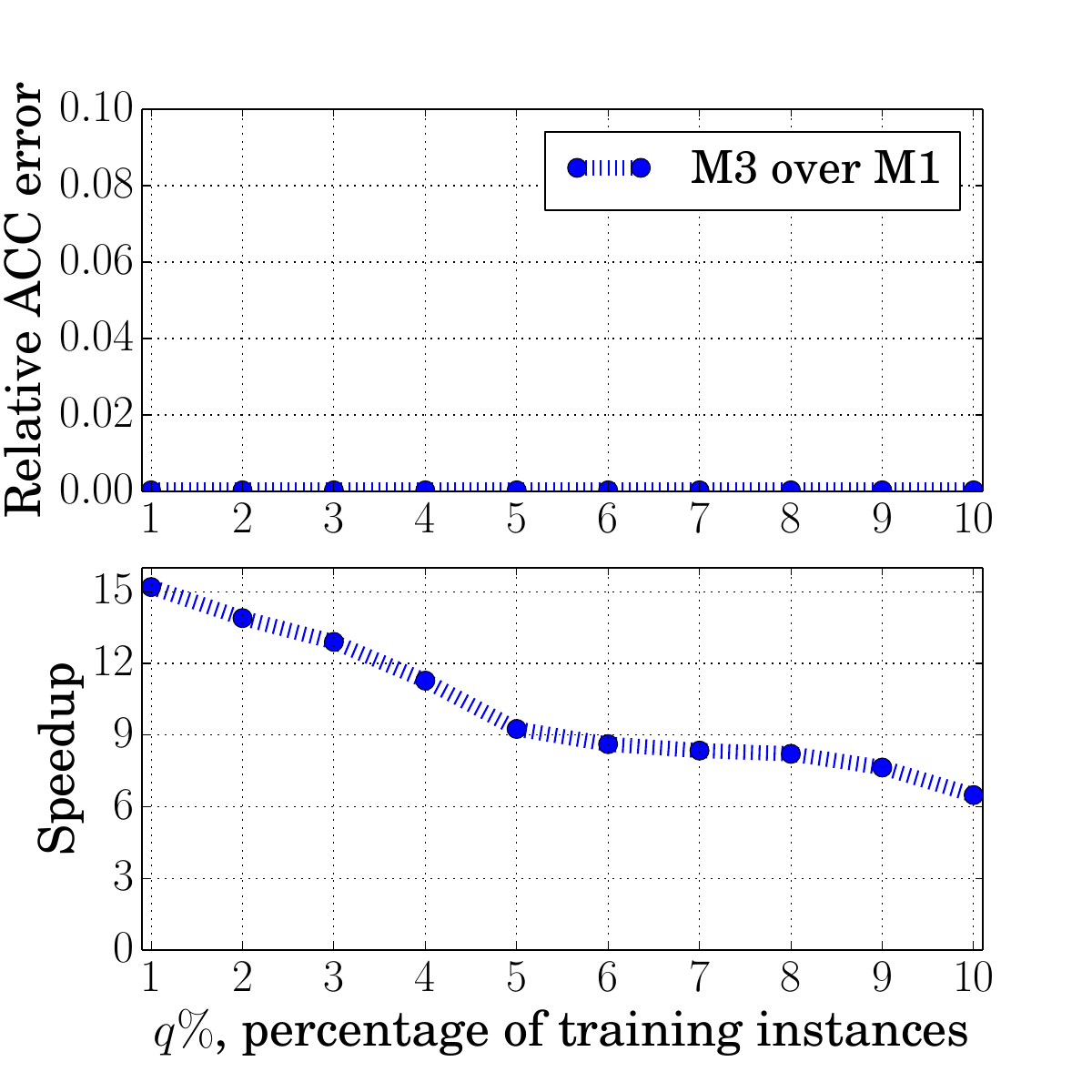}\label{svc_bank_m1_m3}}
%\subfloat[M3 vs. M1]{\includegraphics[width=0.23\textwidth]{./src_time/figs/result-UJIndoorLoc-ridge-acc-time-m2-m3.pdf}\label{ridge_UJIndoorLoc_m2_m3}}
\caption{M3 vs. M1. (a) Relative MSE error and speedup of M3 over M1 for RR on the dataset UJIndoorLoc. (b) Relative ACC error and speedup of M3 over M1 for SVM-SHL on the dataset Bank. %(Change y-axis to ``Relative MSE error" or ``Relative ACC error", x-axis to be ``$q$, percentage of training dataset")
}
\label{ridge_UJIndoorLoc_m123}
\end{figure}

\begin{figure}[t]
\center
\vspace{-6mm}
%\subfloat[L2-BLR]{\includegraphics[width=0.23\textwidth]{./src_time/figs/result-Gaussian-speedup-lr-0.001-error.pdf}\label{lr_speedup}} 
%\subfloat[BLSVC-SHL]
{\includegraphics[width=0.28\textwidth]{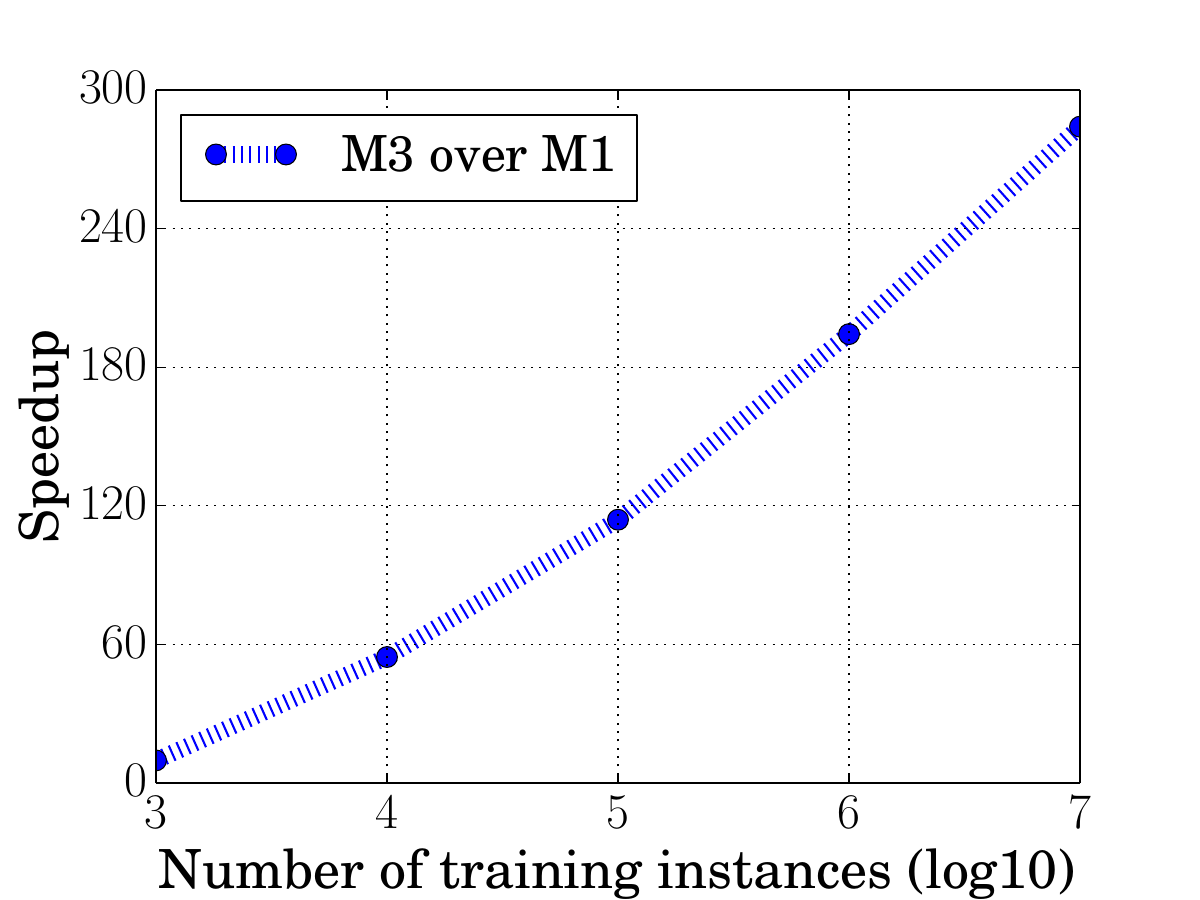}}%\label{svc_speedup}}
\caption{Speedup of M3 over M1 %when learning a model with the same accuracy, 
for SVM-SHL as the training dataset size gets larger. %(x-axis: Number of training instances (log10))
}
\label{speedup_gaussian}
\vspace{-4mm}
\end{figure}

\myparatight{M3 vs. M2} Figure~\ref{svc_bank_123} compares M3 with M2 with respect to their relative performance over M1 as we sample more training dataset for M2 (i.e., we increase $p$). For M3, we set $q\%$ such that the relative MSE (or ACC) error of M3 over M1 is smaller than 0.1\%. In particular, $q\%=3\%$ and $q\%=1\%$ for RR on the UJIndoorLoc dataset and SVM-SHL on the Bank dataset, respectively. We observe that when M3 and M2 achieve the same speedup over M1, the model learnt by M3 is more accurate than that learnt by M2. For instance, for  RR on the UJIndoorLoc dataset, M2 has the same speedup as M3 when sampling 10\% of training dataset, but M2 has around 4\% of relative MSE error while M3's relative MSE error is almost 0. For SVM-SHL on the Bank dataset, M2 has the same speedup as M3 when sampling 4\% to 5\% of training dataset, but M2's relative ACC error is much larger than M3's. 

The reason is that M2 learns both the hyperparameter and the model parameters using a subset of the training dataset. According to Figure~\ref{learningflow}, the \emph{unrepresentativeness} of the subset is ``doubled" because 1) it directly influences the model parameters, and 2) it influences the hyperparameter, through which it indirectly influences the model parameters. In contrast, in M3, such unrepresentativeness only influences the hyperparameter and the learning algorithms are relatively robust to small variations of the hyperparameter.

\begin{figure}[!t]
\center
\vspace{-3mm}
%\subfloat[]{\includegraphics[width=0.23\textwidth]{./src_time/figs/result-bank-svc-acc-time-m1-m3.pdf}\label{svc_bank_m1_m3}}
\subfloat[]{\includegraphics[width=0.25\textwidth]{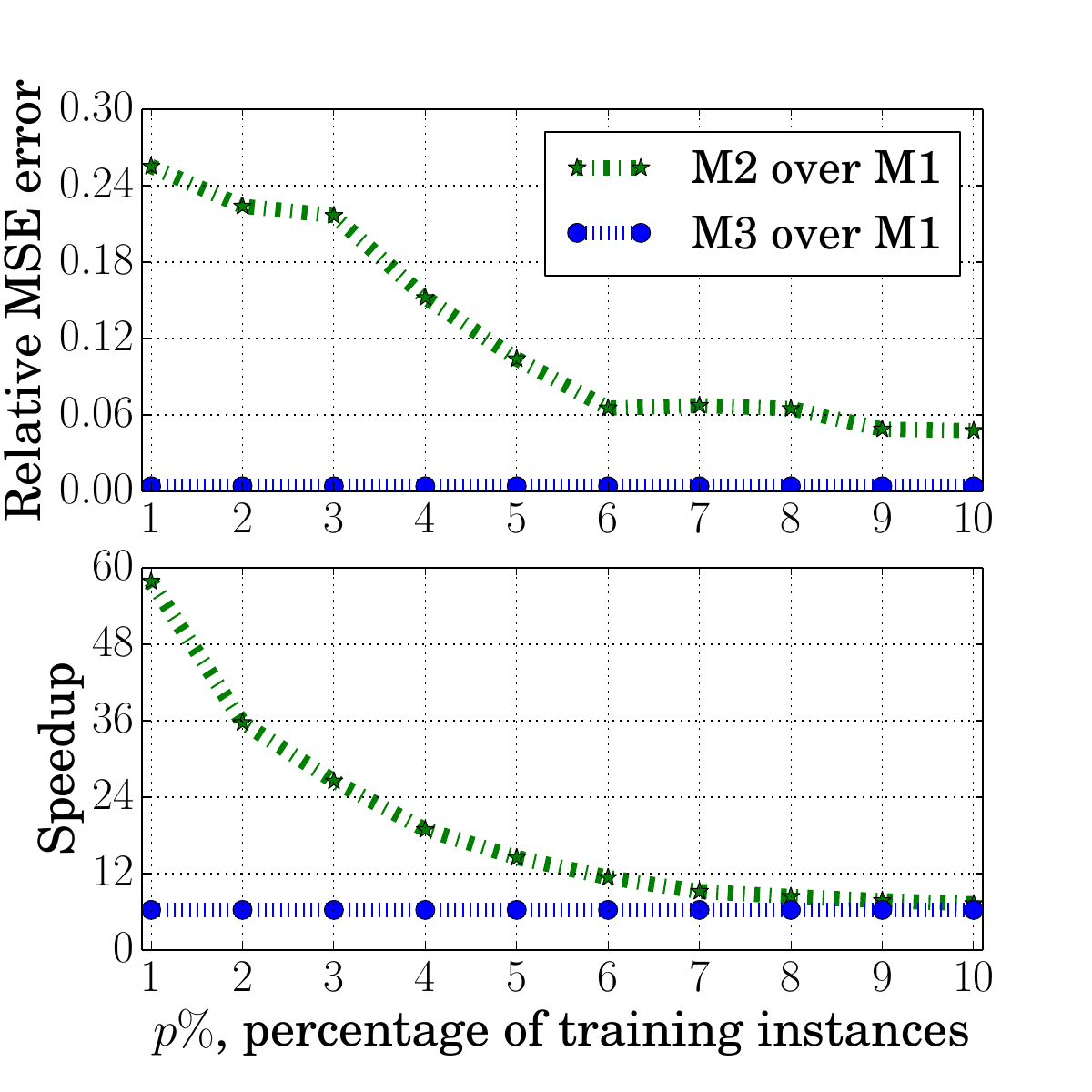}\label{ridge_UJIndoorLoc_m2_m3}}
\subfloat[]{\includegraphics[width=0.25\textwidth]{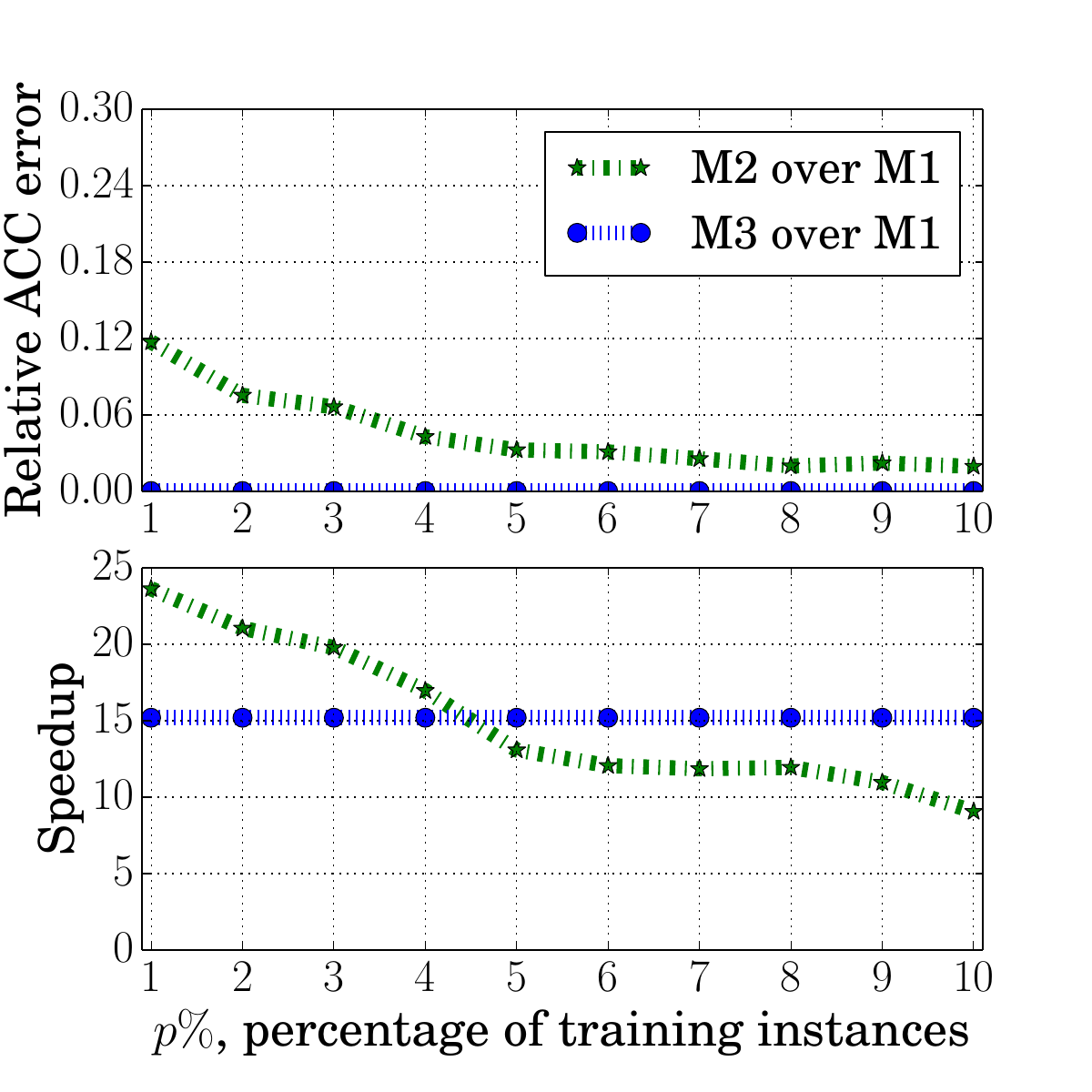}\label{svc_bank_m2_m3}}
\caption{M3 vs. M2. (a) Relative MSE error and speedup of M3 and M2 over M1 for RR on the dataset UJIndoorLoc. (b) Relative ACC error and speedup of M3 and M2 over M1 for SVM-SHL on the dataset Bank. %(x-axis, add $p$)
}
\label{svc_bank_123}
\vspace{-6mm}
\end{figure}

\alan{
\subsubsection{Attacking Amazon Machine Learning} We also evaluate the three methods using Amazon Machine Learning~\cite{amazon}. Amazon Machine Learning is a black-box MLaaS platform, i.e., it does not disclose model parameters nor hyperparameters to users. However, the ML algorithm is known to users, e.g., the default algorithm is logistic regression. 
In our experiments, we use Amazon Machine Learning to learn a logistic regression model (with $L_2$ regularization) for the Bank dataset. We leverage the SigOpt API~\cite{sigopt}, a hyperparameter tuning service for Amazon Machine Learning, to learn the hyperparameter.  We obtained a free API token from SigOpt.

We split the Bank dataset into two halves; one for training and the other for testing. For M2 and M3, we sampled 15\% and 3\% of the training dataset, respectively, i.e., p\%=15\% and q\%=3\% (we selected these settings such that M2 and M3 have around the same overall training costs). Since Amazon Machine Learning is black-box, we use the model parameter stealing attack~\cite{tramer2016stealing} to steal model parameters in our M3. Specifically, in M3, we first used 3\% of the training dataset to learn a logistic regression model. Amazon discloses the prediction API of the learnt model. Second, we queried the prediction API for 200 testing examples and used the \emph{equation-solving-based} attack~\cite{tramer2016stealing} to steal the model parameters. Third, we used our hyperparameter stealing attack to estimate the hyperparameter. Fourth, we used the entire training dataset and the stolen hyperparameter to re-train a logistic regression model. We also evaluated the accuracy of the three models learnt by the three methods on the testing data via their prediction APIs.

The overall training costs for M1, M2, and M3 (including the cost of querying the prediction API for stealing model parameters) are \$1.02, \$0.15, and \$0.16, respectively. The cost per query of the prediction API is \$0.0001. The relative ACC error of M2 over M1 is 5.1\%, while the relative ACC error of M3 over M1 is 0.92\%. Therefore, compared to M1, M3 saves training costs significantly with little accuracy loss. When M2 and M3 have around the same training costs, M3 is much more accurate than M2.

\subsubsection{Summary}
Through empirical evaluations, we have the following key observations. First, M3 (i.e., the Train-Steal-Retrain strategy) can learn a model that is as accurate as that learnt by M1 with much less computational costs. This implies that, for the considered MLaaS platforms, a user can use our  attacks to learn an accurate model while saving a large amount of economic costs. Second, M3 has bigger speedup over M1 when the training dataset is larger. Third, {M3} is more accurate than M2 when having the same speedup over M1. 
}

%\begin{packeditemize}%[noitemsep,topsep=0pt]
%\item M3 (i.e., the Train-Steal-Retrain strategy) can learn a model that is as accurate as that learnt by M1 with much less computational costs. This implies that, for the considered MLaaS platforms, a user can use our  attacks to learn an accurate model while saving a large amount of economic costs. 
%%save much running time and maintain high accuracy when applied to cloud-based ML services. 
%\item M3 has more significant speed over M1 when the training dataset gets larger. 
%\item {M3} is more accurate than M2 when having the same speedup over M1. 
%\end{packeditemize}

\section{Rounding As a Defense}
\label{hyper_defense}

According to our Theorem~\ref{thm_2}, the estimation error of the hyperparameter is linear to the difference between the learnt model parameters and the minimum of the objective function that is closest to them. This theorem implies that we could defend against our hyperparameter stealing attacks via increasing such difference. Therefore, we propose that the learner \emph{rounds} the learnt model parameters before sharing them with the end user. For instance, suppose a model parameter is 0.8675342, rounding the model parameter to one decimal and two decimals results in 0.9 and 0.87, respectively. We note that this rounding technique was also used by Fredrikson et al.~\cite{fredrikson2015model} and Tram\`{e}r et al.~\cite{tramer2016stealing} to obfuscate confidence scores of model predictions to defend against model inversion attacks and model stealing attacks, respectively. 

Next, we perform experiments to empirically evaluate the effectiveness of the rounding technique at defending against our hyperparameter stealing attacks. 

%In this section, we provide preliminary countermeasures against hyperparameter stealing attack via rounding model parameters. We verify the rounding technique on real-world datasets and notice that it has different impacts on different methods. We then explain this phenomenon by analyzing the hyperparameter sensitivity against parameter rounding. 

%\subsection{Hyperparameter Defense Analysis}
\subsection{Evaluations}
\subsubsection{Setup}
%\myparatight{Datasets} 
We use the datasets listed in Table~\ref{data_reg}. Specifically, for each dataset, 
 we first randomly split the dataset into a training dataset and a testing dataset with an equal size. Second, for each ML algorithm we considered, 
 we learn a hyperparameter using the training dataset via 5-fold cross-validation, and learn the model parameters via the learnt hyperparameter and the training dataset. Third, we round each model parameter to a certain number of decimals (we explored from 1 decimal to 5 decimals). Fourth, we estimate the hyperparameter using
 the rounded model parameters.

\myparatight{Evaluation metrics} Similar to evaluating the effectiveness of our attacks, the first metric we adopt is the relative estimation error of the hyperparameter value, which is formally defined in Eqn.~\ref{rer}. We say rounding is an effective defense for an ML algorithm if rounding makes the relative estimation error  larger. Moreover, we say one ML algorithm can more effectively defend against our attacks than another ML algorithm using rounding, if the relative estimation error of the former algorithm increases more than that of the latter one. 

%A higher relative estimation error means that the rounding technique is more effective at defending against our hyperparameter stealing attacks. 

However, relative estimation error alone is insufficient because it only measures security, while ignoring the testing performance of the rounded model parameters. Specifically, severely rounding the model parameters could make the ML algorithm secure against our hyperparameter stealing attacks, but the testing performance of the rounded model parameters might also be affected significantly. 
Therefore, we also consider a metric to measure the testing-performance loss that is resulted from rounding model parameters. In particular, suppose the unrounded model parameters have a testing MSE (or ACC for classification algorithms), and the rounded model parameters have a testing MSE$_r$ (or ACC$_r$) on the same testing dataset. Then, we define the \emph{relative MSE error} and \emph{relative ACC error} as $\frac{|\text{MSE}-\text{MSE}_r|}{\text{MSE}}$ and $\frac{|\text{ACC}-\text{ACC}_r|}{\text{ACC}}$, respectively.  
%\begin{align}
%\text{Relative MSE error}=\frac{|\text{MSE}-\text{MSE}_r|}{\text{MSE}}\\
%\text{Relative ACC error}=\frac{|\text{ACC}-\text{ACC}_r|}{\text{ACC}}
%\end{align}
%Similarly, we can define the \emph{relative ACC error}. 
Note that the relative MSE error and the relative ACC error used in this section are different from those used in Section~\ref{eval_attack}.
A larger relative estimation error and a smaller relative MSE (or ACC) error indicate a better defense strategy.

\begin{figure*}[!t]
\center
\vspace{-4mm}
\subfloat[Diabetes]{\includegraphics[width=0.33\textwidth]{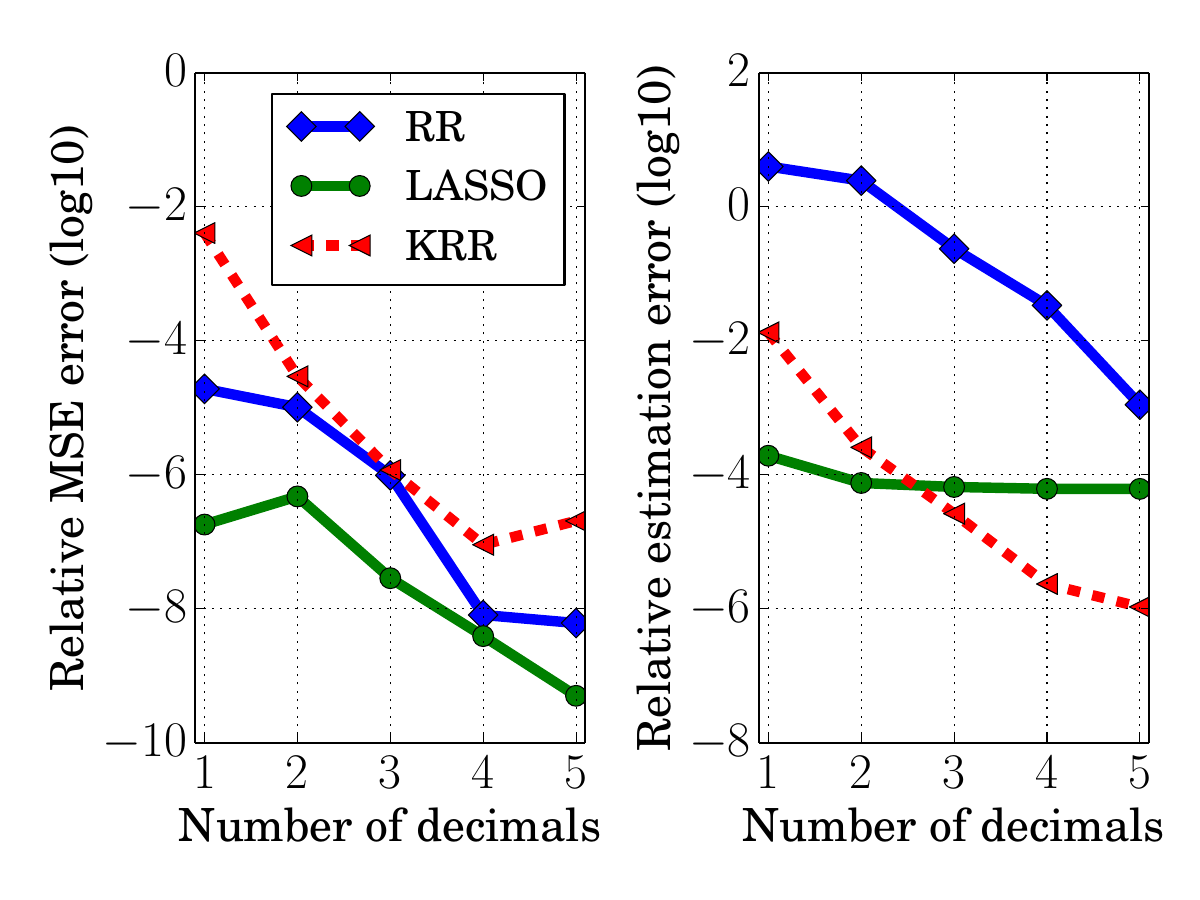}\label{diabetes_defense}} 
\subfloat[GeoOrigin]{\includegraphics[width=0.33\textwidth]{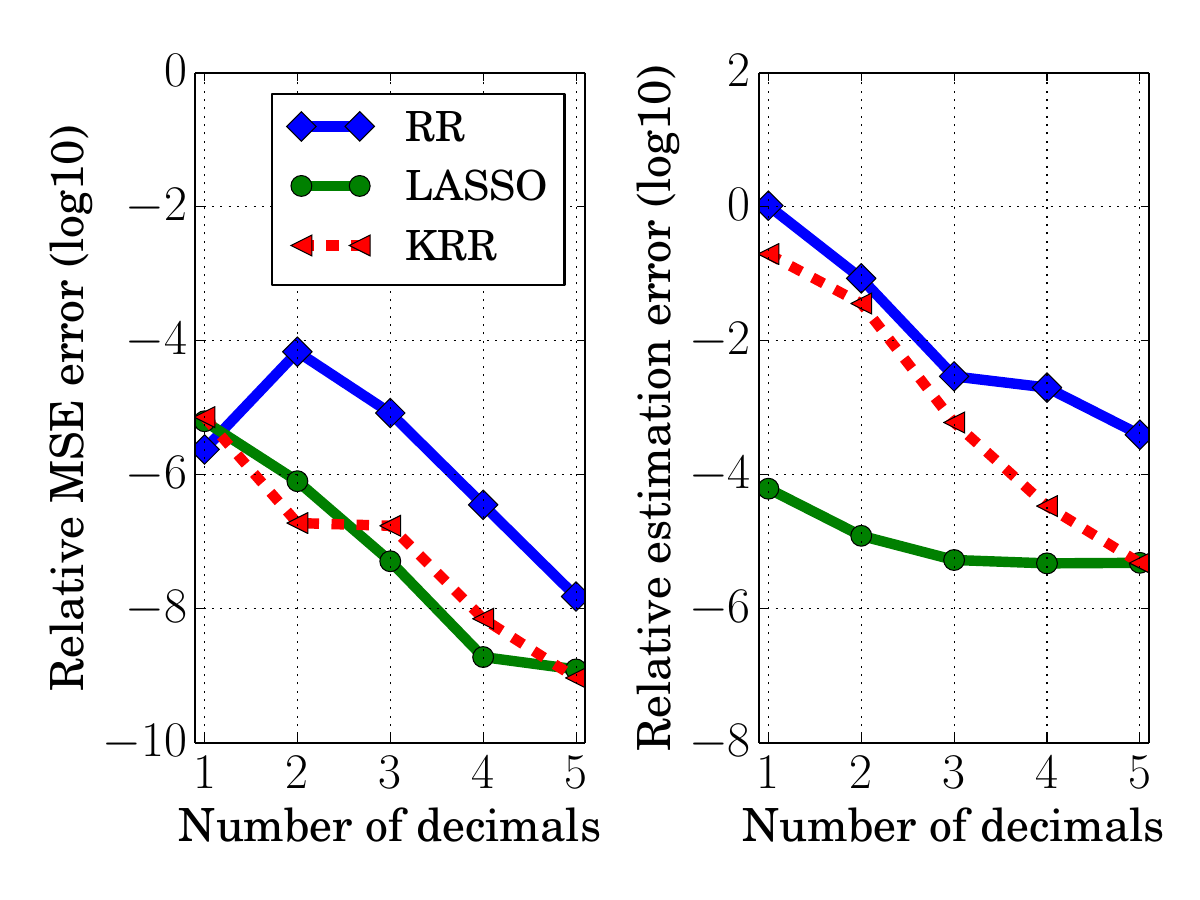}\label{geoorigin_defense}}
\subfloat[UJIIndoor]{\includegraphics[width=0.33\textwidth]{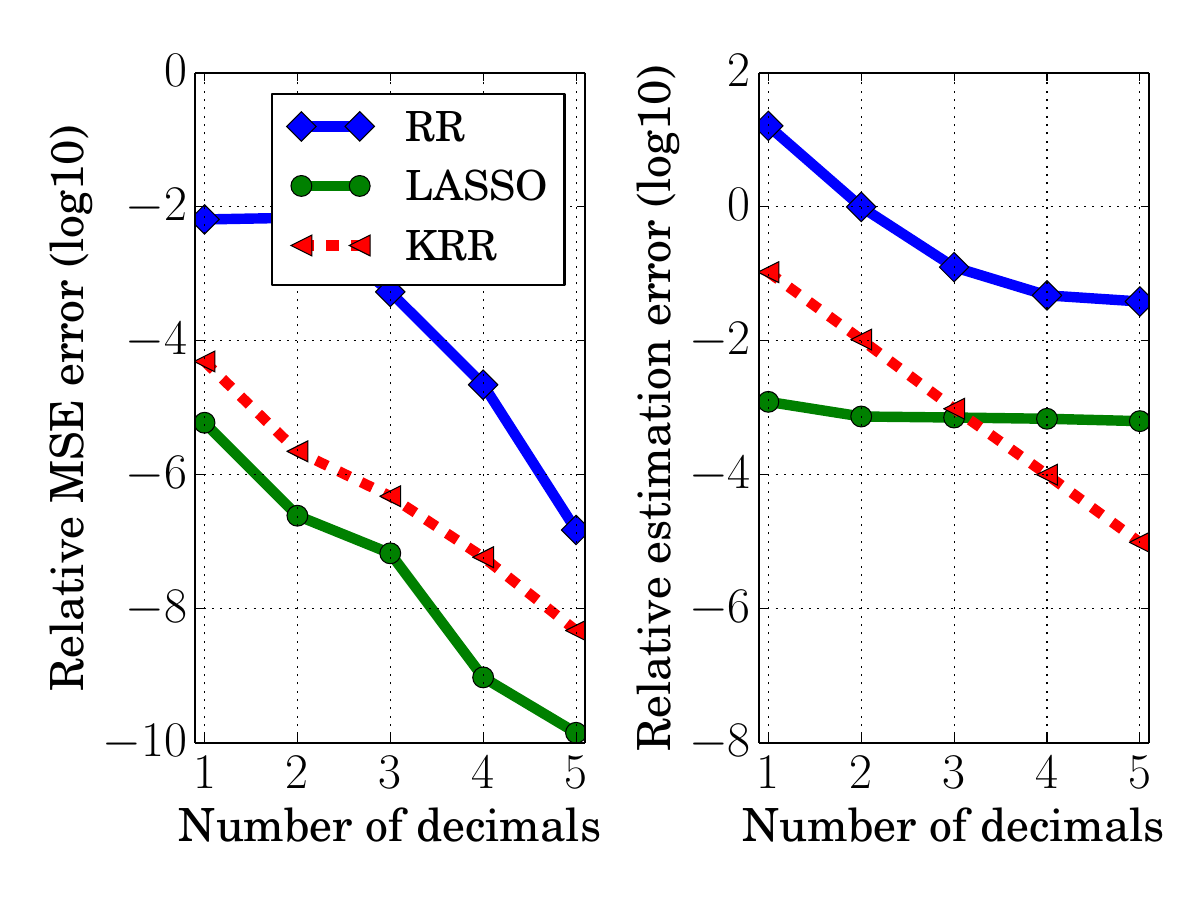}\label{ujiindoor_defense}}
\caption{Defense results of the rounding technique for regression algorithms.}
\label{reg_res_defense}
\vspace{-8mm}
\end{figure*}

\begin{figure*}[t]
\center
\subfloat[Iris]{\includegraphics[width=0.33\textwidth]{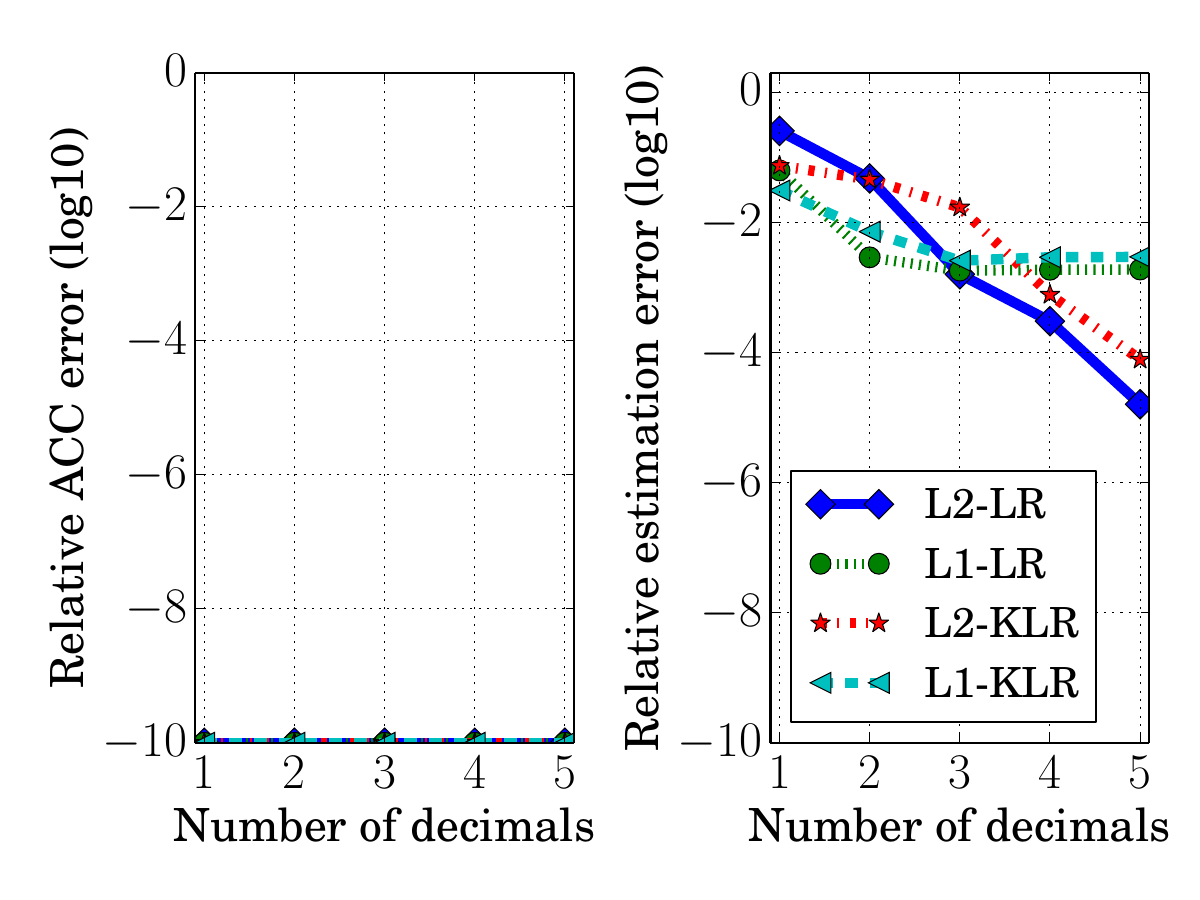}\label{clf_lr_madelon_defense}}
\subfloat[Madelon]{\includegraphics[width=0.33\textwidth]{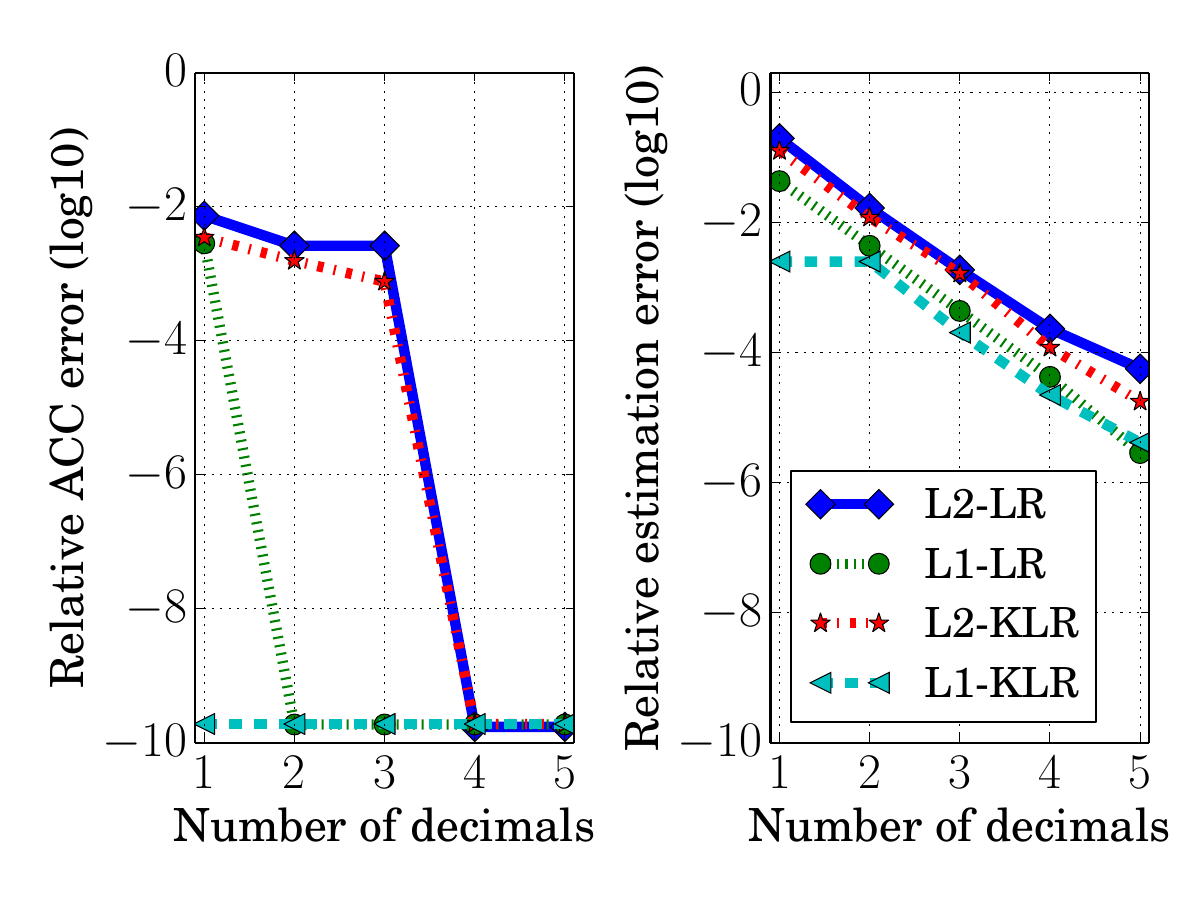}\label{clf_lr_madelon_defense}}
\subfloat[Bank]{\includegraphics[width=0.33\textwidth]{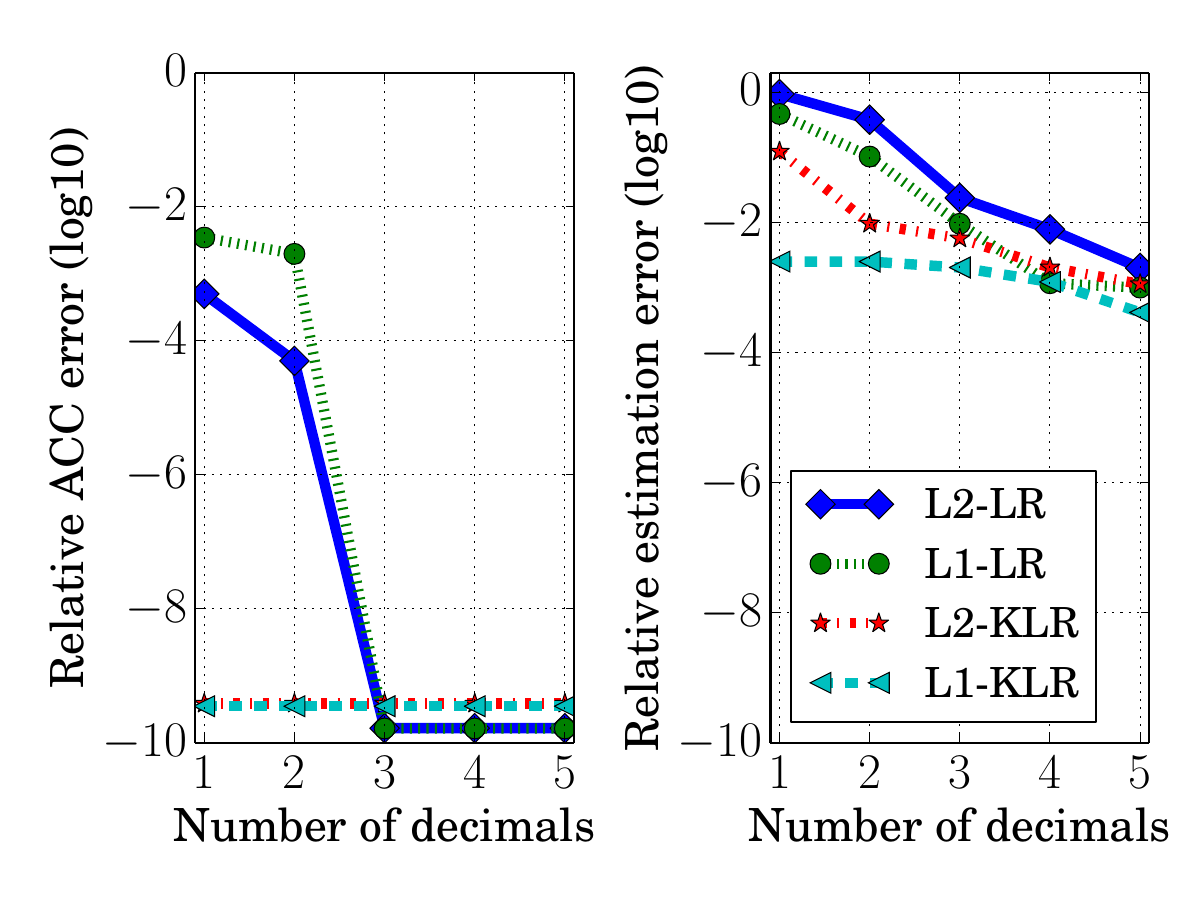}\label{clf_lr_bank_defense}} 
\caption{Defense results of the rounding technique for logistic regression classification algorithms. 
}
\label{lr_res_defense}
\vspace{-8mm}
\end{figure*}

\begin{figure*}[t]
\center
\subfloat[Iris]{\includegraphics[width=0.33\textwidth]{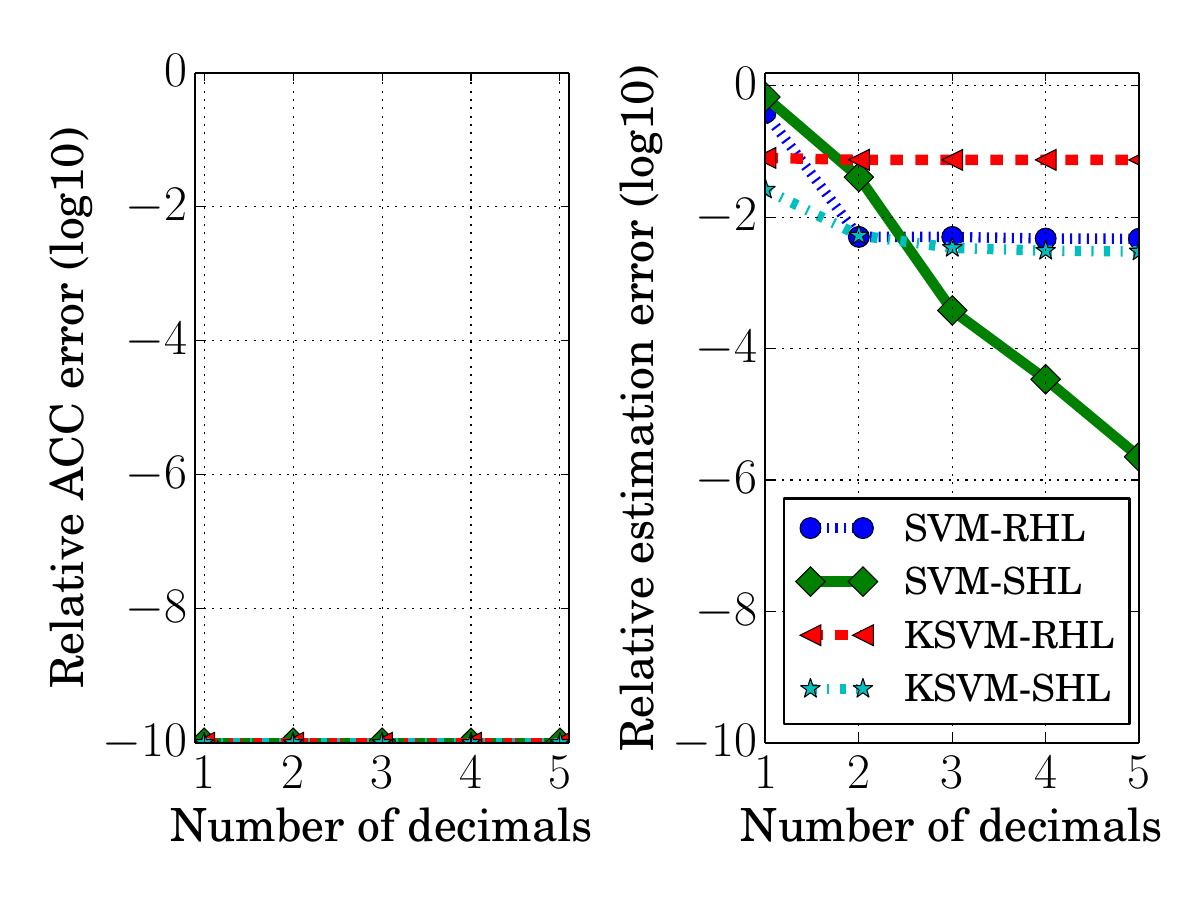}\label{clf_svc_iris_defense}}  
\subfloat[Madelon]{\includegraphics[width=0.33\textwidth]{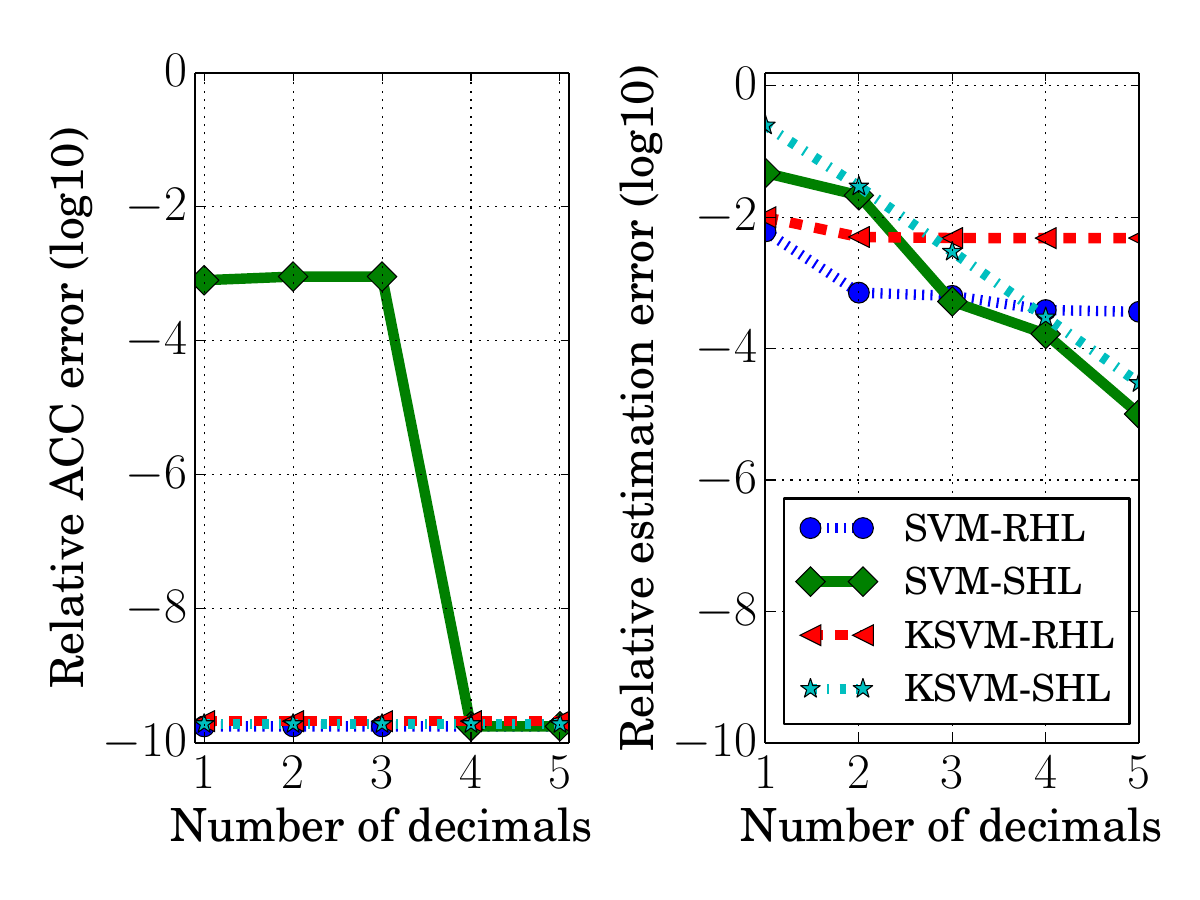}\label{clf_svc_madelon_defense}}
\subfloat[Bank]{\includegraphics[width=0.33\textwidth]{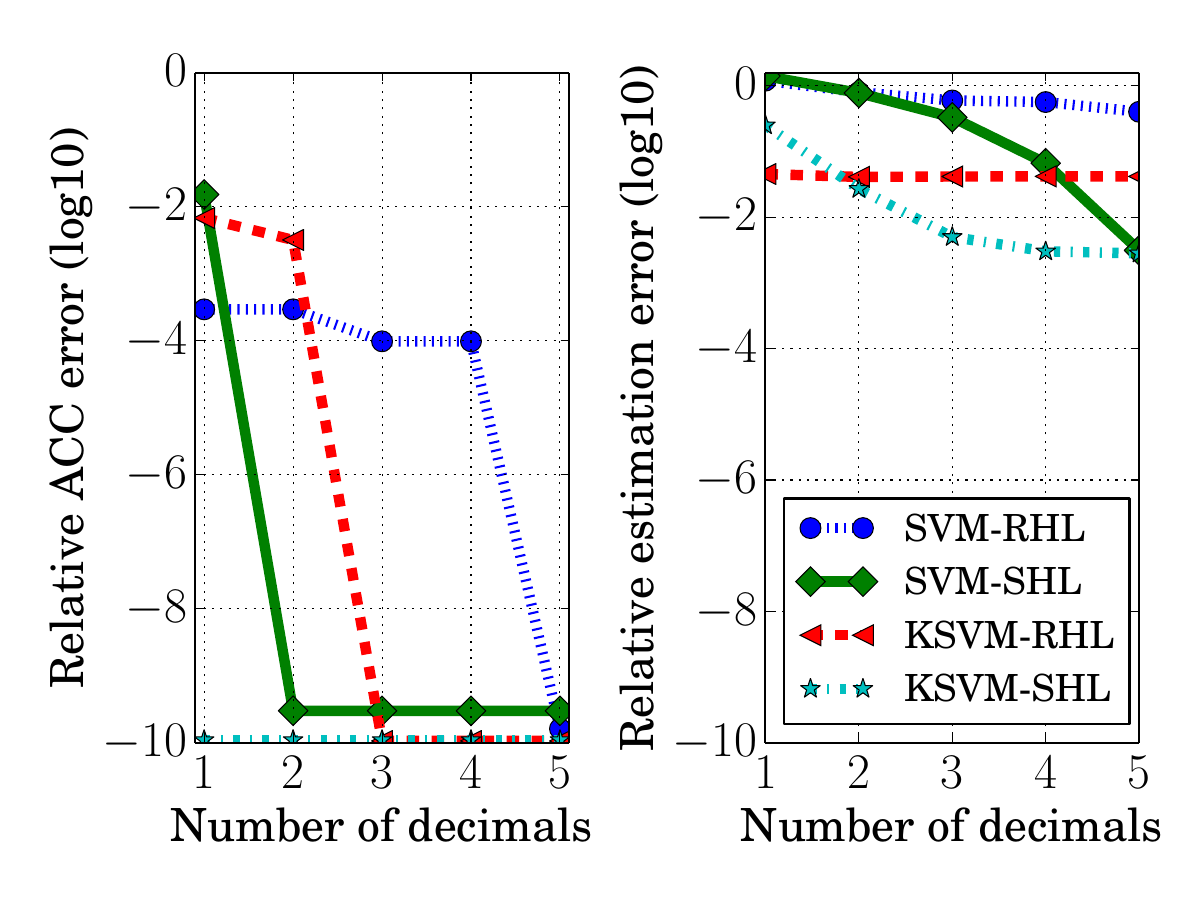}\label{clf_svc_bank_defense}}
%\subfigure[Digits]{\includegraphics[width=0.32\textwidth]{./src_defense/figs/clf/result-clf-mlr-digits-new-new.pdf}\label{clf_mlr_digits_defense}} 
\caption{Defense results of the rounding technique for SVM classification algorithms. 
%(y-axis: ``Relative ACC error", ``Relative estimation error")
}
\label{clf_res_defense}
\vspace{-6mm}
\end{figure*}

\begin{figure}[t]
\center
\vspace{-4mm}
\subfloat[Regression]{\includegraphics[width=0.25\textwidth]{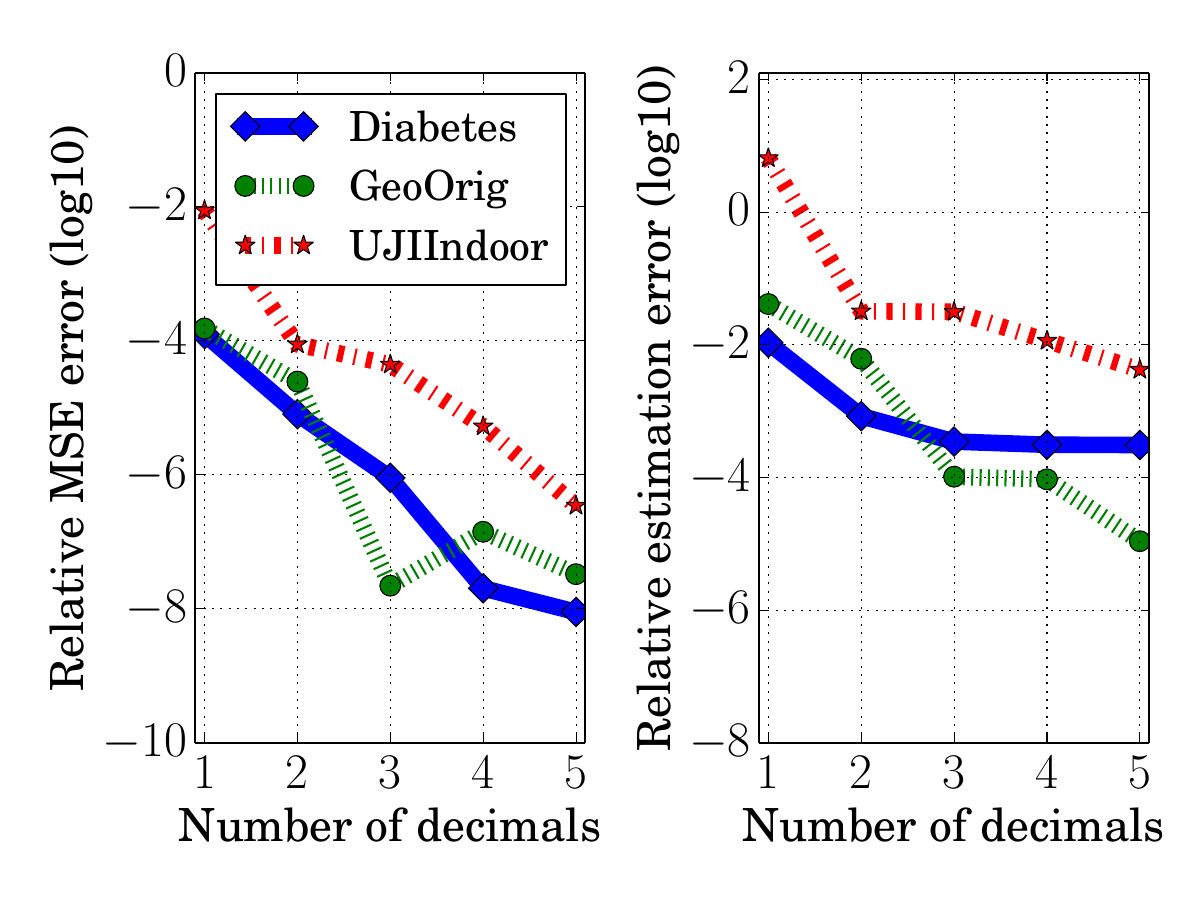}\label{nn_reg_defense}}  
\subfloat[Classification]{\includegraphics[width=0.25\textwidth]{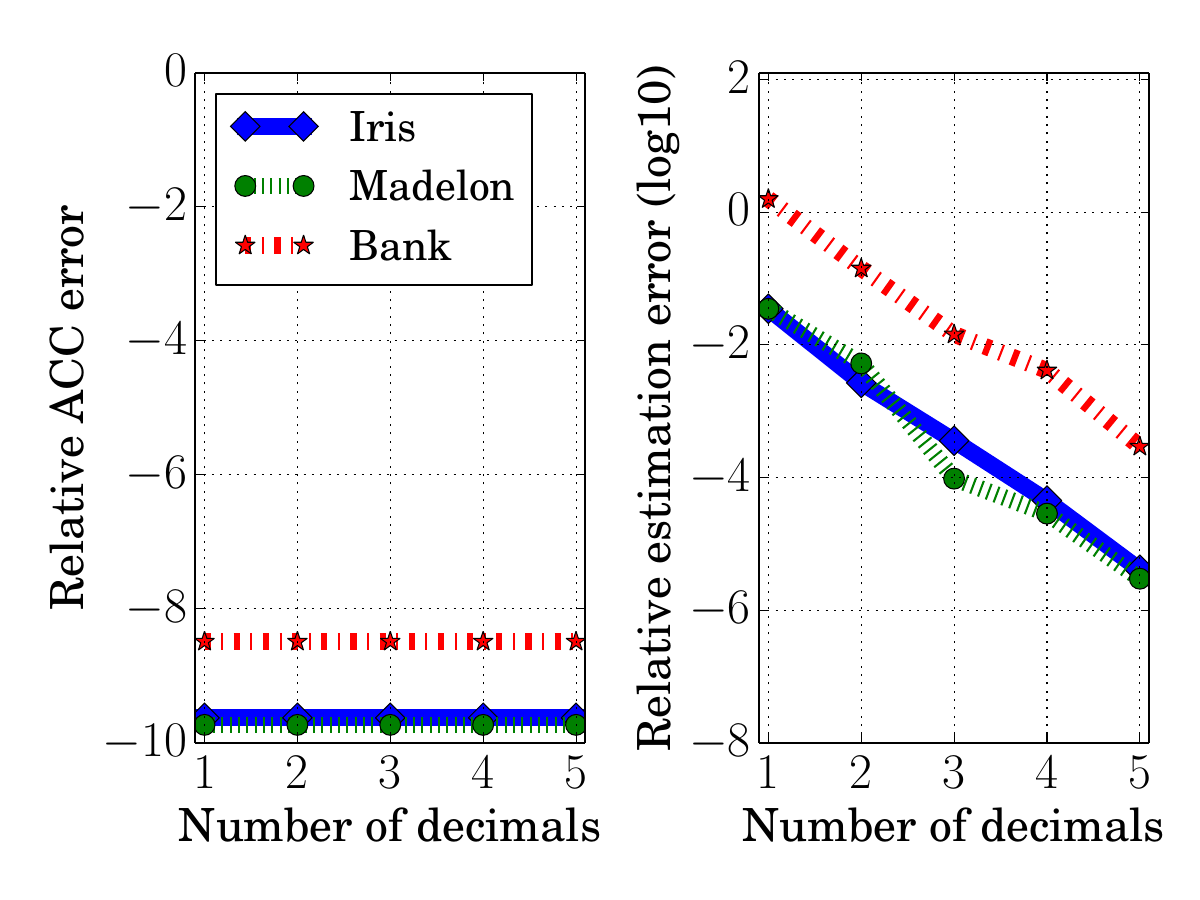}\label{nn_clf_defense}}
%\subfigure[Digits]{\includegraphics[width=0.32\textwidth]{./src_defense/figs/clf/result-clf-mlr-digits-new-new.pdf}\label{clf_mlr_digits_defense}} 
\caption{Defense results of the rounding technique for a) neural network regression algorithm and b) neural network classification algorithm. 
%(y-axis: ``Relative ACC error", ``Relative estimation error")
}
\label{nn_res_defense}
\vspace{-8mm}
\end{figure}

\subsubsection{Results}
\label{defense-eva}

%\subsubsection{Experimental Setup}
%
%We evaluate hyperparameter defense capability by leveraging the rounding technique. 
%%and analyze hyperparameter sensitivity of all above ML methods on several datasets. 
%%First, we randomly split each dataset into the training set with proportion $1/2$ and testing set with proportion $1/2$. 
%First, we randomly split each dataset into the training set and testing set with an equal size. 
%Then, we use the training set and optimization algorithms listed in Table~\ref{sum_ml} and 5-fold CV to learn the optimal model parameters and hyperparameters for each ML algorithm. After that, we steal each model's hyperparameter.
%
%To perform defense against hyperparameter stealing, we round decimals and keep its number from 1 to 5, respectively. We take the metric average relative $\gamma$-hyperparameter error measured on the training set and relative $\delta$-MSE (or relative $\delta$-ACC error) measured on the testing set to quantify defense capability. Our defense strategy is implemented using Python.

Figure~\ref{reg_res_defense}, \ref{lr_res_defense}, \ref{clf_res_defense}, and~\ref{nn_res_defense} illustrate defense results for regression, logistic regression, SVM, and three-layer neural networks, respectively. \alan{Since we use log scale in the figures, we set the relative MSE (or ACC) errors to be $10^{-10}$ when they are 0.}
%Figure~\ref{lr_res_defense} and Figure~\ref{clf_res_defense} show the defense results for logistic regression and SVM classification algorithms, respectively. 
%Figure~\ref{nn_res_defense} shows the defense results for neural network regression and classification. 
%In order to compare the relative estimation errors with our attack results, % shown in Figure~\ref{reg_res_attack}--Figure~\ref{clf_res_attack_svm}, 
%we use a log scale for the relative estimation errors in the figures. %We have the follow four key observations. 
 
\noindent \textbf{Rounding is not effective enough for certain ML algorithms:} 
Rounding has small impact on the testing performance of the models. For instance, when we keep one decimal, 
all ML algorithms have relative MSE (or ACC) errors smaller than 2\%. Moreover, rounding model parameters increases the relative estimation errors of our attacks for all ML algorithms. 
However, for certain ML algorithms, the relative estimation errors are still very small when significantly rounding the model parameters, implying that our attacks are still very effective. For instance, for LASSO, our attacks have relative estimation errors that are consistently smaller than around $10^{-3}$ across the datasets, even if we round the model parameters to one decimal.  
These results highlight the needs for new countermeasures for certain ML algorithms. 
%Rounding is an effective defense strategy against our hyperparameter stealing attacks. 
%Specifically, when we keep a smaller number of decimals for the model parameters, the relative estimation error increases significantly. 
%For instance, when we keep one decimal, our attack has orders of magnitude larger relative estimation errors than unrounding the model parameters for all algorithms.
%Moreover, rounding the model parameters has negligible impact on the testing performance of the model. For instance, when we keep one decimal, 
%all learning algorithms have relative MSE (or ACC) errors that are smaller than 1\%. 
 
%\myparatight{Verifying Theorem~\ref{thm_2}} Second, relative estimation error is linear to the rounding decimals, which is consistent with Theorem~\ref{thm_2}.

\noindent \textbf{Comparing regularization terms: $L_2$ regularization is more effective than $L_1$ regularization:} Different ML algorithms could use different regularization terms, so one natural question is which regularization term can more effectively defend against our attacks using rounding. All the SVM classification algorithms that we studied use $L_2$ regularization term. Therefore, we use results on regression algorithms and logistic regression classification algorithms (i.e., Figure~\ref{reg_res_defense} and Figure~\ref{lr_res_defense}) to compare regularization terms. In particular, we use three pairs: RR vs. LASSO, L2-LR vs. L1-LR, and L2-KLR vs. L1-KLR. The two algorithms in each pair have the same loss function, and use $L_2$ and $L_1$ regularizations, respectively.

%the regression algorithms RR and LASSO use $L_2$ and $L_1$ regularization term, respectively; and they use the same loss function. L2-LR and L2-KLR use $L_2$ regularization term, while  L1-LR and L1-KLR use $L_1$ regularization term. All these four classification algorithms use the same loss function, i.e., cross entropy.  

We observe that $L_2$ regularization can more effectively defend against our attacks than $L_1$ regularization using rounding.
%$L_2$ regularization is more \emph{sensitive} to rounding than $L_1$ regularization. 
%Specifically, if we round the model parameters with one less decimals, the relative estimation errors of algorithms with $L_2$ regularization increase more than those of algorithms with $L_1$ regularization. For instance,  
Specifically, the relative estimation errors of RR (or L2-LR or L2-KLR) increases faster than those of LASSO (or L1-LR or L1-KLR),
% the relative estimation errors of L2-LR increases faster than those of L1-LR, and the relative estimation errors of L2-KLR increases faster than L1-KLR, 
 as we round the model parameters to less decimals. For instance, when we round the model parameters to one decimal, the relative estimation errors increase by $10^{11}$ and $10^2$ for RR and LASSO on the Diabetes dataset, respectively, compared to those without rounding.

These observations are consistent with our Theorem~\ref{thm_2}. In particular, Appendix~\ref{approx} shows our approximations to the gradient $\nabla \hat{\lambda}(\mathbf{w}^{\star})$ in Theorem~\ref{thm_2} for RR, LASSO, L2-LR, L2-KLR, L1-LR, and L1-KLR.
For an algorithm with $L_2$ regularization, the magnitude of the gradient at the exact model parameters is inversely proportional to the $L_2$ norm of the model parameters. However, if the algorithm has $L_1$ regularization, the magnitude of the gradient is inversely proportional to the $L_2$ norm of the sign of the model parameters. For algorithms with $L_2$ regularization, the learnt model parameters are often small numbers, and thus the $L_2$ norm of the model parameters is smaller than that of the sign of the model parameters. As a result, the magnitude of the gradient for an algorithm with $L_2$ regularization is larger than that for an algorithm with $L_1$ regularization. Therefore, according to Theorem~\ref{thm_2}, when we round model parameters to less decimals, the estimation errors of an algorithm with $L_2$ regularization increase more than those with $L_1$ regularization.

%This observation implies that  $L_2$ regularization can more effectively defend against our attacks than $L_1$ regularization using rounding.
%For instance, when we round the model parameters to one decimal, the relative estimation errors increase by XX and XX for RR and LASSO on the Diabetes dataset, respectively, compared to the relative estimation errors when we do not round the model parameters.  
%
%an algorithm with $L_2$ regularization achieves a significantly larger relative estimation error than its counterpart with $L_1$ regularization. 

\noindent \textbf{Comparing loss functions: cross entropy and square hinge loss can more effectively defend against our attacks than regular hinge loss:} We also compare defense effectiveness of different loss functions. Since all regression algorithms we studied have the same loss function, we use classification algorithms to compare loss functions. Specifically, we use two triples: (L2-LR, SVM-SHL, SVM-RHL) and  (L2-KLR, KSVM-SHL, KSVM-RHL). The three algorithms in each triple use cross entropy loss, square hinge loss, and regular hinge loss, respectively, while all using $L_2$ regularization. 

We find that cross entropy and square hinge loss have similar defense effectiveness against our attacks, while they can more effectively defend against our attacks than regular hinge loss using rounding. For instance, Figure~\ref{lossfunction} compares the relative estimation errors of the triple (L2-LR, SVM-SHL, SVM-RHL) on the  dataset Madelon when we use rounding. The relative estimation errors of L2-LR and SVM-SHL increase with a similar speed, but both increase faster than those of SVM-RHL, as we round the model parameters to less decimals.  For instance, when we round the model parameters to one decimal, the relative estimation errors increase by $10^{5}$, $10^6$, and $10^2$ for L2-LR, SVM-SHL, SVM-RHL on the  dataset Madelon, respectively, compared to those without rounding.

%Again, our Theorem~\ref{} can explain the different defense effectiveness between regular hinge loss and square hinge loss functions. 
%For instance, the gradient of SVM-RHL (i.e., Eqn.~\ref{}) is approximately linear to $\sum_{j}y_j\mathbj{x}_j$, where $y_j \mathbj{w}^T\mathbj{x}_j<1$,
% and the gradient of SVM-SHL is approximately linear to $\sum_{j} 2 y_j (1-y_j \mathbj{w}^T\mathbj{x}_j)\mathbj{x}_j$, where $y_j \mathbj{w}^T\mathbj{x}_j<1$. We note that, in SVM, $(1-y_j \mathbj{w}^T\mathbj{x}_j)$ is close to 0 for many  

\subsection{Implications for MLaaS}

% \begin{figure}[!htbp]
% \center
% \subfigure[RR]{\includegraphics[width=0.25\textwidth]{./src_time/figs/result-UJIndoorLoc-ridge-defense-m1-m3.pdf}\label{ridge_defense_m1_m3}} 
% %\subfigure[LASSO]{\includegraphics[width=0.45\textwidth]{./src_time/figs/result-UJIndoorLoc-Lasso-defense-m1-m3.pdf}\label{lasso_defense_m1_m3}}
% \caption{Relative MSE of RR under defense on UJIndoorLoc.}
% \label{real_defense_reg}
% \end{figure}

% \begin{figure*}[t]
% \center
% \subfigure[RR]{\includegraphics[width=0.45\textwidth]{./src_time/figs/result-UJIndoorLoc-ridge-defense-m1-m3.pdf}\label{ridge_defense_m1_m3}} 
% \subfigure[SVM-SHL]{\includegraphics[width=0.45\textwidth]{./src_time/figs/result-bank-svc-defense-m1-m3.pdf}\label{svc_defense_m1_m3}}
% %\subfigure[LASSO]{\includegraphics[width=0.45\textwidth]{./src_time/figs/result-UJIndoorLoc-Lasso-defense-m1-m3.pdf}\label{lasso_defense_m1_m3}}
% \caption{(a) Relative MSE error of RR under defense on the dataset UJIndoorLoc; (b) Relative ACC error of SVM-SHL under defense on the dataset Bank.}
% \label{real_defense_reg}
% \end{figure*}

%\subsubsection{Evaluation Against Rounding-based Defense}

%\begin{figure}[!htbp]
%\center
%\subfloat[L2-LR]{\includegraphics[width=0.23\textwidth]{./src_time/figs/result-bank-lr-defense-m1-m3.pdf}\label{lr_defense_m1_m3}} 
%\subfloat[SVM-SHL]{\includegraphics[width=0.23\textwidth]{./src_time/figs/result-bank-svc-defense-m1-m3.pdf}\label{svc_defense_m1_m3}}
%\caption{Relative ACC error of L2-LR and SVM-SHL under defense on Bank.}
%\label{real_defense_clf}
%\end{figure}

Recall that,  in Section~\ref{eval_attack}, we demonstrate that a user can use {M3}, i.e., the Train-Steal-Retrain strategy, to learn a model through an MLaaS platform with much less economical costs, while not sacrificing the model's testing performance. We aim to study whether M3 is still effective if the MLaaS rounds the model parameters. We follow the same experimental setup as in Section~\ref{eval_attack}, except that the MLaaS platform rounds the model parameters to one decimal before sharing them with the user. 
Figure~\ref{real_defense_reg} compares M3 with M1 with respect to relative  ACC error of M3 over M1. Note that the speedups of M3 over M1 are the same with those in Figure~\ref{ridge_UJIndoorLoc_m123}, so we do not show them again. We observe that M3 can still save many economical costs, though rounding makes the saved costs less. Specifically, when we sample 10\% of the training dataset, the relative ACC error of M3 is less than around 0.1\% in Figure~\ref{real_defense_reg}, while M3 is 6 times faster than M1 (see Figure~\ref{ridge_UJIndoorLoc_m123}).

%rounding can significantly reduce the effectiveness of M3. 
%In particular, the model parameters learnt  by M3 significantly sacrifice testing performance when the MLaaS platform rounds the model parameters. In other words, in order to learn a model using M3 without sacrificing testing performance, the user has to sample a larger fraction of training instances, which means that M3 has smaller speedups and thus saves less economic costs. 

\begin{figure}[t]
\vspace{-4mm}
\center
\subfloat[]{\includegraphics[width=0.25\textwidth]{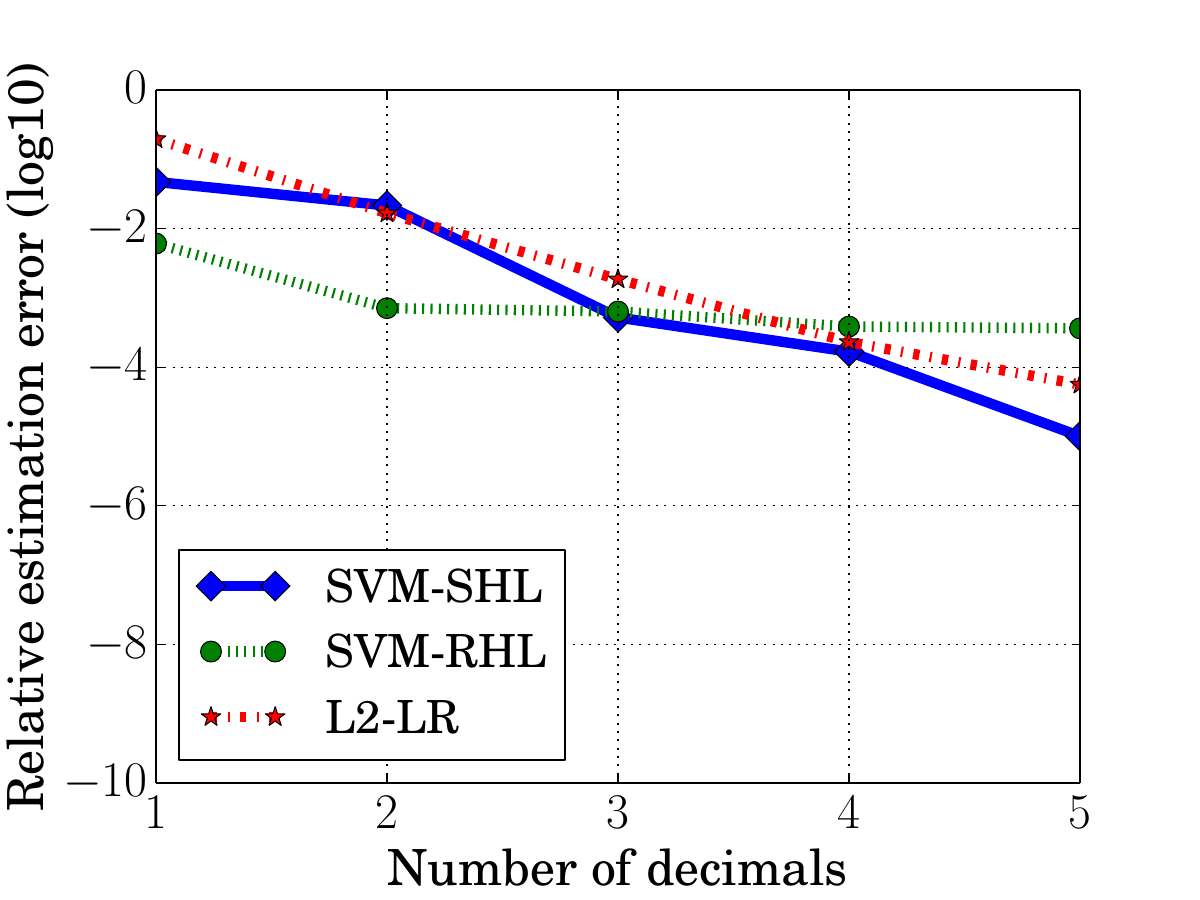}\label{lossfunction}}
%\subfloat[Digits]{\includegraphics[width=0.32\textwidth]{./src_defense/figs/clf/result-clf-mlr-digits-new-new.pdf}\label{clf_mlr_digits_defense}} 
\subfloat[]{\includegraphics[width=0.25\textwidth]{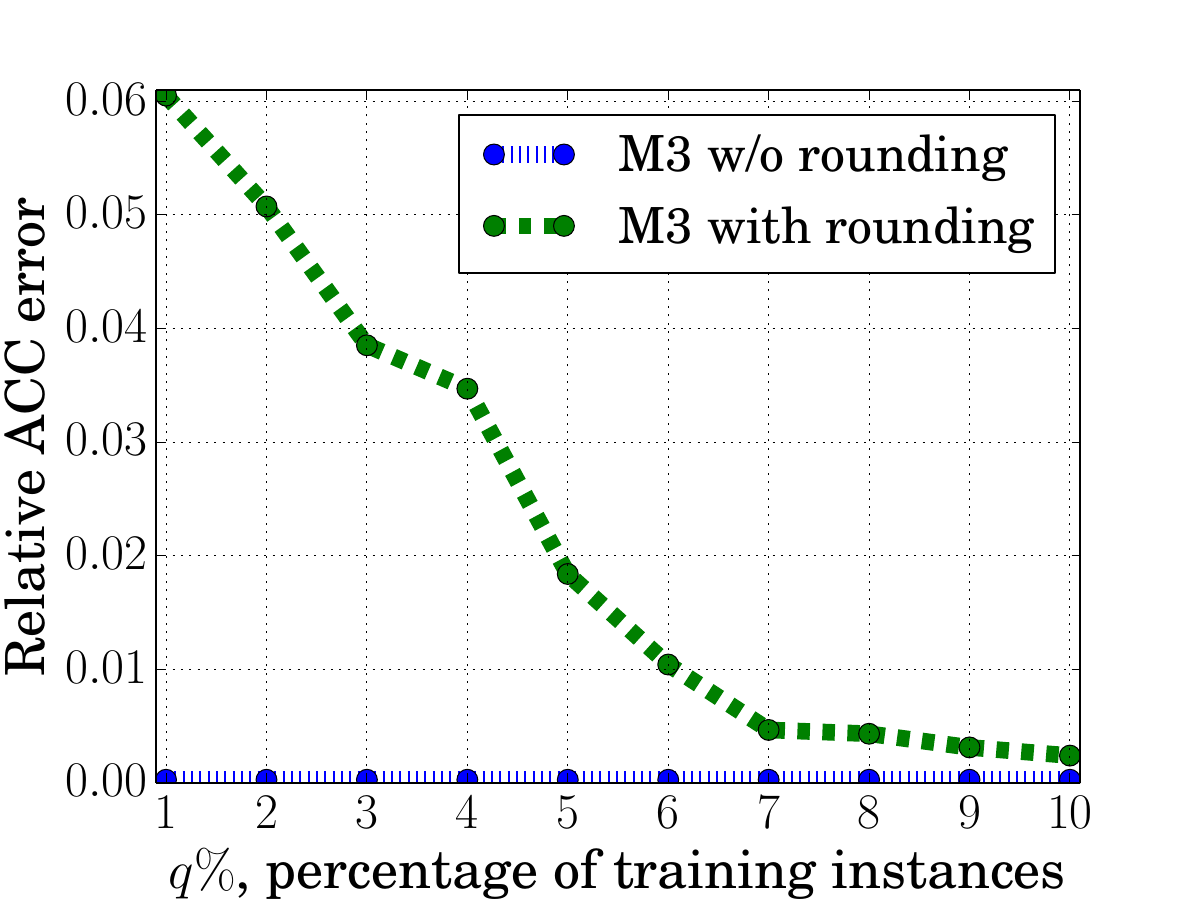}\label{svc_defense_m1_m3}\label{real_defense_reg}}
\caption{(a) Effectiveness of the rounding technique for different loss functions on the dataset Madelon. (b) Relative ACC error of M3 over M1 for SVM-SHL on the dataset Bank.}
\vspace{-4mm}
\end{figure}

\alan{
\subsection{Summary}
Through empirical evaluations, we have the following observations. First, rounding model parameters is not effective enough to prevent our attacks for certain ML algorithms. Second, $L_2$ regularization can more effectively defend against our attacks than  $L_1$ regularization. Third, cross entropy and square hinge loss have similar defense effectiveness. Moreover, they can more effectively defend against our attacks  than regular hinge loss. Fourth, the Train-Steal-Retrain strategy can still save lots of costs when MLaaS adopts rounding.
}

\section{Discussion}
\section{Conclusion and Future Work}

We demonstrate that various ML algorithms are vulnerable to hyperparameter stealing attacks. 
Our attacks encode the relationships between hyperparameters, model parameters, and training dataset into a system of linear equations, 
which is derived by setting the gradient of the objective function to be $\mathbf{0}$. 
Via both theoretical and empirical evaluations, we show that our attacks can accurately steal hyperparameters. 
Moreover, we find that rounding model parameters can increase the estimation errors of our attacks, with negligible 
impact on the testing performance of the model. However, for certain ML algorithms, our attacks still achieve very small estimation errors, highlighting the needs for new countermeasures.  
Future work includes studying security of other types of hyperparameters and new countermeasures. 

\noindent
{\bf Acknowledgement:} We thank the anonymous reviewers for their constructive comments. We also thank SigOpt for sharing a free API token. 
%The more accuate model to estimate the hyperparameter, the much easier to be defensed. 

\bstctlcite{IEEEexample:BSTcontrol}

\bibliographystyle{IEEEtranS}
\bibliography{refs}

% \balance
% {
% %\vspace{6mm}
% \bibliographystyle{IEEEtranS}
% \bibliography{refs}
% }
%\input{appendix}
%\newpage
%\appendix
\appendices

\begin{figure}
\center
%\vspace{-6mm}
\includegraphics[width=0.4\textwidth]{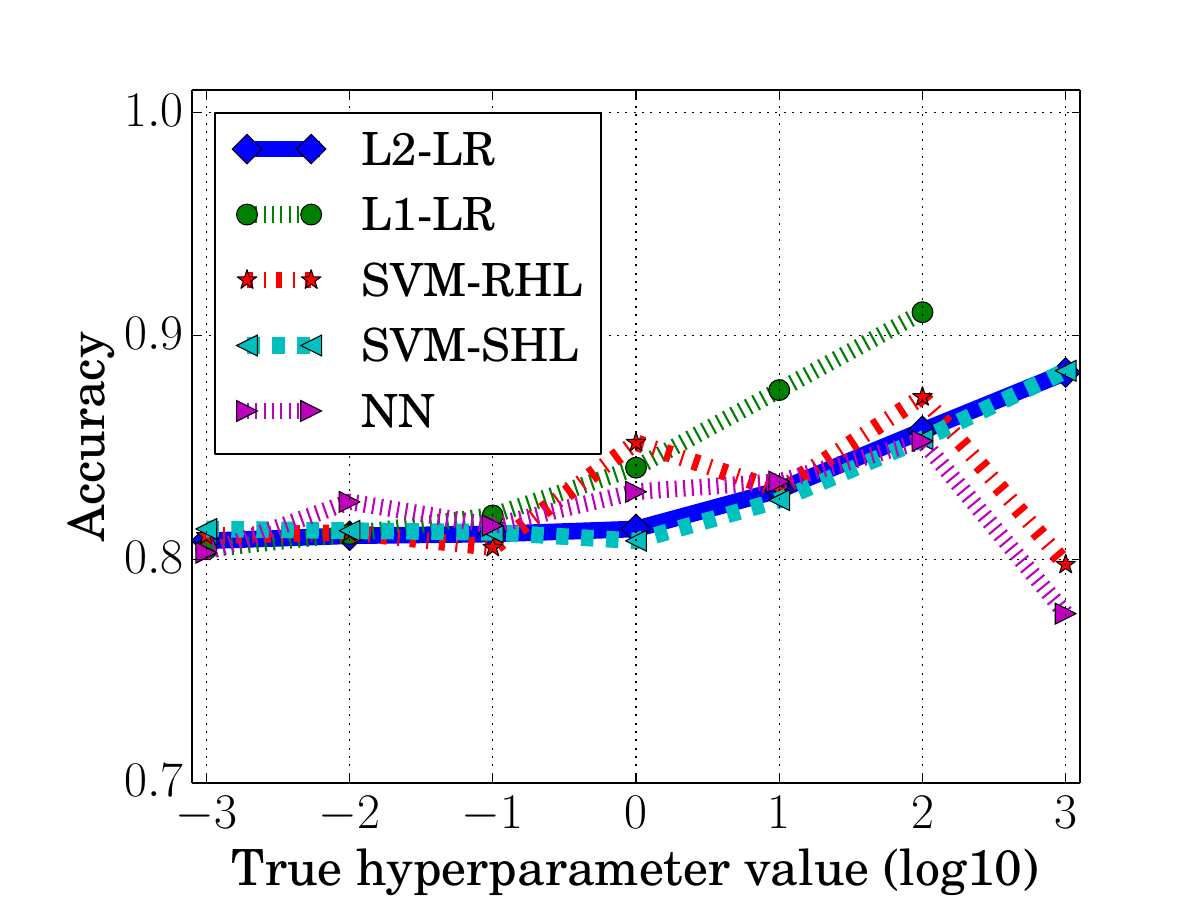}
\caption{Testing accuracy vs. hyperparameter (log10 scale) of classification algorithms on Madelon in Table~\ref{data_reg}.} 
%shown in Table~\ref{data_reg}.}
\label{accvshyper}
%\vspace{-8mm}
\end{figure}

\section{Attacks to Other Learning Algorithms}
\label{app:other}

\myparatight{LASSO}
%\label{app:lasso}
The objective function of LASSO is:
\begin{small}
\begin{align}
\label{lasso}
\mathcal{L}(\mathbf{w}) = \| \mathbf{y}-\mathbf{X}^T\mathbf{w} \|_2^2 + \lambda \| \mathbf{w} \|_1,
\end{align}
\end{small}
whose gradient  is:
\begin{small}
\begin{align*}
%\label{lasso_grad}
\frac{\partial \mathcal{L}(\mathbf{w})}{\partial \mathbf{w}} = -2 \mathbf{X y} + 2 \mathbf{X}\mathbf{X}^T\mathbf{w} + \lambda \text{sign}(\mathbf{w}),
%2  \mathbf{x}_i^T \mathbf{x}_i \mathbf{w}  - 2 \left\langle \mathbf{x}_i, \mathbf{r}_i \right\rangle + \lambda \partial \|\mathbf{w}\|_1 |_{w_i},
\end{align*}
\end{small}
where 
%the $i$th entry $s_i$ of the vector $\mathbf{s}$ is defined as:
%\begin{small}
%\begin{align}
%\label{vectors}
%s_i=\frac{\partial |w_i|}{\partial w_i}=
%\begin{cases}
%-1 & \text{if } w_i < 0 \\
%0 & \text{if } \, w_i = 0 \\
%1 & \text{if } w_i > 0
%\end{cases}
%\end{align}
%\end{small}
  $| w_i |$ is \emph{not differentiable} when $w_i=0$, so
we define the derivative at $w_i=0$ as 0, which means that we do not use the model parameters that are 0 to estimate the hyperparameter. By setting the gradient to be $\mathbf{0}$, we can estimate $\lambda$ using Eqn.~\ref{fianl_est} with $\mathbf{a}=\text{sign}(\mathbf{w})$ and $\mathbf{b}=-2 \mathbf{X y} + 2 \mathbf{X}\mathbf{X}^T\mathbf{w}$.
 
 We note that if $\lambda \geq \lambda_{\max} = \| \mathbf{X} \mathbf{y} \|_{\infty}$, then $\mathbf{w} = \mathbf{0}$. In such cases, we cannot estimate the exact hyperparameter. However, in practice,  $\lambda < \lambda_{\max}$ must hold in order to learn meaningful model parameters. 

\myparatight{$L_2$-regularized LR (L2-LR)} Its objective function is 
% The objective function of L2-LR is 
\begin{small}
\begin{align}
\label{l2lr}
\mathcal{L}(\mathbf{w}) = \text{L}(\mathbf{X}, \mathbf{y}, \mathbf{w}) + \lambda \| \mathbf{w} \|_2^2,
\end{align}
\end{small}
where $\text{L}(\mathbf{X}, \mathbf{y}, \mathbf{w})$=$-\sum_{i=1}^n  ( y_i \, \log h_\mathbf{w}(\mathbf{x}_i) + (1-y_i)$ $ \log (1- h_\mathbf{w}(\mathbf{x}_i)))$, which is called \emph{cross entropy} loss function.  $h_\mathbf{w}(\mathbf{x})$ is defined to be $\frac{1}{1+\exp{(- \mathbf{w}^T \mathbf{x})}} $. 
The gradient of the objective function with respect to $\mathbf{w}$ is:
\begin{small}
\begin{align*}
%\label{l2lr_grad}
\begin{split}
& \frac{\partial \mathcal{L}(\mathbf{w})}{\partial \mathbf{w}} =
% - \sum_{i=1}^n \left[ \frac{y_i}{h_\mathbf{w}(\mathbf{x}_i)} \frac{\partial h_\mathbf{w}(\mathbf{x}_i)}{\partial \mathbf{w}} + \frac{-(1-y_i)}{1-h_\mathbf{w}(\mathbf{x}_i)} \frac{\partial h_\mathbf{w}(\mathbf{x}_i)}{\partial \mathbf{w}}\right] + 2 \lambda \mathbf{w} \\
%& \quad = - \sum_{j=1}^n \left[ \frac{y_j }{h_\mathbf{w}(\mathbf{x}_i)} h_\mathbf{w}(\mathbf{x}_i) (1-h_\mathbf{w}(\mathbf{x}_i)) - \frac{(1-y_j)}{1-h_\mathbf{w}(\mathbf{x}_i)} h_\mathbf{w}(\mathbf{x}_i) (1-h_\mathbf{w}(\mathbf{x}_i)) \right] \frac{\partial \mathbf{w}^T \mathbf{x}^j}{\partial \mathbf{w}} + 2 \lambda \mathbf{w} \\
%& \quad = - \sum_{i=1}^n \left(y_i (1-h_\mathbf{w}(\mathbf{x}_i)) - (1-y_i) h_\mathbf{w}(\mathbf{x}_i) \right) \mathbf{x}_i + 2 \lambda \mathbf{w} \\ 
%& \quad = \sum_{i=1}^n (h_\mathbf{w}(\mathbf{x}_i) - y_i) \mathbf{x}_i + 2 \lambda \mathbf{w} = 
\mathbf{X} (h_\mathbf{w}(\mathbf{X})  - \mathbf{y}) + 2 \lambda \mathbf{w},
\end{split}
\end{align*}
\end{small}
where  $h_\mathbf{w}(\mathbf{X})  = \left[ h_\mathbf{w}(\mathbf{x}_1); h_\mathbf{w}(\mathbf{x}_2); \cdots; h_\mathbf{w}(\mathbf{x}_n) \right]$ is a vector. 
Via setting the gradient to be \textbf{0}, we can estimate $\lambda$ using Eqn.~\ref{fianl_est} with
$\mathbf{a} = 2 \mathbf{w}$ and $\mathbf{b} = \mathbf{X} (h_\mathbf{w}(\mathbf{X})   - \mathbf{y})$.

%\begin{small}
%\begin{align}
%\label{l1lr}
%\mathcal{L}(\mathbf{w}) = \sum_{i=1}^n \text{NLL}(\mathbf{w};\mathbf{x}_i,y_i) + \lambda \| \mathbf{w} \|_1,
%%\mathcal{L}_{L1-LR}(\lambda, \mathbf{w}) = - \sum_{i=1}^n \left[ y_j \log h_\mathbf{w}(\mathbf{x}_j) + (1-y_j) \log(1-h_\mathbf{w}(\mathbf{x}_j)) \right] + \lambda \| \mathbf{w} \|_1.
%\end{align}
%\end{small}
%where it uses cross entropy loss function and $L_1$ regularization. Generally, the optimal $\mathbf{w}$ can be approximately solved via LBFGS~\cite{liu1989limited}. 
%The subderivative of $\mathcal{L}(\mathbf{w})$ with respect to $\mathbf{w}$ is 
%\begin{small}
%\begin{align*}
%%\label{l1lr_grad}
%\frac{\partial \mathcal{L}_{L1-LR}}{\partial \mathbf{w}} = \mathbf{X} (h_\mathbf{w}(\mathbf{X}) - \mathbf{y}) + \lambda \partial_\mathbf{w} \| \mathbf{w} \|_1. 
%\end{align*}
%\end{small}
%Setting it to be \textbf{0} and removing zero $w_i$, we have
%$\mathbf{a} = \text{sign}(\mathbf{w}) $ and $\mathbf{b} = \mathbf{X} (h_\mathbf{w}(\mathbf{X}) - \mathbf{y}) $.

\myparatight{SVM with regular hinge loss (SVM-RHL)} 
The objective function of SVM-RHL is:
\begin{small}
\begin{align}
\label{hl_bsvc}
\mathcal{L}(\mathbf{w}) = \sum_{i=1}^n L(\mathbf{x}_i, y_i, \mathbf{w}) + \lambda \| \mathbf{w} \|_2^2, 
\end{align}
\end{small}
where $L(\mathbf{x}_i, y_i, \mathbf{w}) = \max(0, 1 - y_i  \mathbf{w}^T \mathbf{x}_i)$ is called regular hinge loss function.  
The gradient with respect to $\mathbf{w}$ is: 
\begin{small}
\begin{align*}
%\label{RHL_subg}
\frac{\partial L}{\partial \mathbf{w}} = 
\begin{cases}
-y_i \mathbf{x}_i & \text{if } y_i  \mathbf{w}^T \mathbf{x}_i  < 1 \\
%[-y_j \mathbf{x}_j, \mathbf{0}] & \text{if } y_j \langle \mathbf{w}, \mathbf{x}_j \rangle = 1 \\
\mathbf{0} & \text{if } y_i  \mathbf{w}^T \mathbf{x}_i  > 1,
\end{cases}
\end{align*}
\end{small}
where $L(\mathbf{x}_i, y_i, \mathbf{w})$ is non-differentiable at the point where  $y_i  \mathbf{w}^T \mathbf{x}_i  = 1$.
 We estimate $\lambda$ using only training instances $\mathbf{x}_i$ that satisfy $y_i  \mathbf{w}^T \mathbf{x}_i  < 1$. 
Specifically, we estimate $\lambda$ using Eqn.~\ref{fianl_est} with $\mathbf{a} = 2 \mathbf{w}$ and $\mathbf{b} = \sum_{i=1}^n -y_i \mathbf{x}_i \mathbf{1}_{y_i  \mathbf{w}^T \mathbf{x}_i < 1}$, where $\mathbf{1}_{y_i  \mathbf{w}^T \mathbf{x}_i < 1}$ is an indicator function with value 1 if $ y_i  \mathbf{w}^T \mathbf{x}_i < 1 $ and 0 otherwise. 

\myparatight{SVM with square hinge loss (SVM-SHL)}  
The objective function of SVM-SHL is:
\begin{small}
\begin{align}
\label{hl_bsvc}
\mathcal{L}(\mathbf{w}) = \sum_{i=1}^n L(\mathbf{x}_i, y_i, \mathbf{w}) + \lambda \| \mathbf{w} \|_2^2, 
\end{align}
\end{small}
where $L(\mathbf{x}_i, y_i, \mathbf{w}) = \max(0, 1 - y_i  \mathbf{w}^T \mathbf{x}_i)^2$ is called square hinge loss function.  
The gradient with respect to $\mathbf{w}$ is: 
\begin{small}
\begin{align*}
%\label{RHL_subg}
\frac{\partial L}{\partial \mathbf{w}} = 
\begin{cases}
-2y_i \mathbf{x}_i (1 - y_i  \mathbf{w}^T \mathbf{x}_i) & \text{if } y_i  \mathbf{w}^T \mathbf{x}_i  \leq 1 \\
%[-y_j \mathbf{x}_j, \mathbf{0}] & \text{if } y_j \langle \mathbf{w}, \mathbf{x}_j \rangle = 1 \\
\mathbf{0} & \text{if } y_i  \mathbf{w}^T \mathbf{x}_i  > 1.
\end{cases}
\end{align*}
\end{small}
  
Therefore, we estimate $\lambda$ using Eqn.~\ref{fianl_est} with $\mathbf{a} =  \mathbf{w}$ and $\mathbf{b} = \sum_{i=1}^n - y_i \mathbf{x}_i (1 - y_i  \mathbf{w}^T \mathbf{x}_i) \mathbf{1}_{y_i  \mathbf{w}^T \mathbf{x}_i \leq 1}$.

%$L_{SHL}(y_j, \langle \mathbf{w}, \mathbf{x}_j \rangle) = \max(0, 1 - y_j \langle \mathbf{w}, \mathbf{x}_j \rangle)^2$. SHL is convex, smooth, and differential. Its subgradient with respect to $\mathbf{w}$ is given by
%\begin{small}
%\begin{align*}  
%%\label{SHL_subg}
%\frac{\partial L_{SHL}}{\partial \mathbf{w}} =
%\begin{cases}
%-2 y_j \mathbf{x}_j (1 - y_j \langle \mathbf{w}, \mathbf{x}_j \rangle) & \text{if } y_j \langle \mathbf{w}, \mathbf{x}_j \rangle <= 1 \\
%\mathbf{0}  & \text{if } y_j \langle \mathbf{w}, \mathbf{x}_j \rangle > 1
%\end{cases}
%\end{align*}
%\end{small}
%Therefore, for SVM-SHL, we have $\mathbf{a} = \mathbf{w}$ and $\mathbf{b} = \sum_{j=1}^n -2 y_j \mathbf{x}_j (1 - y_j \langle \mathbf{w}, \mathbf{x}_j \rangle) \mathbf{1}_{y_j \langle \mathbf{w}, \mathbf{x}_j \rangle \leq 1}$.

%\myparatight{SVM with square hinge loss (SVM-SHL)}

\myparatight{$L_1$-regularized kernel LR (L1-KLR)} Its objective function is:
% The objective function of L1-KLR is:
\begin{small}
\begin{align}
\label{l2klr_dual}
\begin{split}
\mathcal{L}(\bm{\alpha}) =\text{L}(\mathbf{X}, \mathbf{y}, \mathbf{\bm \alpha}) + 
\lambda \| \mathbf{K} \bm{\alpha} \|_1,
\end{split}
\end{align} 
\end{small}
where $\text{L}(\mathbf{X}, \mathbf{y}, \mathbf{\bm \alpha}) = $ $-\sum_{i=1}^n ( y_i \, \log h_\mathbf{\bm \alpha}(\mathbf{x}_i) + (1-y_i)$ $ \log (1- h_\mathbf{\bm \alpha}(\mathbf{x}_i)))$ and $h_\mathbf{\bm \alpha}(\mathbf{x}) = \frac{1}{1 + \exp{(- \sum_{j=1}^n  \alpha_j \phi(\mathbf{x})^T\phi(\mathbf{x}_j))}}$. 
%$\bm{\alpha}$ is often solved approximately via coordinate descent.  
The gradient of the objective function with respect to $\bm{\alpha}$ is:
\begin{small}
\begin{align*}
%\label{l2klr_soln_1}
\begin{split}
\frac{\partial \mathcal{L}(\bm{\alpha})}{\partial \bm{\alpha}} 
%& = \sum_{j=1}^n (h_{\bm{\alpha}}(\mathbf{k}_j) - y_j) \mathbf{k}_j + 2 \lambda \mathbf{K} \bm{\alpha} \\
& = \mathbf{K} (h_\mathbf{\bm \alpha}(\mathbf{X}) - \mathbf{y} +  \lambda \mathbf{t}),
\end{split}
\end{align*}
\end{small}
where $h_\mathbf{\bm \alpha}(\mathbf{X}) = \left[ h_\mathbf{\bm \alpha}(\mathbf{x}_1); h_\mathbf{\bm \alpha}(\mathbf{x}_2); \cdots; h_\mathbf{\bm \alpha}(\mathbf{x}_n) \right]$ 
and $\mathbf{t}=\text{sign}(\mathbf{K} {\bm \alpha})$.
Via setting the gradient to be \textbf{0} and considering that $\mathbf{K}$ is invertible, we can estimate $\lambda$ using Eqn.~\ref{fianl_est} with  
$\mathbf{a} = \mathbf{t}$ and $\mathbf{b} = h_\mathbf{\bm \alpha}(\mathbf{X}) - \mathbf{y}$.

\myparatight{$L_2$-regularized kernel LR (L2-KLR)} Its objective function of L1-KLR is: 
% The objective function of L2-KLR is:
\begin{small}
\begin{align}
\label{l2klr_dual}
\begin{split}
\mathcal{L}(\bm{\alpha}) =\text{L}(\mathbf{X}, \mathbf{y}, \mathbf{\bm \alpha}) + 
\lambda \bm{\alpha}^T \mathbf{K} \bm{\alpha},
\end{split}
\end{align} 
\end{small}
where $\text{L}(\mathbf{X}, \mathbf{y}, \mathbf{\bm \alpha}) = $ $-\sum_{i=1}^n ( y_i \, \log h_\mathbf{\bm \alpha}(\mathbf{x}_i) + (1-y_i)$ $ \log (1- h_\mathbf{\bm \alpha}(\mathbf{x}_i)))$ is a cross entropy loss function and $h_\mathbf{\bm \alpha}(\mathbf{x}) = \frac{1}{1 + \exp{(- \sum_{j=1}^n  \alpha_j \phi(\mathbf{x})^T\phi(\mathbf{x}_j))}}$. 
%$\bm{\alpha}$ is often solved approximately via coordinate descent.  
The gradient of the objective function with respect to $\bm{\alpha}$ is:
\begin{small}
\begin{align*}
%\label{l2klr_soln_1}
\begin{split}
\frac{\partial \mathcal{L}(\bm{\alpha})}{\partial \bm{\alpha}} 
%& = \sum_{j=1}^n (h_{\bm{\alpha}}(\mathbf{k}_j) - y_j) \mathbf{k}_j + 2 \lambda \mathbf{K} \bm{\alpha} \\
& = \mathbf{K} (h_\mathbf{\bm \alpha}(\mathbf{X}) - \mathbf{y} + 2 \lambda \bm{\alpha}),
\end{split}
\end{align*}
\end{small}
where $h_\mathbf{\bm \alpha}(\mathbf{X})= \left[ h_\mathbf{\bm \alpha}(\mathbf{x}_1); h_\mathbf{\bm \alpha}(\mathbf{x}_2); \cdots; h_\mathbf{\bm \alpha}(\mathbf{x}_n) \right]$.
% and $\mathbf{K}^T = \mathbf{K}$.
Via setting the gradient to be \textbf{0} and considering that $\mathbf{K}$ is invertible, we can estimate $\lambda$ using Eqn.~\ref{fianl_est} with  
$\mathbf{a} = 2 \bm{\alpha}$ and $\mathbf{b} = h_\mathbf{\bm \alpha}(\mathbf{X}) - \mathbf{y}$.

%\myparatight{Kernel SVM with regular hinge loss (KSVM-RHL)}
%The objective function of KSVM-RHL  is:
%%\begin{small}
%%\begin{align*}
%%%\label{kernel_bsvc_prime}
%%\mathcal{L}_{KSVM}(\lambda, \mathbf{w^\prime})  = \lambda \sum_{i=1}^n L_{HL}(y_j, \langle \mathbf{w^\prime}, \phi(\mathbf{x}_j) \rangle) + \frac{1}{2} \| \mathbf{w^\prime} \|_2^2. 
%%\end{align*}
%%\end{small}
%%Using $\mathbf{w}^\prime = \Phi(\mathbf{X})^\lambda \bm{\alpha}$, we have its dual form
%\begin{small}
%\begin{align}
%\label{kernel_bsvc_dual}
%\begin{split}
%\mathcal{L}(\bm{\alpha}) = \sum_{i=1}^n L(\mathbf{x}_i, y_i,  \bm{\alpha}) + \lambda \bm{\alpha}^T \mathbf{K} \bm{\alpha},
%\end{split}
%\end{align}
%\end{small}
%where $L(\mathbf{x}_i, y_i,  \bm{\alpha}) = \max(0, 1 - y_i  \bm{\alpha}^T \mathbf{k}_i )$, where $\mathbf{k}_i$ is the $i$th column of the gram matrix $\mathbf{K}$.
%Following the same methodology we used for SVM-RHL,  
%we estimate $\lambda$ using Eqn.~\ref{fianl_est} with $\mathbf{a} = 2 \mathbf{K} \bm{\alpha}$ and $\mathbf{b} = \sum_{i=1}^n -y_i \mathbf{k}_i \mathbf{1}_{y_i  \bm{\alpha}^T \mathbf{k}_i < 1}$. 

%we estimate $\lambda$ for each $\alpha_i$ as 
%\begin{footnotesize}
%\begin{align}
%\label{kernel_bsvc_rhl_soln}
%\hat{\lambda}^{(i)}_{KSVM-RHL} = \frac{\mathbf{k}^i \bm{\alpha}}{\sum\limits_{j, y_j \langle \bm{\alpha}, \mathbf{k}_j \rangle <1} y_j k_{j,i}}.
%\end{align} 
%\end{footnotesize}
%%\vspace{-5mm}

\myparatight{Kernel SVM with square hinge loss (KSVM-SHL)}
The objective function of KSVM-SHL  is:
%\begin{small}
%\begin{align*}
%%\label{kernel_bsvc_prime}
%\mathcal{L}_{KSVM}(\lambda, \mathbf{w^\prime})  = \lambda \sum_{i=1}^n L_{HL}(y_j, \langle \mathbf{w^\prime}, \phi(\mathbf{x}_j) \rangle) + \frac{1}{2} \| \mathbf{w^\prime} \|_2^2. 
%\end{align*}
%\end{small}
%Using $\mathbf{w}^\prime = \Phi(\mathbf{X})^\lambda \bm{\alpha}$, we have its dual form
\begin{small}
\begin{align}
\label{kernel_bsvc_dual}
\begin{split}
\mathcal{L}(\bm{\alpha}) = \sum_{i=1}^n L(\mathbf{x}_i, y_i,  \bm{\alpha}) + \lambda \bm{\alpha}^T \mathbf{K} \bm{\alpha},
\end{split}
\end{align}
\end{small}
where $L(\mathbf{x}_i, y_i,  \bm{\alpha}) = \max(0, 1 - y_i  \bm{\alpha}^T \mathbf{k}_i )^2$.
Following the same methodology we used for SVM-SHL,  
we estimate $\lambda$ using Eqn.~\ref{fianl_est} with $\mathbf{a} = \mathbf{K} \bm{\alpha}$ and $\mathbf{b} = \sum_{i=1}^n -y_i \mathbf{k}_i (1-y_i  \bm{\alpha}^T \mathbf{k}_i) \mathbf{1}_{y_i  \bm{\alpha}^T \mathbf{k}_i \leq 1}$. 

\begin{figure*}
\center
%\vspace{-4mm}
\subfloat[Attack]{\includegraphics[width=0.45\textwidth]{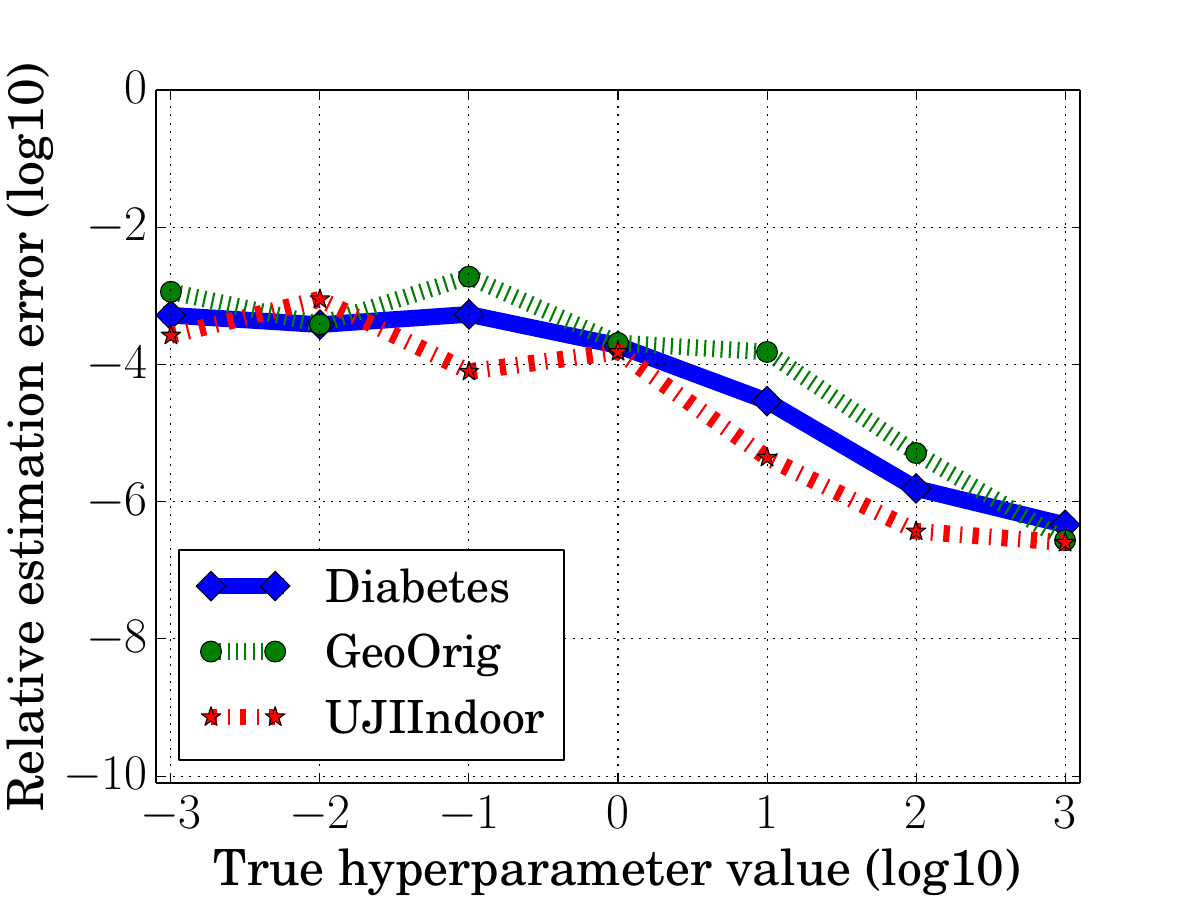}\label{ent_attack}} 
\subfloat[Defense]{\includegraphics[width=0.45\textwidth]{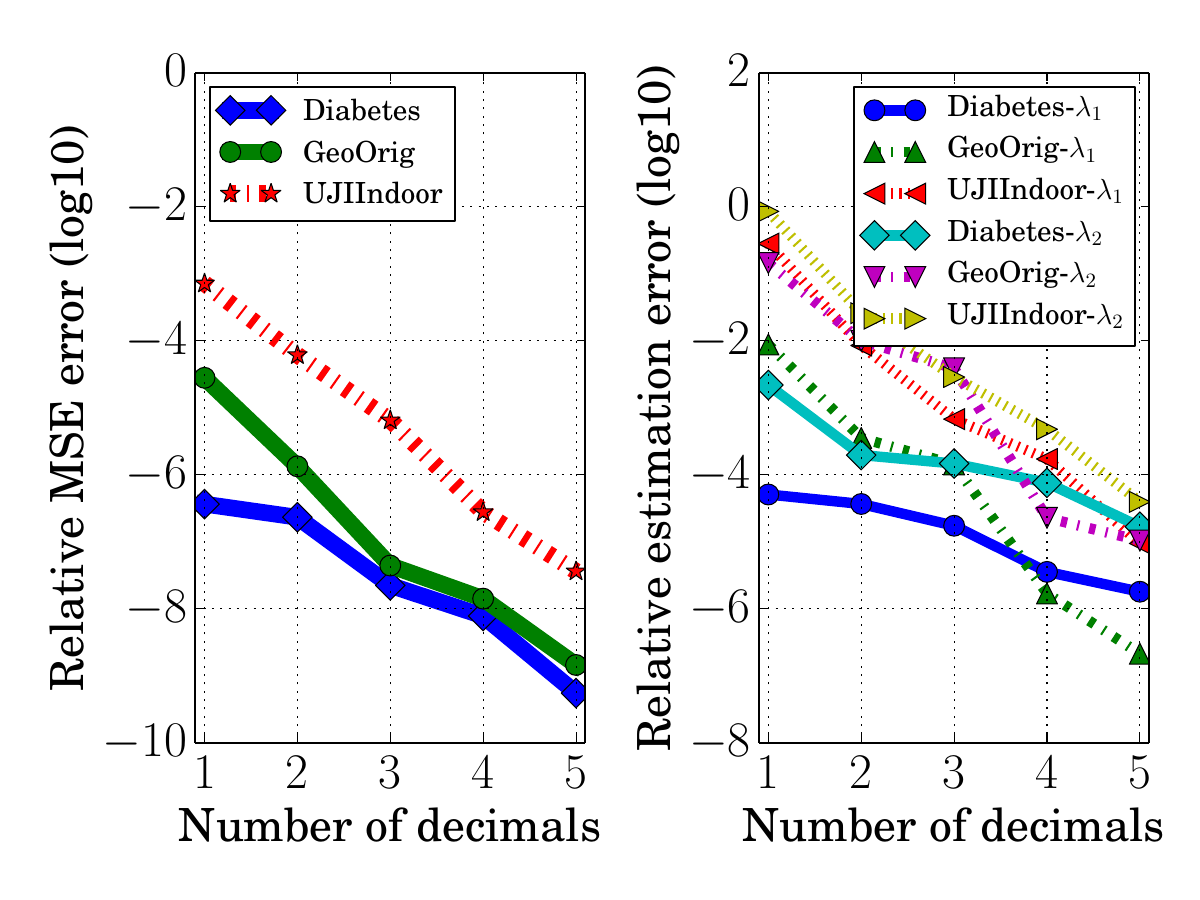}\label{enet_defense}}
%\subfloat[Iris]{\includegraphics[width=0.32\textwidth]{./src_defense/figs/clf/result-clf-mlr-iris-new-new.pdf}\label{mlr_ionosphere_defense}}
\caption{Attack and defense results of ENet.}
\label{enet_res}
%\vspace{-6mm}
\end{figure*}

%\myparatight{SHL} $L_{SHL}(y_j, \langle \bm{\alpha}, \mathbf{k}_j \rangle) = \max(0, 1 - y_j \langle \bm{\alpha}, \mathbf{k}_j \rangle )^2$. We take the derivative of $\mathcal{L}(\bm{\alpha})$ for $\bm{\alpha}$ and set it to be $\mathbf{0}$. Then, for KSVM-SHL, we have
%$\mathbf{a} =\mathbf{K} \bm{\alpha}$ and $\mathbf{b} = \sum_{j=1}^n -2 y_j \mathbf{k}_j (1 - y_j \langle \bm{\alpha}, \mathbf{k}_j \rangle) \mathbf{1}_{y_j \langle \bm{\alpha}, \mathbf{k}_j \rangle \leq 1}$.
%
%
%
%\myparatight{Kernel SVM with regular hinge loss (KSVM-RHL)}
%
%\myparatight{Kernel SVM with square hinge loss (KSVM-SHL)}
%

\section{More than One Hyperparameter}
\label{app:enet}
We use a popular regression algorithm called \emph{Elastic Net (ENet)}~\cite{zou2005regularization} as an example to illustrate how we can apply attacks 
to learning algorithms with more than one hyperparameter. 
The objective function of ENet is: 
%\subsubsection{Elastic Net: More general case}
\begin{small}
\begin{align}
\label{enet}
\mathcal{L}(\mathbf{w}) = \| \mathbf{y}-\mathbf{X}^T\mathbf{w} \|_2^2 + \lambda_1 \| \mathbf{w} \|_1 + \lambda_2 \| \mathbf{w} \|_2^2,
\end{align}
\end{small}
where the loss function is least square and regularization term is the combination of $L_2$ regularization and $L_1$ regularization. 
We compute the gradient as follows: 
\begin{small}
\begin{align*}
%\label{lasso_grad}
\frac{\partial \mathcal{L}(\mathbf{w})}{\partial \mathbf{w}} = -2 \mathbf{X y} + 2 \mathbf{X}\mathbf{X}^T\mathbf{w} + \lambda_1 \text{sign}(\mathbf{w}) + 2 \lambda_2 \mathbf{w} \odot |\text{sign}(\mathbf{w})|,
%2  \mathbf{x}_i^T \mathbf{x}_i \mathbf{w}  - 2 \left\langle \mathbf{x}_i, \mathbf{r}_i \right\rangle + \lambda \partial \|\mathbf{w}\|_1 |_{w_i},
\end{align*}
\end{small}
where $\mathbf{w} \odot |\text{sign}(\mathbf{w})|=[w_1|\text{sign}(\mathbf{w}_1)|; \cdots;w_m|\text{sign}(\mathbf{w}_m)|]$. 
Similar to LASSO, we do not use the model parameters that are 0 to estimate the hyperparameters. Via setting the gradient to $\mathbf{0}$ and using the linear least square to solve the overdetermined system, we have:
\begin{small}
\begin{align}
\hat{\bm{\lambda}} = -\left( \mathbf{A}^T \mathbf{A} \right)^{-1} \mathbf{A}^T \mathbf{b},
\end{align}
\end{small}
where $\hat{\bm{\lambda}}=[\hat{\lambda}_1;\hat{\lambda}_2]$,  $\mathbf{A}=[\text{sign}(\mathbf{w}); 2\mathbf{w} \odot |\text{sign}(\mathbf{w})|]$, and $\mathbf{b}=-2 \mathbf{X y} + 2 \mathbf{X}\mathbf{X}^T\mathbf{w}$.

Figure~\ref{enet_res} shows the attack and defense results for ENet on the three regression datasets. We observe that our attacks are effective for learning algorithms with more than one hyperparameter, and rounding is also an effective defense.

%on the dataset XX. XXXX

\section{Neural Network (NN)}
\label{neuralNet}
We evaluate attack and defense on a three-layer neural network (NN) for both regression and classification.

\myparatight{Regression} The objective function of the three-layer NN for regression is defined as
\begin{small}
\begin{align}
%\label{nn_reg}
\mathcal{L}(\mathbf{W}_1, \mathbf{w}_2) = \| \mathbf{y}- \hat{\mathbf{y}} \|_2^2 + \lambda \big( \| \mathbf{W}_1 \|_F^2  + \| \mathbf{w}_2 \|_2^2 \big),
\end{align}
\end{small}
where $ \hat{\mathbf{y}} = \textbf{sig} \big( \mathbf{X}^T  \mathbf{W}_1 + \mathbf{b}_1 \big) \mathbf{w}_2  + \mathbf{b}_2$; $\textbf{sig}(\mathbf{A}) = \frac{1}{1+ \exp(-\mathbf{A})}$ is the logistic function; $\mathbf{W}_1 \in \mathbb{R}^{m \times d}$ is the weight matrix of input layer to hidden layer and $\mathbf{w}_2 \in \mathbf{R}^{d}$ is the weight vector of hidden layer to output layer; $d$ is the number of hidden units; $\mathbf{b}_1$ and $\mathbf{b}_2$ are the bias terms of the two layers. 

To perform our hyperparameter stealing attack, we can leverage the gradient of $\mathcal{L}(\mathbf{W}_1, \mathbf{w}_2)$ with respect to either $\mathbf{W}_1$ or $\mathbf{w}_2$. For simplicity, we use $\mathbf{w}_2$. Specifically, the gradient of the objective function of $\mathbf{w}_2$ is:
\begin{small}
\begin{align*}
%\label{lasso_grad}
\frac{\partial \mathcal{L}(\mathbf{W}_1, \mathbf{w}_2)}{\partial \mathbf{w}_2} = 2 \textbf{sig} \big( \mathbf{X}^T  \mathbf{W}_1 + \mathbf{b}_1 \big)^T (\mathbf{y}- \hat{\mathbf{y}}) + 2 \lambda \mathbf{w}_2.
%2  \mathbf{x}_i^T \mathbf{x}_i \mathbf{w}  - 2 \left\langle \mathbf{x}_i, \mathbf{r}_i \right\rangle + \lambda \partial \|\mathbf{w}\|_1 |_{w_i},
\end{align*}
\end{small}
By setting the gradient to be $\mathbf{0}$, we can estimat $\lambda$ using Eqn.~\ref{fianl_est} with $\mathbf{a} = \mathbf{w}_2 $ and $\mathbf{b} = \textbf{sig} \big( \mathbf{X}^T  \mathbf{W}_1 + \mathbf{b}_1 \big)^T (\mathbf{y}- \hat{\mathbf{y}})$.

\myparatight{Classification} The objective function of the three-layer NN for binary classification is defined as
\begin{small}
\begin{align}
%\label{nn_clf}
\mathcal{L}(\mathbf{W}_1, \mathbf{w}_2) &= - \sum_{i=1}^n (y_i \log \hat{y}_i + (1-y_i) \log (1 - \hat{y}_i)) \nonumber \\
&+ \frac{\lambda}{2} \big( \| \mathbf{W}_1 \|_F^2  + \| \mathbf{w}_2 \|_2^2 \big),
\end{align}
\end{small}
where $\hat{y}_i = \textbf{sig} \big( \mathbf{w}_2^T \textbf{sig} \big( \mathbf{W}_1^T \mathbf{x}_i + \mathbf{b}_1 \big)  + {b}_2 \big) $.
Similarly with regression, we use $\mathbf{w}_2$ to steal the hyperparameter $\lambda$. The gradient of the objective function of $\mathbf{w}_2$ is:
\begin{small}
\begin{align*}
%\label{lasso_grad}
\frac{\partial \mathcal{L}(\mathbf{W}_1, \mathbf{w}_2)}{\partial \mathbf{w}_2} = - \sum_{i=1}^n (y_i - \hat{y}_i) \textbf{sig} \big( \mathbf{W}_1^T \mathbf{x}_i + \mathbf{b}_1 \big) + \lambda \mathbf{w}_2.
%2  \mathbf{x}_i^T \mathbf{x}_i \mathbf{w}  - 2 \left\langle \mathbf{x}_i, \mathbf{r}_i \right\rangle + \lambda \partial \|\mathbf{w}\|_1 |_{w_i},
\end{align*}
\end{small}
Setting the gradient to be $\mathbf{0}$, we can estimat $\lambda$ via Eqn.~\ref{fianl_est} with $\mathbf{a} =  \mathbf{w}_2 $ and $\mathbf{b} = - \sum_{i=1}^n (y_i - \hat{y}_i) \textbf{sig} \big( \mathbf{W}_1^T \mathbf{x}_i + \mathbf{b}_1 \big) $.

\section{Proof of Theorem 5.1}
\label{app:thm_1}

When $\mathbf{w}$ is an exact minimum of the objective function, we have $\mathbf{b}=-{\lambda} {\mathbf{a}}$.
%\begin{align*}
%\frac{\partial \mathcal{L}(\mathbf{w})}{\mathbf{w}} = \mathbf{0} \Longrightarrow {\mathbf{b}} + {\lambda} {\mathbf{a}} = \mathbf{0}.
%\end{align*}
%Additionally, we know $\mathbf{b} + \lambda^{\star} \mathbf{a} = \mathbf{0}$. 
%Since $\tilde{\mathbf{w}^{\star}} = \mathbf{w}^{\star}$, thus $\tilde{\lambda} = \lambda^{\star}$.
Therefore, we have:  
\begin{align*}
\hat{\lambda} & = - \big( \mathbf{a}^T \mathbf{a} \big)^{-1} \mathbf{a}^T \mathbf{b} = - \big( \mathbf{a}^T \mathbf{a} \big)^{-1} \mathbf{a}^T (- \lambda \mathbf{a}) \\
			& = \lambda \big( \mathbf{a}^T \mathbf{a} \big)^{-1} \mathbf{a}^T \mathbf{a} = \lambda. 
\end{align*}
\section{Proof of Theorem 5.2}
\label{app:thm_2}
We prove the theorem for linear learning algorithms, as it is similar for kernel learning algorithms. 
We treat our estimated hyperparameter $\hat{\lambda}$ as a function of model parameters.
We expand $\hat{\lambda}(\mathbf{w}^{\star} + \Delta \mathbf{w})$ at $\mathbf{w}^{\star}$ using Taylor expansion:
\begin{align*}
 \hat{\lambda}(\mathbf{w}^{\star} + \Delta \mathbf{w}) &= \hat{\lambda}(\mathbf{w}^{\star}) + \Delta \mathbf{w}^T \nabla \hat{\lambda}(\mathbf{w}^{\star}) \\
 & + \frac{1}{2} \Delta \mathbf{w}^T \nabla^2 \hat{\lambda}(\mathbf{w}^{\star}) \Delta \mathbf{w} +\cdots \\
 &= \hat{\lambda}(\mathbf{w}^{\star}) + \Delta \mathbf{w}^T \nabla \hat{\lambda}(\mathbf{w}^{\star}) + O(\|\Delta \mathbf{w}\|_2^2) \\
 &= \lambda + \Delta \mathbf{w}^T \nabla \hat{\lambda}(\mathbf{w}^{\star}) + O(\|\Delta \mathbf{w}\|_2^2)
 %\frac{1}{2} \Delta (\mathbf{w}^{\star})^T \mathbf{H}(\mathbf{w}^{\star}) \Delta \mathbf{x} + O(\|\Delta \mathbf{w}\|_2^2).
\end{align*}
Therefore, $\Delta \lambda$=$ \hat{\lambda}(\mathbf{w}^{\star} + \Delta \mathbf{w}) - \lambda$ =$\Delta \mathbf{w}^T \nabla \hat{\lambda}(\mathbf{w}^{\star}) + O(\|\Delta \mathbf{w}\|_2^2)$.

\section{Approximations of Gradients}
\label{approx}

We approximate the gradient $\nabla \hat{\lambda}(\mathbf{w}^{\star})$ in Theorem~\ref{thm_2} for RR, LASSO, L2-LR, L2-KLR, L1-LR, and L1-KLR.
According to the definition of gradient, we have:
%Now, we calculate each ML algorithm the gradient $\nabla \hat{\lambda}(\mathbf{w}^{\star})$. Mathematically, it is defined as
\begin{footnotesize}
\begin{align*}
\nabla \hat{\lambda}(\mathbf{w}^{\star}) = \lim_{\Delta \mathbf{w} \rightarrow \mathbf{0}} \frac{\hat{\lambda}(\mathbf{w}^{\star} + \Delta \mathbf{w}) - \hat{\lambda}(\mathbf{w}^{\star})}{\Delta \mathbf{w}},
\end{align*}
\end{footnotesize}
where the division and limit are component-wise for $\Delta \mathbf{w}$.

%Consider that the exact form of $\nabla \hat{\lambda}(\mathbf{w}^{\star})$ of each algorithm is complicated, we thus use some approximation as $\Delta \mathbf{w}$ is sufficiently small. 

%We informally define the sensitive of $\hat{\lambda}$ with respect to $\mathbf{w}^{\star}$ as $ \frac{\Delta \hat{\lambda}}{\Delta \mathbf{w}}$.  

\myparatight{RR} 
We approximate the gradient as follows:
% \begin{small}
% \begin{align*}
% \hat{\lambda}_{RR}(\mathbf{w}^{\star}) = \frac{(\mathbf{w}^{\star})^T \mathbf{X} (\mathbf{y} - \mathbf{X}^T \mathbf{w}^{\star})}{(\mathbf{w}^{\star})^T \mathbf{w}^{\star}}. 
% \end{align*}
% \end{small}
% Therefore, 
\begin{footnotesize}
\begin{align*}
&  \nabla \hat{\lambda}_{RR}(\mathbf{w}^{\star})  =  \lim_{\Delta \mathbf{w} \to \mathbf{0}} \frac{\hat{\lambda}_{RR}(\mathbf{w}^{\star} + \Delta \mathbf{w}) - \hat{\lambda}_{RR}(\mathbf{w}^{\star})}{\Delta \mathbf{w}} \\
& =  \lim_{\Delta \mathbf{w} \to \mathbf{0}}  \frac{1}{\Delta \mathbf{w}} \Big( \frac{(\mathbf{w}^{\star} + \Delta \mathbf{w})^T (\mathbf{Xy} - \mathbf{X}\mathbf{X}^T (\mathbf{w}^{\star}+\Delta \mathbf{w}))} {(\mathbf{w}^{\star}+\Delta \mathbf{w})^T (\mathbf{w}^{\star}+\Delta \mathbf{w})} \\ 
& \qquad - \frac{(\mathbf{w}^{\star})^T (\mathbf{Xy} - \mathbf{X}\mathbf{X}^T \mathbf{w}^{\star})}{(\mathbf{w}^{\star})^T \mathbf{w}^{\star}} \Big) \\
& \approx \lim_{\Delta \mathbf{w} \to \mathbf{0}}  \frac{1}{\Delta \mathbf{w}} \Big( \frac{\Delta \mathbf{w}^T (\mathbf{Xy} - 2 \mathbf{X} \mathbf{X}^T \mathbf{w}^{\star}) - \Delta \mathbf{w}^T \mathbf{X} \mathbf{X}^T \Delta \mathbf{w}}{(\mathbf{w}^{\star})^T \mathbf{w}^{\star}} \Big) \\
& \approx  \frac{\mathbf{X} (\mathbf{y}- 2 \mathbf{X}^T \mathbf{w}^{\star})} {\|\mathbf{w}^{\star}\|_2^2},
\end{align*}
\end{footnotesize}
where in the third and fourth equations, we use $(\mathbf{w}^{\star} + \Delta \mathbf{w})^T (\mathbf{w}^{\star} + \Delta \mathbf{w}) \approx (\mathbf{w}^{\star})^T \mathbf{w}^{\star} $ and $\Delta \mathbf{w}^T \mathbf{X} \mathbf{X}^T \Delta \mathbf{w} \approx 0$ for sufficiently small $\Delta \mathbf{w}$, respectively.
%Therefore $ \nabla \hat{\lambda}_{RR}(\mathbf{w}^{\star}) = O(\frac{1}{\|\mathbf{w}^{\star}\|_2^2})$.

%\myparatight{KRR}
%Similar to RR, we approximate the gradient as:
%\begin{footnotesize}
%\begin{align*}
%\nabla \hat{\lambda}_{KRR}(\bm{\alpha}^{\star}) \approx \frac{\mathbf{K} (\mathbf{y}- 2 \mathbf{K}^T \bm{\alpha}^{\star})} {\|\bm{\alpha}^{\star}\|_2^2}.
%\end{align*}
%\end{footnotesize}

\myparatight{LASSO} 
We approximate the gradient as follows:
% \begin{small}
% \begin{align*}
% \hat{\lambda}_{LASSO}(\mathbf{w}^{\star})  = \frac{2 \text{sign}(\mathbf{w}^{\star})^T (\mathbf{Xy} - \mathbf{X} \mathbf{X}^T \mathbf{w}^{\star})}{\text{sign}(\mathbf{w}^{\star})^T \text{sign}(\mathbf{w}^{\star})}.
% \end{align*}
% \end{small} 
% Therefore,
\begin{footnotesize}
\begin{align*}
& \nabla \hat{\lambda}_{LASSO}(\mathbf{w}^{\star}) = \lim_{\Delta \mathbf{w} \to \mathbf{0}} \frac{ \hat{\lambda}_{LASSO}(\mathbf{w}^{\star} + \Delta \mathbf{w}) - \hat{\lambda}_{LASSO}(\mathbf{w}^{\star})}{\Delta \mathbf{w}} \\
& = \lim_{\Delta \mathbf{w} \to \mathbf{0}}  \frac{1}{\Delta \mathbf{w}}  \Big( \frac{2 \text{sign}(\mathbf{w}^{\star} + \Delta \mathbf{w})^T (\mathbf{Xy} - \mathbf{X} \mathbf{X}^T (\mathbf{w}^{\star} + \Delta \mathbf{w}))}{\text{sign}(\mathbf{w}^{\star} + \Delta \mathbf{w})^T \text{sign}(\mathbf{w}^{\star} + \Delta \mathbf{w})} \\
&	\qquad - \frac{2 \text{sign}(\mathbf{w}^{\star})^T (\mathbf{Xy} - \mathbf{X} \mathbf{X}^T \mathbf{w}^{\star})}{\text{sign}(\mathbf{w}^{\star})^T \text{sign}(\mathbf{w}^{\star})} \Big) \\
& \approx \lim_{\Delta \mathbf{w} \to \mathbf{0}} \frac{1}{\Delta \mathbf{w}} \frac{\text{sign}(\mathbf{w}^{\star})^T \mathbf{X} \mathbf{X}^T \Delta \mathbf{w}}{\|\text{sign}(\mathbf{w}^{\star})\|_2^2}  \\
& \approx \lim_{\Delta \mathbf{w} \to \mathbf{0}} \frac{1}{\Delta \mathbf{w}} \frac{\Delta \mathbf{w}^T \mathbf{X} \mathbf{X}^T \text{sign}(\mathbf{w}^{\star})}{\|\text{sign}(\mathbf{w}^{\star})\|_2^2} \approx \frac{\mathbf{X} \mathbf{X}^T \text{sign}(\mathbf{w}^{\star})}{\|\text{sign}(\mathbf{w}^{\star})\|_2^2},
\end{align*}
\end{footnotesize}
where in the third equation, we use $\text{sign}(\mathbf{w}^{\star} + \Delta \mathbf{w}) \approx \text{sign}(\mathbf{w}^{\star})$. 
%Note that $\text{sign}(\mathbf{w}^{\star})$ is not differentiable at the dimensions where ${w}^{\star}_i=0$. 
%However, we do not use model parameters that are 0 to estaimate the hyperparameter. 
%Therefore, $\hat{\lambda}_{LASSO}$ is always differentiable at $\mathbf{w}^{\star}$.

%Therefore $ \nabla \hat{\lambda}_{LASSO}(\mathbf{w}^{\star}) = O(\frac{1}{\|\text{sign}(\mathbf{w}^{\star})\|_2^2})$.

\myparatight{L2-LR} 
We approximate the gradient as follows:
% \begin{small}
% \begin{align*}
% \hat{\lambda}_{L2-LR} = \frac{(\mathbf{w}^{\star})^T \mathbf{X} (\mathbf{y} - \mathbf{h}_{\mathbf{w}^{\star}}(\mathbf{X}))}{(\mathbf{w}^{\star})^T \mathbf{w}^{\star}}. 
% \end{align*}
% \end{small}
% Therefore,
%The variant of $\hat{\lambda}_{L2-LR}$ with respect to the variant of $\mathbf{w}^{\star}$ is 
\begin{footnotesize}
\begin{align*}
& \nabla \hat{\lambda}_{L2-LR}(\mathbf{w}^{\star}) = \lim_{\Delta \mathbf{w} \to \mathbf{0}} \frac{\hat{\lambda}_{L2-LR}(\mathbf{w}^{\star} + \Delta \mathbf{w}) - \hat{\lambda}_{L2-LR}(\mathbf{w}^{\star})}{\Delta \mathbf{w}} \\ 
& = \lim_{\Delta \mathbf{w} \to \mathbf{0}}  \frac{1}{\Delta \mathbf{w}}  \Big( \frac{(\mathbf{w}^{\star}+ \Delta \mathbf{w})^T (\mathbf{y} - \mathbf{h}_{\mathbf{w}^{\star}+\Delta \mathbf{w}}(\mathbf{X}))}{(\mathbf{w}^{\star} + \Delta \mathbf{w})^T (\mathbf{w}^{\star} + \Delta \mathbf{w})} \\
& \qquad -  \frac{(\mathbf{w}^{\star})^T \mathbf{X} (\mathbf{y} - \mathbf{h}_{\mathbf{w}^{\star}}(\mathbf{X}))}{(\mathbf{w}^{\star})^T \mathbf{w}^{\star}} \Big) \\
& \approx \lim_{\Delta \mathbf{w} \to \mathbf{0}}  
\frac{1}{\Delta \mathbf{w}} \frac{\Delta (\mathbf{w}^{\star})^T \mathbf{X}(\mathbf{y} - \mathbf{h}_{\mathbf{w}^{\star} + \Delta \mathbf{w}}(\mathbf{X}))}{(\mathbf{w}^{\star})^T \mathbf{w}^{\star}} \\
& \qquad - \frac{(\mathbf{w}^{\star})^T \mathbf{X} (\mathbf{h}_{\mathbf{w}^{\star} + \Delta \mathbf{w}}(\mathbf{X}) - \mathbf{h}_{\mathbf{w}^{\star}}(\mathbf{X}))}{(\mathbf{w}^{\star})^T \mathbf{w}^{\star}}   \\
%\frac{1}{\Delta \mathbf{w}} \frac{\Delta (\mathbf{w}^{\star})^T \mathbf{X}(\mathbf{y} - \mathbf{h}_{\mathbf{w}^{\star} + \Delta \mathbf{w}}(\mathbf{X})) - (\mathbf{w}^{\star})^T \mathbf{X} (\mathbf{h}_{\mathbf{w}^{\star} + \Delta \mathbf{w}}(\mathbf{X}) - \mathbf{h}_{\mathbf{w}^{\star}}(\mathbf{X}))}{(\mathbf{w}^{\star})^T \mathbf{w}^{\star}}   \\
& \approx \frac{\mathbf{X}(\mathbf{y} - \mathbf{h}_{\mathbf{w}^{\star}}(\mathbf{X})) }{\|\mathbf{w}^{\star}\|_2^2},
\end{align*}
\end{footnotesize}
where in the third and fourth equations, we use $(\mathbf{w}^{\star} + \Delta \mathbf{w})^T (\mathbf{w}^{\star} + \Delta \mathbf{w}) \approx (\mathbf{w}^{\star})^T \mathbf{w}^{\star}$ and $\mathbf{h}_{\mathbf{w}^{\star} + \Delta \mathbf{w}}(\mathbf{X}) \approx \mathbf{h}_{\mathbf{w}^{\star}}(\mathbf{X})$ for sufficiently small $\Delta \mathbf{w}$, respectively.

%Therefore $ \nabla \hat{\lambda}_{L2-LR}(\mathbf{w}^{\star}) = O(\frac{1}{\|\mathbf{w}^{\star}\|_2^2})$.
\myparatight{L2-KLR} 
Similar to L2-LR, we  approximate the gradient as:
\begin{footnotesize}
\begin{align*}
\nabla \hat{\lambda}_{L2-KLR}(\bm{\alpha}^{\star})  \approx \frac{\mathbf{K}(\mathbf{y} - \mathbf{h}_{\bm{\alpha}^{\star}}(\mathbf{K})) }{\|\bm{\alpha}^{\star}\|_2^2}.
\end{align*}
\end{footnotesize}

\myparatight{L1-LR} 
We approximate the gradient as follows:
% \begin{small}
% \begin{align*}
% \hat{\lambda}_{L1-LR}(\mathbf{w}^{\star}) = \frac{\text{sign}(\mathbf{w}^{\star})^T \mathbf{X} (\mathbf{y} - \mathbf{h}_{\mathbf{w}^{\star}}(\mathbf{X}))}{\text{sign}(\mathbf{w}^{\star})^T \text{sign}(\mathbf{w}^{\star})}. 
% \end{align*}
% \end{small} 
% Therefore,
%The variant of $\hat{\lambda}_{L1-LR}$ with respect to the variant of $\mathbf{w}^{\star}$ is
\begin{footnotesize}
\begin{align*}
& \nabla \hat{\lambda}_{L1-LR}(\mathbf{w}^{\star})  = \lim_{\Delta \mathbf{w} \to \mathbf{0}} \frac{\hat{\lambda}_{L1-LR}(\mathbf{w}^{\star} + \Delta \mathbf{w}) - \hat{\lambda}_{L1-LR}(\mathbf{w}^{\star})}{\Delta \mathbf{w}} \\ 
& = \lim_{\Delta \mathbf{w} \to \mathbf{0}}  \frac{1}{\Delta \mathbf{w}} \Big( \frac{\text{sign}(\mathbf{w}^{\star} + \Delta \mathbf{w})^T \mathbf{X} (\mathbf{y} - \mathbf{h}_{\mathbf{w}^{\star} + \Delta \mathbf{w}}(\mathbf{X}))}{\text{sign}(\mathbf{w}^{\star} + \Delta \mathbf{w})^T \text{sign}(\mathbf{w}^{\star} + \Delta \mathbf{w})} - \\
& \qquad \frac{\text{sign}(\mathbf{w}^{\star})^T \mathbf{X} (\mathbf{y} - \mathbf{h}_{\mathbf{w}^{\star}}(\mathbf{X}))}{\text{sign}(\mathbf{w}^{\star})^T \text{sign}(\mathbf{w}^{\star})} \Big) \\
& \approx \lim_{\Delta \mathbf{w} \to \mathbf{0}} \frac{(\mathbf{h}_{\mathbf{w}^{\star}+\Delta \mathbf{w}}(\mathbf{X}) - \mathbf{h}_{\mathbf{w}^{\star}}(\mathbf{X}))^T }{\Delta \mathbf{w}} \frac{\mathbf{X}^T \text{sign}(\mathbf{w}^{\star})}{{\| \text{sign}(\mathbf{w}^{\star})\|_2^2}} \\
& \approx \frac{\nabla \mathbf{h}_{\mathbf{w}^{\star}}(\mathbf{X})) \mathbf{X}^T \text{sign}(\mathbf{w}^{\star})^T}{{\| \text{sign}(\mathbf{w}^{\star})\|_2^2}}. % \\
%& = \frac{\text{sign}(\mathbf{w}^{\star})^T}{{\| \text{sign}(\mathbf{w}^{\star})\|_2^2}} \mathbf{X} \mathbf{X}^T \mathbf{h}_{\mathbf{w}^{\star}}(\mathbf{X}) (1- \mathbf{h}_{\mathbf{w}^{\star}}(\mathbf{X})), 
\end{align*}
\end{footnotesize}
%where we use $\text{sign}(\mathbf{w}^{\star} + \Delta \mathbf{w}) \approx \text{sign}(\mathbf{w}^{\star})$.
%Therefore $ \nabla \hat{\lambda}_{L1-LR}(\mathbf{w}^{\star}) = O(\frac{1}{\|\text{sign}(\mathbf{w}^{\star})\|_2^2})$.

\myparatight{L1-KLR} 
Similar to L1-LR, we approximate the gradient as:
\begin{footnotesize}
\begin{align*}
\nabla \hat{\lambda}_{L1-KLR}(\bm{\alpha}^{\star}) \approx \frac{\nabla \mathbf{h}_{\bm{\alpha}^{\star}}(\mathbf{K})) \mathbf{K}^T \text{sign}(\bm{\alpha}^{\star})^T}{{\| \text{sign}(\bm{\alpha}^{\star})\|_2^2}}.
\end{align*}
\end{footnotesize}

\end{document}